\documentclass[twoside,11pt]{article}

\usepackage{blindtext}

\usepackage{jmlr2e}
\RequirePackage{amsmath}
\RequirePackage{amssymb}
\RequirePackage{bm} 
\RequirePackage{url}
\usepackage{multirow}
\usepackage{natbib}

\usepackage{graphicx}
\usepackage{subfigure}
\usepackage{makecell}
\usepackage{booktabs}
\usepackage{array}
\usepackage{url}
\usepackage{algorithm}
\usepackage{algorithmic}
\usepackage{dsfont}
\usepackage{multirow}
\usepackage{setspace}
\usepackage{xcolor}

\newcommand{\bb}{\bm{b}}

\newcommand{\xb}{\mathbf{x}}
\newcommand{\yb}{\mathbf{y}}

\newcommand{\Ab}{\mathbf{A}}
\newcommand{\Bb}{\mathbf{B}}
\newcommand{\Cb}{\mathbf{C}}
\newcommand{\Db}{\mathbf{D}}
\newcommand{\Eb}{\mathbf{E}}
\newcommand{\Hb}{\mathbf{H}}
\newcommand{\Mb}{\mathbf{M}}
\newcommand{\Ob}{\mathbf{O}}
\newcommand{\Ub}{\mathbf{U}}
\newcommand{\Vb}{\mathbf{V}}
\newcommand{\Wb}{\mathbf{W}}
\newcommand{\Xb}{\mathbf{X}}
\newcommand{\Yb}{\mathbf{Y}}
\newcommand{\Zb}{\mathbf{Z}}

\newcommand{\RR}{\mathbb{R}}

\newcommand{\bX}{\bm{X}}

\newcommand{\bgamma}{\mathbf{\gamma}}
\newcommand{\bGamma}{\mathbf{\Gamma}}
\newcommand{\Gb}{\mathbf{G}}
\newcommand{\bSigma}{\mathbf{\Sigma}}
\newcommand{\Rb}{\mathbf{R}}
\newcommand{\Deltab}{\mathbf{\Delta}}
\newcommand{\bPhi}{\mathbf{\Phi}}
\newcommand{\bPsi}{\mathbf{\Psi}}
\newcommand{\bXi}{\mathbf{\Xi}}

\newcommand{\cL}{\mathcal{L}}

\newcommand{\argmin}{\mathop{\mathrm{argmin}}}

\newcommand{\diag}{{\rm diag}}

\newcommand{\mat}{\mathrm{mat}}

\newcommand{\bDelta}{\boldsymbol{\Delta}}
\newcommand{\bTheta}{\boldsymbol{\Theta}}

\newcommand{\bPi}{\boldsymbol{\Pi}}

\newcommand{\tr}{\mathop{\mathrm{tr}}}
\newcommand{\Tr}{\mathop{\mathrm{tr}}}

\newcommand{\transpose}{\scriptscriptstyle \sf T}
\def\trans{^{\transpose}}

\renewcommand{\vec}{\mathrm{vec}}

\ifx\theorem\undefined
\newtheorem{theorem}{Theorem}
\fi
\ifx\example\undefined
\newtheorem{example}[theorem]{Example}
\fi
\ifx\property\undefined
\newtheorem{property}{Property}
\fi
\ifx\lemma\undefined
\newtheorem{lemma}[theorem]{Lemma}
\fi
\ifx\proposition\undefined
\newtheorem{proposition}[theorem]{Proposition}
\fi
\ifx\remark\undefined
\newtheorem{remark}[theorem]{Remark}
\fi
\ifx\corollary\undefined
\newtheorem{corollary}[theorem]{Corollary}
\fi
\ifx\definition\undefined
\newtheorem{definition}[theorem]{Definition}
\fi
\ifx\conjecture\undefined
\newtheorem{conjecture}[theorem]{Conjecture}
\fi
\ifx\fact\undefined
\newtheorem{fact}[theorem]{Fact}
\fi
\ifx\claim\undefined
\newtheorem{claim}[theorem]{Claim}
\fi
\ifx\assumption\undefined
\newtheorem{assumption}[theorem]{Assumption}
\fi
\ifx\cond\undefined
\newtheorem{cond}[theorem]{Condition}
\fi
\ifx\assump\undefined
\newtheorem{assump}[theorem]{Assumption}
\fi

\usepackage{lastpage}
\jmlrheading{00}{0000}{1-\pageref{LastPage}}{0/00; Revised 0/00}{0/00}{00-0000}{Zebin Wang, Ziming Gan, Weijing Tang, Zongqi Xia, Tianrun Cai, Tianxi Cai and Junwei Lu}

\ShortHeadings{DANIEL: A Distributed and Scalable Representation Learning Approach}{Wang, Gan, Tang, Xia, Cai, Cai and Lu}
\firstpageno{1}

\begin{document}

\title{DANIEL: A Distributed and Scalable Approach for Global Representation Learning with EHR Applications}

\author{\name Zebin Wang$^{\bm{1}}$ \email zebinwang@g.harvard.edu \\
        \name Ziming Gan$^{\bm{2}}$ \email zimingg@uchicago.edu \\
        \name Weijing Tang$^{\bm{3}}$ \email weijingt@andrew.cmu.edu \\
        \name Zongqi Xia$^{\bm{4}}$ \email zxia1@pitt.edu \\
        \name Tianrun Cai$^{\bm{5}}$ \email tcai1@bwh.harvard.edu \\
        \name Tianxi Cai$^{\bm{1}*}$ \email tcai@hsph.harvard.edu \\
        \name Junwei Lu$^{\bm{1}*}$ \email junweilu@hsph.harvard.edu \\
        \addr $^{{1}}$Department of Biostatistics, Harvard University, Boston, MA 02115, USA \\
        \addr $^{{2}}$Department of Statistics, University of Chicago, Chicago, IL 60637, USA \\
        \addr $^{{3}}$Department of Statistics \& Data Science, Carnegie Mellon University, Pittsburgh, PA 15213, USA\\
        \addr $^{{4}}$Department Neurology, University of Pittsburgh, Pittsburgh, PA 15213, USA  \\
        \addr $^{{5}}$Department of Medicine, Brigham and Women's Hospital, Boston, MA 02115, USA \\
        \addr $^{*}$Equal Contribution in Corresponding Authorship
       }

\editor{DANIEL Editors}

\maketitle

\begin{abstract}

Classical probabilistic graphical models face fundamental challenges in modern data environments, which are characterized by high dimensionality, source heterogeneity, and stringent data-sharing constraints. In this work, we revisit the Ising model, a well-established member of the Markov Random Field (MRF) family, and develop a distributed framework that enables scalable and privacy-preserving representation learning from large-scale binary data with inherent low-rank structure. Our approach optimizes a non-convex surrogate loss function via bi-factored gradient descent, offering substantial computational and communication advantages over conventional convex approaches. We evaluate our algorithm on multi-institutional electronic health record (EHR) datasets from 58,248 patients across the University of Pittsburgh Medical Center (UPMC) and Mass General Brigham (MGB), demonstrating superior performance in global representation learning and downstream clinical tasks, including relationship detection, patient phenotyping, and patient clustering. These results highlight a broader potential for statistical inference in federated, high-dimensional settings while addressing the practical challenges of data complexity and multi-institutional integration.
\end{abstract}


 \definecolor{purple}{RGB}{250,000,180}
 \def\purple{\color{purple}}

 \newcommand{\crossout}[1]{{\brown \textst{#1}}}
 \newcommand{\tcomm}[1]{{\tiny\bf\purple \fbox{#1}}}

\section{Introduction}

Markov Random Fields (MRFs) form a fundamental class of probabilistic graphical models used to study the conditional dependencies between high-dimensional variables \citep{Wainwright2008MRF, koller2009probabilistic, Sucar2015PGM}. In particular, the Ising model, as a computationally efficient and theoretically supported variant of MRFs, has been widely adopted in modeling the pairwise interactions between \textit{binary} variables in domains where features naturally exhibit dichotomous patterns \citep{torquato2011toward,cheng2014sparse,marinazzo2014information, zhang2018unsupervised}. Compared to alternative graphical models such as Bayesian networks \citep{zhao2011combining,bandyopadhyay2015data,chen2024likelihood}, the Ising model provides a more interpretable pairwise interaction framework with more scalable computation when inferring network structures \citep{pourret2008bayesian}.

Numerous estimation and inference methods have been developed for the Ising model, including the variational methods \citep{chow1968approximating,nguyen2012mean}, graph recovery methods based on the mutual information \citep{montanari2009graphical,bresler2015efficiently,kandiros2023learning}, and $M$-estimators based on the convex regularized conditional likelihood loss \citep{shah2021learning,zhu2023simple}. However, these studies often overlook the practical challenges associated with consolidating cross-institutional data under strict confidentiality regulations. For example, deposition transcripts typically contain sensitive information stored separately across multiple law firms, which may be temporary allies in specific legal contexts but otherwise operate independently or even competitively \citep{Perlin2024Law}. Despite the inter-institutional boundaries, collaborative analysis of deposition data is often necessary to support coordinated legal decision-making and case evaluation \citep{zhang2023fedlegal}. On the other hand, recent federated statistical methods introduce surrogate loss optimization to enable distributed learning while alleviating privacy concerns \citep{jordan2019dist,duan2022heterogeneity, lin2022fednlp, fan2023communication, Nagy2023FedNLP}. However, these methods either lack theoretical performance guarantees or rely on conventional optimization frameworks involving large-scale matrix computations, which pose substantial scalability challenges in high-dimensional settings. Consequently, the study of MRFs and the Ising model under modern settings, characterized by large-scale multi-source data and stringent privacy constraints, remains largely underexplored.

One of the most prominent application domains for Ising model-based methodologies designed to capture dependencies among large-scale binary variables is the derivation of low-dimensional embeddings that preserve meaningful clinical semantics from Electronic Health Records (EHRs) {\citep{wang2017knowledge,lu2019learning,murali2023towards,carvalho2023knowledge}.} The EHR systems capture rich longitudinal clinical data, including diagnoses, treatments, medical procedures, and laboratory results \citep{Forrey1996LOINC,choi2017using, li2019distributed, Theodorou2023EHR}. Particularly, EHR datasets are commonly represented as binary indicators of clinical features, such as the dichotomous indicator of specific diagnoses and treatments within designated time windows \citep{Li2022Binary, Piao2024BernGraph}. Leveraging their large sample size and comprehensive feature coverage, EHRs have become a cornerstone for clinical knowledge discovery and translational research \citep{asaria2016using, bean2017knowledge,munkhdalai2018clinical, ahuja2020phenotyping, christopoulou2020adverse, liu2021phenotyping}. In particular, the knowledge graph (KG) embedding learning algorithms offer efficient querying and automated reasoning capabilities \citep{zhou2022multiview}, shedding light on the conditional dependencies among medical entities {\citep{liu2015automatic,salem2016fixing, kaewprag2017predictive,shen2018cbn, yu2022bios, zhao2023comorbidity, Gan2025ARCH}.} These embeddings can further support a broad range of downstream tasks, including phenotype classification, disease sub-typing, and predictive modeling \citep{kirby2016phekb, glicksberg2018automated,de2021phe2vec, Zhou2025Federated}, which reflect the core objectives of clinical research aimed at enhancing disease understanding and advancing precision medicine \citep{ carvalho2023knowledge, Liu2024Ontology}.

However, while the Ising model is well-suited for learning embeddings of the EHR data with inherent sparsity and low-rank structure, its conventional implementations, such as those studied by \citet{nussbaum2019ising}, face two fundamental challenges:

(1) \textit{Computation scalability for large-scale EHR data.} First, the volume of EHR database grows rapidly in complexity and computational demands in the recent years \citep{Dash2019BigEHR, Gan2025ARCH}. This makes conventional optimization approaches for statistical graphical models with sparsity and low-rankness, based on convex regularized likelihood loss and require singular value decomposition (SVD) at each iteration \citep{chandrasekaran2010latent, ravikumar2010high, liu2012implementable, meng2014learning, shah2021learning}, increasingly impractical for large-scale applications. Second, the heterogeneity across multiple EHR sources complicates cross-institutional data unification, as format reconciliation of EHR data is both time-consuming and prone to errors \citep{Glynn2019Hetero, Fu2020Hetero, Zhou2025Federated}.

(2) \textit{Stringent healthcare data privacy regulations.} First, the EHR datasets, which contain sensitive personal health information, are subject to strict federal and state privacy laws, where violations may result in serious legal and reputational consequences \citep{Annas2003HIPAA, Mello2018Legal, Apathy2020Policies}. Second, building centralized, professionally managed EHR databases across healthcare institutions that fully comply with legal and administrative requirements is costly and often infeasible \citep{Williams2008HIPAA, Mia2022HIPAA}. The alternative embedding learning approaches based on graph learning and natural language processing, such as Graph Neural Networks (GNNs) \citep{Scarselli2009GNN} and Bidirectional Encoder Representations from Transformers (BERT) \citep{devlin2018bert}, will typically rely on supervised domain-specific fine-tuning which requires centralized access to large volumes of labeled EHR data that are subject to HIPAA compliance \citep{lee2020biobert, pubmedbert}, or show their vulnarable to adversarial model extraction \citep{Wu2022Attack, Zhao2025survey, Wang2025CEGA}.

To overcome the challenges mentioned above, we propose a distributed approach for learning embeddings for large-scale, multi-institutional binary data subject to privacy regulations, termed \underline{D}istributed \underline{A}rchitecture for \underline{N}on-Convex \underline{I}sing-based \underline{E}mbedding \underline{L}earning (DANIEL). Our framework optimizes a distributed Ising model with a low-rank parameter matrix that captures the semantic embeddings for the features. Without loss of generality, we focus on global EHR embedding by generating visualized KGs and evaluating numerous downstream clinical tasks. To improve computation scalability, DANIEL performs optimization on a non-convex bi-factored surrogate loss that fully leverages the low-rank structure of embeddings intrinsic to EHRs, substantially reducing the dimensionality of matrices updated at each iteration. To address data-sharing limitations, DANIEL adopts a distributed structure where information exchange across institutions is limited to a \textit{single round of communication} involving locally computed \textit{gradients of loss}, without requiring direct sharing of sensitive data. Extensive simulation studies and empirical evaluations on real-world EHR datasets demonstrated the superiority of DANIEL across core downstream tasks of global EHR embedding, including relationship pair detection, patient phenotyping, and patient clustering, under realistic data-sharing constraints, offering substantial practical advantages over existing alternatives. 

\subsection{Our Contributions}

In summary, the contributions of this paper are four-fold:

\begin{itemize}
    \item [$\bullet$] \textbf{Modern Research Question Formulation. } We introduce the modern setting for learning dependency structures via MRFs, specifically the Ising model, in scenarios where high-dimensional binary data originates from multiple sources governed by strict privacy constraints. This formulation addresses the pressing needs in real-world applications such as collaborative legal research and healthcare analytics, where traditional assumptions of centralized data access and moderate dimensionality are no longer applicable.

    \item [$\bullet$] \textbf{Scalable and Privacy-Aware Methodology Design.} We propose DANIEL, an unsupervised framework designed to learn global embeddings from large-scale binary datasets with inherent low-rank structure. DANIEL achieves computational scalability by optimizing a non-convex bi-factored surrogate loss based on the Ising model subject to regularization for embedding balance. It ensures privacy and communication efficiency by requiring only one-shot communication of the gradient statistics, without sharing sensitive raw data across institutions.

    \item [$\bullet$] \textbf{Simulation-Supported Theoretical Guarantees.} We provide theoretical results for the convergence behavior and statistical validity of DANIEL in distributed settings. Specifically, we show that DANIEL \textit{retains the classical rate} held by convex low-rank Ising estimators under centralized settings, and the rate between the DANIEL estimator and the oracle centralized estimator is \textit{of lower order} than their overall estimation error. To the best of our knowledge, this is the first theoretical performance guarantee for a non-convex distributed Ising estimator, with potential generalizability to other loss functions beyond the current model.

    \item [$\bullet$] \textbf{Extensive Evaluation in EHR-Based Patient Representation Learning.} We evaluate DANIEL on patient representation learning using real-world EHR datasets from two leading healthcare systems, the University of Pittsburgh Medical Center (UPMC) and Mass General Brigham (MGB). DANIEL shows superiority over established MRF-based and BERT-style alternatives by consistently outperforming them across key downstream tasks in healthcare research, including relationship detection, patient phenotyping, and patient clustering.
    
\end{itemize}

\subsection{Related Literature} \label{sec:RL}

\textbf{Low-Rank Parameter Matrix Estimation Methods.} \citep{zheng2015highdim} leads the convex relaxation methods with a nuclear norm penalty, with subsequent developments by \citep{jianqing2017nuclear, park2022lowrank}. Despite their theoretical appeal, these approaches are computationally inefficient as they require iterative implementation of the singular value decomposition (SVD) during optimization. Additionally, they are designed for centralized computing environments and are therefore not suitable for application in distributed settings.
Non-convex optimization methods have been considered for other statistical problems \citep{zhang2016provable,sun2016guaranteed, chen2018harnessing, li2019non,li2019rapid, chi2019nonconvex, danilova2022recent}. However, these approaches also rely on centralized data access, which raises significant concerns regarding data security, for example, the risk of extrajudicial surveillance during centralized data aggregation \citep{shokri2015privacy}.

\noindent \textbf{Distributed Learning.} \citep{gao2022review} provides a systematic review of the distributed statistical inference methods, which highlights that a vast amount of research considered the surrogate likelihood approach, which only aggregates gradient or Hessian information from local sites \citep{wang2017efficient,jordan2019dist,duan2022heterogeneity,fan2023communication}. However, these approaches typically depend on the convexity of the surrogate loss. Thus, these approaches cannot be directly extended to the non-convex settings we implement in our work for the purpose of prioritizing efficiency and scalability. Several distributed algorithms have been proposed for non-convex problems under the general optimization framework \citep{tatarenko2017non,sun2020improving,xin2021fast,gorbunov2021marina,zhao2022beer}. However, these methods generally require multiple rounds of communication and do not offer theoretical guarantees for the estimation error. In contrast, our proposed DANIEL algorithm overcomes both challenges by conducting estimation under a non-convex surrogate loss with one-shot communication, offering a rigorous statistical error rate for the obtained estimator.

\section{Preliminaries} \label{sec:intro_model}

\noindent\textbf{Notations.} Suppose that there are $m$ institutions and a total of $n$ independent samples. For each sample, we observe a binary feature vector $\xb=(x_1, \cdots, x_p) \in \mathbb{R}^p$, representing the presence or absence of $p$ features in real-world practice. We model the probability density of the observed samples as $p(\xb) \propto \exp\big\{ \sum_{j \neq k} \theta_{jk} x_{j} x_{k} + \sum_{j=1}^p \theta_{jj}x_j \big\}$, where the parameter matrix $\bTheta = \{\theta_{jk}\} \in \mathbb{R} ^ {p \times p}$ encodes the dependency structure among features, with $\theta_{jk} = 0$ implying conditional independence between feature $j$ and feature $k$ given the rest. We use $\|\cdot\|_{\mathrm{op}}$, $\|\cdot\|_{\mathrm{F}}$, and $\|\cdot\|_{*}$ to denote the operator norm, the Frobenius norm, and the nuclear norm of a matrix, respectively. We use $\langle \cdot, \cdot \rangle$ to denote the Frobenius inner product of matrices, where $\langle \Ab,\Ab \rangle = \|\Ab\|_{\mathrm{F}} ^ 2$. We use $\|\cdot\|_{\max}$ to denote the entrywise max norm. We use $\vec(\cdot)$ to denote the vectorization of a matrix and $\tr (\cdot)$ to denote the trace of a matrix. We use $\|\cdot\|_{p,q}$ to denote the $L_{p,q}$ norm of a matrix $\Ab \in \mathbb{R} ^ {d_1 \times d_2}$.

\noindent\textbf{Multi-Institutional Data Setup.} We consider a multi-institutional setup where data is collected across $m$ distinct institutions, denoted as $\mathcal{S}_{1}, \ldots, \mathcal{S}_{m}$. For simplicity, we assume that $n$ samples are evenly stored in each distinct institution, i.e., each institution $\mathcal{S}_i$ for $i \in 1,...,m$ holds $n/m$ independent samples. For sample $\ell$ from the $i$th institution, the sample information is represented by a $p$-dimensional vector $\xb^{(\ell)} = \big(x_1^{(\ell)}, x_2^{(\ell)}, ..., x_p^{(\ell)} \big)\trans$, where $x_j^{(\ell)} \in \{ \pm 1\}$ indicates the presence or absence of the $j$th feature in sample $\ell$. We assume that the dependency structure among features is shared across all institutions and follows the Ising model with a unified \textit{true} parameter matrix $\bTheta^* = \big\{ \theta^*_{jk} \big\} \in \mathbb{R} ^ {p \times p}$. The probability density is given by
\begin{equation} \label{equ:def_X}
    p_{\bTheta^*}(\xb^{(\ell)}) = \mathbb{P}_{\bTheta^*} \big(X_j = x_j^{(\ell)} \text{ for all } j \in [p] \big) \propto \exp\bigg\{ \sum_{j \neq k} \theta_{jk}^* x_{j}^{(\ell)} x_{k}^{(\ell)} + \sum_{j=1}^p \theta_{jj}^* x_j^{(\ell)} \bigg\},
\end{equation}
where $\bTheta^*$ is \textit{symmetric}, i.e., $\theta_{jk}^*=\theta_{kj}^*$. Heuristically, the off-diagonal entries of $\bTheta^*$, i.e.,  $\theta_{jk}^*$ with $j\neq k$, characterize the \textit{pairwise dependencies} between the $j$th and $k$th features in condition on all other features. In contrast, the diagonal entries of $\bTheta^*$, i.e., $\theta_{jj}^*$ for $j \in [p]$, represent the \textit{individual tendencies} of the $j$th feature to occur, capturing the feature-specific heterogeneity in their occurrence patterns. We model the Ising model with inherent low-rank patterns by assuming that $\mathrm{rank}(\bTheta ^ *) = d \ll p$. Accordingly, we assume that $\bTheta ^ *$ satisfies a low-rank decomposition in the form of $\bTheta^* = \Ub^* {\Vb^*}\trans$, where $\Ub^*, \Vb^* \in \mathbb{R} ^ {p \times d}$ and $\Vb^*$ equals $\Ub^*$ subject to a sign per column, i.e., $\Vb^* = \Ub^* \Db^*$ for some diagonal matrix $\Db^*\in \RR^{d\times d}$ with $\Db^*_{jj} \in \{\pm 1\}$. Here, each row of $\Ub^*$ refers to a $d$-dimensional embedding for the corresponding feature, which can be obtained via rank-$d$ SVD of $\bTheta ^ *$. Inspired by existing work, such as  \citep{ravikumar2010high, vanborkulo2014elasso}, we estimate the parameters $\theta ^*_{jk}$ by considering the conditional distribution of $x_j^{(\ell)}$ given all the other feature indicators $\xb_{-j}^{(\ell)} = \big(x_1^{(\ell)}, ..., x_{j-1}^{(\ell)}, x_{j+1}^{(\ell)}, ..., x_p^{(\ell)} \big)\trans$, with a conditional likelihood given as
\begin{equation} \label{equ:conditional_dist}
    \mathbb{P}_{\bTheta^*}(X_j = x_j^{(\ell)} | \bX_{-j} = \xb_{-j}^{(\ell)}) = \frac{\exp{\big(2\theta_{jj}^* x_j^{(\ell)}+ 2\sum_{k \neq j} \theta_{kj}^* x_j^{(\ell)}x_k^{(\ell)} \big)}}{\exp{\big(2\theta_{jj}^* x_j^{(\ell)}+ 2\sum_{k \neq j} \theta_{kj}^* x_j^{(\ell)}x_k^{(\ell)}\big)} + 1}.
\end{equation}
Based on \eqref{equ:conditional_dist}, we have the conditional log-likelihood  loss for all the $n$ samples from $m$ institutions as
\begin{equation} \label{equ:objective}
\begin{split}
    \mathcal{L}(\boldsymbol{\Theta}) = - \frac{1}{n} \sum_{\tau=1}^m \sum_{\ell \in \mathcal{S}_\tau}  \sum_{j = 1} ^ p \bigg\{2\theta_{jj} x_j^{(\ell)} &+ 2\sum_{k \neq j} \theta_{kj} x_j^{(\ell)}x_k^{(\ell)} \\ & - \log \Big(\exp \big(2\theta_{jj} x_j^{(\ell)}+ 2\sum_{k \neq j} \theta_{kj} x_j^{(\ell)}x_k^{(\ell)} \big) + 1 \Big)\bigg\}.
\end{split}
\end{equation}
The low-rank structure of the parameter matrix can also be motivated by a latent mixed-value graphical model, which will be discussed in Section~\ref{sec:real_data}.

\section{Methodology}  \label{sec:NCDnC}

In this section, we propose DANIEL, a distributed algorithm designed to estimate $\bTheta ^ *$ via a non-convex approach. Compared to the Communication-efficient Surrogate Likelihood (CSL) framework \citep{jordan2019dist}, which operates under a convex loss, DANIEL leverages a non-convex bi-factored surrogate loss to improve efficiency by exploiting the inherent low-rank structure of the dependency structure of interest. In specific, suppose that the $n$ samples are partitioned across $m$ institutions, denoted by $\mathcal{S}_{1}, \ldots, \mathcal{S}_{m}$, with each institution storing $n/m$ samples. We set $\mathcal{S}_1$ as the \textit{hub institution}, which is responsible for managing communication and aggregating intermediate results. We denote the log-likelihood on the $i$th institution as
\begin{equation*} \label{equ:objective_distributed}
\begin{split}
    \mathcal{L}_i(\boldsymbol{\Theta})  = - \frac{1}{n/m} \sum_{\ell \in \mathcal{S}_i} \sum_{j = 1} ^ p \bigg\{ 2\theta_{jj} x_j^{(\ell)}&+ 2\sum_{k \neq j} \theta_{kj} x_j^{(\ell)}x_k^{(\ell)} \\
    & - \log \Big(\exp \big(2\theta_{jj} x_j^{(\ell)}+ 2\sum_{k \neq j} \theta_{kj} x_j^{(\ell)}x_k^{(\ell)} \big) + 1 \Big) \bigg\}.
\end{split}
\end{equation*}
We propose the bi-factored surrogate loss 
\begin{equation} \label{equ:surrogate_obj}
    \widetilde{\mathcal{L}} (\Ub, \Vb; \widehat\bTheta_0) = \underbrace{\mathcal{L}_1(\Ub\Vb \trans)}_{\mathrm{Loss\ on\ Hub}}  + \underbrace{\big\langle \nabla \mathcal{L} (\widehat\bTheta_0) - \nabla \mathcal{L}_1 (\widehat\bTheta_0), {\Ub} {\Vb} \trans  \big \rangle}_{\mathrm{Gradient\ Correction}}  + \underbrace{\frac{1}{4}\big\|\Ub\trans \Ub - \Vb \trans \Vb\big\|_{\mathrm{F}} ^ 2}_{\mathrm{Embedding\ Balancer}},
\end{equation}
where $\nabla \mathcal{L}_i(\cdot)$ represents the gradient of $\mathcal{L}_i(\cdot)$, and $\nabla \mathcal{L}(\cdot) = \frac{1}{m} \sum_{i = 1} ^ m \nabla \mathcal{L}_i(\cdot)$. In \eqref{equ:surrogate_obj}, the initialization $\widehat\bTheta_0$ is obtained on the hub institution $\mathcal{S}_1$ using only its local data, by solving the following convex optimization problem regularized by the nuclear norm 
\begin{equation} \label{equ:question}
    \widehat{\boldsymbol\Theta}_{\rm{cvx}} = \argmin_{\boldsymbol\Theta \in \mathcal{T} } \big\{ \mathcal{L}_1(\boldsymbol\Theta) + \lambda \| \bTheta \|_* \big\}, \text{ with } \mathcal{T} = \big\{\bTheta = \{\theta_{jk}\} \in \mathbb{R}^ {p \times p} \big| \| \bTheta \|_{\mathrm{F}} \leq B\big\},
\end{equation}
where $B$ is a fixed positive constant. Theorem \ref{thm:main} shows the theoretical rate of $\widehat{\boldsymbol\Theta}_{\rm{cvx}}$ and ensures that a full convergence in optimizing \eqref{equ:question} is \textit{sufficient but not necessary} to obtain a valid $\widehat\bTheta_0$, as the rank-$d$ SVD of $\widehat{\boldsymbol\Theta}_{\rm{cvx}}$, for subsequent optimization. The bi-factored surrogate loss, as indicated by \eqref{equ:surrogate_obj}, consists of three terms: (i) the loss on the hub institution that approximates the overall loss; (ii) a correction term that incorporates global gradient information via a one-shot communication; and (iii) a regularization term that balances $\Ub$ and $\Vb$ to ensure the identifiability of the bi-factor structure, as inspired by \citep{Park2016lowrank}. After collecting $\nabla \mathcal{L}_i(\widehat{\bTheta}_0)$ from all institutions, we minimize \eqref{equ:surrogate_obj} in the hub institution via gradient descent. The hyperparameter $1/4$ assigned to the regularization term is tuning-insensitive and potentially replaceable by any positive constant larger than 1/4. As indicated by the summary of DANIEL in Algorithm \ref{al:Divide_conquer}, our proposed algorithm requires only one round of communication between each institution and the hub institution. In DANIEL, the number of iteration steps $\Gamma$ is discussed in Theorem \ref{thm:Divide_conquer}, and the choice of step size $\eta$ is discussed in Appendix \ref{app:eta}. Finally, we highlight an important structural property of DANIEL in Remark \ref{rem:invariance}.

\begin{algorithm}[h]
\caption{The Proposed Framework of DANIEL}\label{al:Divide_conquer}
\begin{algorithmic}
\STATE \textbf{Input}: $\forall i \in \{1,\ldots,m\}$, samples $\{ \xb^{\ell} \}_{\ell \in S_i}$ stored in the $i$th institution; the dimension $d$, the step size $\eta$, and the number of iteration steps $\Gamma$.
\STATE \textbf{Initialization}: $\widehat{\bTheta}_0 = {\Ub}_0{\Vb}_0 \trans$, where $\widehat{\bTheta}_0$ is the rank-$d$ SVD of a solution to \eqref{equ:question}.
    \STATE $\mathcal{S}_1$ distributes $\widehat{\bTheta}_0 = {\Ub}_0{\Vb}_0 \trans$ to the institutions $\mathcal{S}_{2:m}$;
    \FOR{$i$ \textbf{from} 2 \textbf{to} $m$} 
        \STATE The $i$th institution $\mathcal{S}_{i}$ calculates the gradient $\nabla \mathcal{L}_i(\widehat{\bTheta}_{0})$;
        \STATE $\mathcal{S}_{i}$ returns the gradient $\nabla \mathcal{L}_i(\widehat{\bTheta}_{0})$ to $\mathcal{S}_1$;
    \ENDFOR
      \FOR{$\gamma$ \textbf{from} 1 \textbf{to} $\Gamma$ \textbf{in} $\mathcal{S}_1$,} 
        \STATE $\Ub_\gamma  = \Ub_{\gamma - 1} - \eta \nabla_{\Ub} \widetilde{\mathcal{L}} (\Ub, \Vb_{\gamma-1}  ; \widehat{\bTheta}_{0}) {\big|_{\Ub = \Ub _{\gamma -1}}}$;
        \STATE $\Vb_{\gamma}  = \Vb_{\gamma-1} - \eta \nabla_{\Vb} \widetilde{\mathcal{L}} (\Ub_{\gamma - 1} , \Vb; \widehat{\bTheta}_{0}) {\big|_{\Vb = \Vb_{\gamma -1} }}$;
\ENDFOR
\STATE \textbf{Output}: $\widehat{\Ub}  = \Ub_{\Gamma} $, $\widehat{\Vb} = \Vb_\Gamma$, and $\widehat{\bTheta} = \widehat{\Ub}\widehat{\Vb}\trans$.
\end{algorithmic}
\end{algorithm}
\vspace{-10pt}
\begin{remark} [Invariance of Symmetry]\label{rem:invariance}
   If the initialization $\widehat\bTheta_0$ is a symmetric matrix, we can initialize $\Ub_0=\Vb_0\Db_0$, where $\Db_0\in\mathbb{R}^{d\times d}$ is a diagonal matrix with its diagonal entries being $1$ or $-1$, i.e., $\Vb_0$ is the same as $\Ub_0$ up to a sign flip for each column. Then the output estimator of Algorithm~\ref{al:Divide_conquer},  $\widehat\bTheta$, will be a symmetric matrix. 
\end{remark}

\section{Theoretical Analysis}

In this section, we summarize the theoretical results on DANIEL's performance and validity. Specifically, we address three core aspects: Section \ref{sec:statistical_rate} establishes the statistical rate of the DANIEL estimator; Section \ref{sec:rational_initial} highlights the feasibility of the initialization procedure of DANIEL in order to achieve the desired statistical rate; Section \ref{sec:cd} validates the distributed design of DANIEL by proving that the distance between the distributed DANIEL estimator and its centralized counterpart is \textit{of lower order} than their statistical error to the truth.

\subsection{Statistical Rate of DANIEL} \label{sec:statistical_rate}

First, we show the statistical rate of the proposed DANIEL estimator $\widehat{\bTheta} = \widehat{\Ub}\widehat{\Vb}\trans$ and illustrate its efficiency.
We begin by stating the set of assumptions necessary for our theoretical analysis. We denote the \textit{per-sample} gradient as $\Wb_\ell (\bTheta)$, such that
 $\nabla \mathcal{L}(\bTheta)  = \frac{1}{n} \sum_{\ell = 1} ^ n \mathbf{W}_\ell(\bTheta)$. Additionally, we denote the empirical Hessian matrix of $\mathcal{L}(\boldsymbol\Theta)$ as $\widehat{\Hb}(\boldsymbol\Theta)$, with expectation $\Hb(\bTheta) = \mathbb{E}_{\bTheta}[\widehat{\Hb}(\cdot)]$.

Assumption \ref{ass:Positive_minimal_eigenvalue} represents the convexity and smoothness conditions of $\mathcal{L} (\bTheta)$ at the population level, which is widely used in the theoretical analysis of high-dimensional logistic regression \citep{ravikumar2010high}. In particular, Assumption \ref{ass:Positive_minimal_eigenvalue} ensures that the loss follows the \textit{restricted strongly convexity and smoothness} conditions, as proposed in the existing literature \citep{negahban2012RSC, jianqing2017nuclear, park2022lowrank}.

\begin{assump} [Restricted Strong Convexity \& Smoothness]\label{ass:Positive_minimal_eigenvalue} For any matrix $\bDelta \in \mathbb{R} ^ {p \times p}$, there exist constants $0 < \kappa_{\min} \leq \kappa_{\max} < \infty$ such that
\begin{equation} \label{equ:constraint_main}
    0 < \kappa_{\min} \|\bDelta\|^2_{\mathrm{F}} \leq \mathrm{vec}(\bDelta) \trans \mathbf{H} (\bTheta) \mathrm{vec}(\bDelta) \leq \kappa_{\max} \|\bDelta\|^2_{\mathrm{F}} \text{\ \  for all }\bTheta \in \mathcal{T},
\end{equation}
where the region $\mathcal{T}$ is defined in \eqref{equ:question}.
\end{assump}

Assumption \ref{ass:reg_grad} provides an upper bound on the scale of the {per-sample} gradient matrix $\Wb_\ell (\cdot)$, which is commonly employed in the study of he statistical rate of general $M$-estimators, such as \citep{negahban2012RSC}.

\begin{assump}[Regularization on the Gradient] \label{ass:reg_grad} There exists some positive constant $C$ such that $\Big\|\mathbb{E}_{\bTheta ^ *} \big[\mathbf{W}_\ell (\bTheta) \mathbf{W}_\ell\trans(\bTheta)\big] \Big\|_{\mathrm{op}} \leq C p$.
\end{assump}

We then present Proposition \ref{prop:rationality_assumption}, which provides a sufficient condition for Assumption \ref{ass:reg_grad}. Unlike in general multinomial low-rank logistic regression \citep{lei2019using}, directly verifying Assumption \ref{ass:reg_grad} by checking the eigenvalues of the covariance matrix for the binary covariates $\xb$ in our setting is less straightforward. In contrast, Proposition~\ref{prop:rationality_assumption} provides a sufficient condition that depends solely on the \textit{underlying graph structure} and the \textit{true parameter values}, thereby making the evaluation of the eigenvalue properties of the covariance matrix more transparent. The proof of Proposition \ref{prop:rationality_assumption} is provided in Appendix \ref{app:E}.  

\begin{proposition} \label{prop:rationality_assumption}
    Consider an Ising graphical model proposed in Section \ref{sec:intro_model} with the true parameter matrix $\bTheta ^ * = \big\{ \theta_{jk} ^ * \big\}$. If we have
    \begin{align*}
        (1) & \text{ For any arbitrary } (j,k) \in \mathcal{E}, 0 \leq \theta_{jk} ^ * \leq \theta_{\max{}} < \infty; \\
        (2) & \text{ The maximum degree for all the vertices } \tau = o(p); \\ 
        (3) &    ~\tau   \tanh{(\theta_{\max{}})}  \le 1/2;
    \end{align*}
    and there exists some $C > 0$ such that $\|\bTheta ^ *\|_{1, \infty} < C/p$, then $\Big\|\mathbb{E}_{\bTheta ^ *} \big[\mathbf{W}_\ell (\bTheta) \mathbf{W}_\ell\trans(\bTheta)\big] \Big\|_{\rm op}= O(p)$.
\end{proposition}

We now present the theoretical performance guarantee for DANIEL as implemented in Algorithm \ref{al:Divide_conquer}, with the proof provided in Appendix \ref{sec:proof-main}.

\begin{theorem} [Statistical Rate of DANIEL] \label{thm:Divide_conquer}
Under Assumptions \ref{ass:Positive_minimal_eigenvalue} and \ref{ass:reg_grad}, suppose that the sample size $n \geq Cd^2 p^4 \log p$, and the institution number $m \lesssim  \sqrt{n/(dp^7\log p) }$. Let the step size be $  \eta \leq 1/ \big(12\big\|\widehat\bTheta_0 \big\|^{1/2}_{\rm op} \max\{1,\kappa_{\min} + \kappa_{\max}\} \big)
$ and the number of iterations $\Gamma = \Omega \Big( \log \big(\frac{n}{d p \log p} \big) \Big)$, where the initialization of Algorithm \ref{al:Divide_conquer}, denoted by $\widehat\bTheta_0$, is obtained from $\widehat{\bTheta}_{\mathrm{cvx}}$, the solution to \eqref{equ:question}. We have

\begin{equation*}
    \inf_{\Ob \in \mathcal{O}(d)} \Big\{ \big\|\widehat{\Ub}  - \Ub^* \Ob \big\|_{\mathrm{F}} ^ 2 + \big\|\widehat{\Vb}  - \Vb^* \Ob \big\|_{\mathrm{F}} ^ 2 \Big\} \lesssim \frac{d p \log p}{n} \text{ \ \ with probability \ \ } 1 - O(p ^ {-10}),
\end{equation*}
where $\mathcal{O}(d)$ denotes the collection of $d \times d$ orthogonal matrices. Furthermore, we have
\begin{equation*}
    \big\|\widehat{\bTheta} - \bTheta ^ * \big\|_{\mathrm{F}} \lesssim \sqrt{\frac{d p \log p}{n}} \text{ \ \ with probability \ \ } 1 - O(p ^ {-10}).
\end{equation*}
\end{theorem}

\begin{remark} [Theoretical Contribution of Theorem \ref{thm:Divide_conquer}] Theorem \ref{thm:Divide_conquer} shows that DANIEL preserves statistical performance despite its distributed and non-convex structure, attaining the same rate as the classical result from a non-distributed convex estimator for the low-rank Ising model \citep{nussbaum2019ising}. To the best of our knowledge, this is the first theoretical result on the statistical rate for non-convex estimators in the low-rank Ising model. Our rate for the non-convex estimators $\widehat{\Ub}$ and $\widehat{\Vb}$ matches those of state-of-the-art centralized non-convex estimators in matrix completion problems \citep{chen2020noisy}. 
\end{remark}

\subsection{Rationality for Initialization of DANIEL} \label{sec:rational_initial}

The proof of Theorem~\ref{thm:Divide_conquer} requires the initialization $\widehat{\boldsymbol\Theta}_{0}$ to achieve a sufficiently sharp rate for the validity of DANIEL, as we specify in Appendix~\ref{sec:app:ini} for details. In Theorem \ref{thm:main}, we show that the statistical rate of $\widehat{\boldsymbol\Theta}_{\rm{cvx}}$, the solution to the optimization problem \eqref{equ:question}, satisfies the initialization requirements, thereby allowing DANIEL to attain its theoretically guaranteed performance. The proof of Theorem~\ref{thm:main} is provided in Appendix \ref{sec: convex thm}.

\begin{theorem}[Rate Guarantee for DANIEL Initialization] \label{thm:main}
Under Assumptions \ref{ass:Positive_minimal_eigenvalue}-\ref{ass:reg_grad}, if $md^2 p^4 \log p/n = O(1)$ and we choose $\lambda = C_0 \sqrt{\frac{p \log p}{n/m}}$ for a sufficiently large constant $C_0$, then the solution to \eqref{equ:question}, denoted as $\widehat{\boldsymbol\Theta}_{\rm{cvx}}$, satisfies that
\begin{equation*}
    \big \| \widehat{\boldsymbol\Theta}_{\rm{cvx}} - \bTheta ^ * \big \|_{\mathrm{F}} \lesssim \sqrt{\frac{dp \log p}{n/m}}, \text{\ \ \ with probability \ \ } 1 - O(p^{-10}).
\end{equation*}
\end{theorem}

\begin{remark} [Contribution of Theorem \ref{thm:main} on Valid Initialization] \label{rem:main} Theorem \ref{thm:main} shows that the statistical rate of $\widehat{\boldsymbol\Theta}_{\rm{cvx}}$ is of lower order as required by Lemma \ref{lem:initial_value} in Appendix \ref{sec:app:ini}, indicating the sufficiency of solving \eqref{equ:question} locally on the hub institution as an initialization strategy for DANIEL. The empirical validity of our proposed strategy is further examined through simulation studies presented in Section \ref{sec:simulation}.

\end{remark}

\subsection{Comparison to the Centralized Estimator} \label{sec:cd}

Finally, we present the validity of the distributed layout of DANIEL as compared with a centralized estimator $\widetilde{\bTheta} = \widetilde{\Ub}\widetilde{\Vb}\trans$, where we replace the distributed loss in \eqref{equ:surrogate_obj} to a centralized bi-factored loss as \begin{equation}
    \label{eq:centralized_estimator}
    (\widetilde{\Ub}, \widetilde{\Vb}) = \argmin_{\Ub, \Vb} \mathcal{L}(\Ub\Vb\trans) + \frac{1}{4}\big\|\Ub\trans \Ub - \Vb \trans \Vb\big\|_{\mathrm{F}} ^ 2,
\end{equation}
where $\mathcal{L}(\cdot)$ is defined in \eqref{equ:objective}. Theorem \ref{thm:rate_ctr_DANIEL} shows that the rate between the DANIEL estimator $\widehat{\bTheta}$ and the centralized estimator $\widetilde{\bTheta}$ is of lower order than the rate between $\widehat{\bTheta}$ and the truth $\bTheta ^*$. The proof of Theorem \ref{thm:rate_ctr_DANIEL} is provided in Appendix~\ref{app:rate_ctr_DANIEL}.

\begin{theorem}[DANIEL's Proximity to the Centralized Estimator]\label{thm:rate_ctr_DANIEL}
    Under the conditions of Theorem~\ref{thm:Divide_conquer} and Theorem~\ref{thm:main}, with the initialization $\widehat\bTheta_0$ obtained from $\widehat\bTheta_{\rm cvx}$ that solves \eqref{equ:question} and number of institutions $m= o(\sqrt{n/(dp^7\log p)})$, we have 
\begin{equation*}
    \inf_{\Ob \in \mathcal{O}(d)} \Big\{ \big\|\widehat{\Ub}  - \widetilde\Ub \Ob \big\|_{\mathrm{F}} ^ 2 + \big\|\widehat{\Vb}  - \widetilde\Vb \Ob \big\|_{\mathrm{F}} ^ 2 \Big\} = o\Bigl( \frac{d p \log p}{n} \Bigr) \text{ \ \ with probability \ \ } 1 - O(p ^ {-10}),
\end{equation*}
where $\mathcal{O}(d)$ denotes the collection of $d \times d$ orthogonal matrices. Furthermore, we have 
\begin{equation*}
    \big\|\widehat{\bTheta} - \widetilde\bTheta \big\|_{\mathrm{F}} = o\biggl(\sqrt{\frac{d p \log p}{n}} \biggr) \text{ \ \ with probability \ \ } 1 - O(p ^ {-10}).
\end{equation*}
\end{theorem}  

\begin{remark} [Contribution of Strategy in Proof to Theorem \ref{thm:rate_ctr_DANIEL}] The proof of Theorem \ref{thm:rate_ctr_DANIEL} imposes a slightly stronger requirement on the number of institutions $m$ than that in Theorem~\ref{thm:Divide_conquer}, in order to guarantee that the distance between $\widehat{\bTheta}$ and $\widetilde{
\bTheta}$ is asymptotically negligible compared to their rates to the true parameter $\bTheta ^*$. Unlike the proof to Theorem~\ref{thm:Divide_conquer}, here we must show that the scale of the gradient  $\nabla_{\bTheta}\cL(\widetilde\bTheta)$ is $o_P\Bigl(\sqrt{{d p \log p}/{n}} \Bigr)$, yet $\nabla_{\bTheta}\cL(\widetilde\bTheta)$ is generally non-zero due to the inherent non-convexity of the loss function in \eqref{eq:centralized_estimator}. To overcome this challenge, we develop a novel proof technique that enables a more refined analysis of the distributed-centralized gap in the presence of non-convexity, as detailed in Appendix~\ref{app:rate_ctr_DANIEL}. To the best of our knowledge, this is the first theoretical result establishing the equivalence of rates between distributed and centralized estimators for non-convex losses. Moreover, our proof strategy extends naturally to a broader class of non-convex losses with strong convexity and smoothness regularity, such as Assumption~\ref{ass:Positive_minimal_eigenvalue}, thereby providing a general methodological framework for future theoretical analyses of distributed optimization in non-convex settings.

\end{remark}

\section{Simulation Study} \label{sec:simulation}



 
In this section, we present simulation studies to evaluate the performance of the proposed DANIEL algorithm, focusing on three key research questions: \textbf{RQ1:} How does DANIEL perform in a centralized setting compared to convex optimization-based alternatives in terms of estimation accuracy? \textbf{RQ2:} How does DANIEL perform in the proposed distributed setting, as compared with the centralized setting, in terms of estimation accuracy? \textbf{RQ3:} How does the computational efficiency of DANIEL compare to existing centralized and distributed baseline algorithms? 


\subsection{Simulation Study Settings} \label{sec:sim_protocol}

\noindent \textbf{Tasks and Preparation of Simulated Datasets.} To provide a comprehensive analysis, we compare DANIEL with baseline methods in terms of estimation accuracy and computation efficiency under both centralized and distributed settings. In our simulation study, the dimension of features $p$ is selected from the set $\{50, 100, 150, 200\}$, with the rank $d$ set at $d = p/10$. The total sample size is set as $n \in \{1000, 5000, 10000, 20000\}$. For each setup of $(p,d,n)$, we generate a random matrix $\Ub ^ * \in \mathbb{R}^{p\times d}$ with entries $\Ub^*_{ij} \overset{i.i.d.}{\sim} \mathcal{N} \big(0,(dp)^{-1}\big)$. Subsequently, we draw $n$ independent samples from the Ising model \eqref{equ:def_X} with the true parameter matrix $\bTheta ^ * = \Ub^* \Ub^{* \transpose} \in \mathbb{R} ^ {p \times p}$. The samples are evenly partitioned across $m = \lfloor n^{x} \rfloor$ institutions, where $x \in [0,1]$ controls the \textit{distributedness level} of the simulated data. Specifically, $x = 0$ implies that $m = n^0 = 1$, which represents the case where the centralized algorithm is applied without a distributed layout. The initialization of DANIEL, denoted as $\widehat{\bTheta}_{\rm 0}$, is obtained by solving the convex optimization problem \eqref{equ:question} on the samples stored at the hub institution $\mathcal{S}_1$, and the estimator $\widehat{\bTheta}$ is derived by applying DANIEL until convergence. Details regarding the selection of hyperparameters for DANIEL, including the number of iterations, learning rates, and initialization strategies, are provided in Appendix~\ref{app:hyperparam_sim}.



\noindent \textbf{Baselines.} In our simulation study, the performance of DANIEL is compared against four baseline models that showcase different convex optimization techniques: (I) \textit{SV-Soft}, which applies the proximal gradient descent (PGD) algorithm with soft-thresholding on singular values \citep{cai2010singular}; (II) \textit{SV-Hard}, which applies singular value hard-thresholding after each gradient update of the likelihood loss \citep{gavish2014optimal}; (III) \textit{SV-Topd}, which retains only the top-$d$ singular values after each gradient update \citep{eckart1936approximation}; and (IV) \textit{PSD-Cvx}, which minimizes the likelihood loss \eqref{equ:objective} under positive semi-definite (PSD) constraints \citep{boyd2004convex}. To ensure consistency, all baseline models are evaluated in centralized and distributed settings, identical to the evaluation of DANIEL, and are trained with the same stopping criteria. Details on the explanation and hyperparameter setup for all baseline models are included in Appendix~\ref{app:hyperparam_base}.




\noindent \textbf{Training Protocol.} In our simulation study, we compare the performance of DANIEL with the baseline methods under both centralized and distributed settings. To adapt the baseline models to distributed settings, we follow the procedures proposed by \cite{jordan2019dist}, with additional implementation details provided in Appendix~\ref{app:hyperparam_base}. In the visualizations, including Figure~\ref{fig:simu_result_errtime}, Figure~\ref{fig:tmp1}, and Figure~\ref{fig:app_simu} in Appendix~\ref{app:simu_result}, the methods are distinguished by color, as the data points and trajectories of DANIEL are shown in red, whereas the baseline trajectories focusing on convex optimization are shown in other colors. The simulation procedure for each method under each $(p,d,n,x)$ configuration is repeated 200 times to ensure reliability. 


\noindent \textbf{Evaluation Metrics.} We evaluate the performance of DANIEL and baselines with two key metrics. First, estimation accuracy is measured by the Frobenius norm error of the estimated parameter from the true parameter $\bTheta ^ *$, defined as $\mathrm{Err}_{\mathrm{F}}(\cdot) = \|\cdot - \bTheta^*\|_{\mathrm{F}}$. Second, computational efficiency is measured by the average computation time of each method under the respective configuration. All empirical evaluations are based on consistent settings following widely adopted standards in prior work, such as \citep{wang2017efficient,alqahtani2019performance}.




\subsection{Evaluation on Estimation Performance of DANIEL}

\noindent To answer \textbf{RQ1}, we evaluate the estimation accuracy of DANIEL under centralized settings, where the bi-factored surrogate loss \eqref{equ:surrogate_obj} degenerates to the loss on a single institution with balancing regularization. Specifically, we compare the Frobenius norm error of the DANIEL estimator with those obtained from baseline models detailed in Section~\ref{sec:sim_protocol} on the scenario where the distributedness level is fixed at $x = 0$.

\noindent To answer \textbf{RQ2}, we examine the effectiveness of DANIEL in distributed settings, for which it is designed. We compare the trajectories of the Frobenius norm error for DANIEL and the baseline methods adapted in a distributed setup as the distributedness level $x$ rises incrementally from $0$ to $0.6$. Our evaluation primarily focuses on the slope of these trajectories, where the Frobenius norm error increases as the same amount of samples is partitioned across more institutions subject to data privacy constraints.

\begin{figure}[h!]
     \begin{center}
         \includegraphics[scale = 0.32]{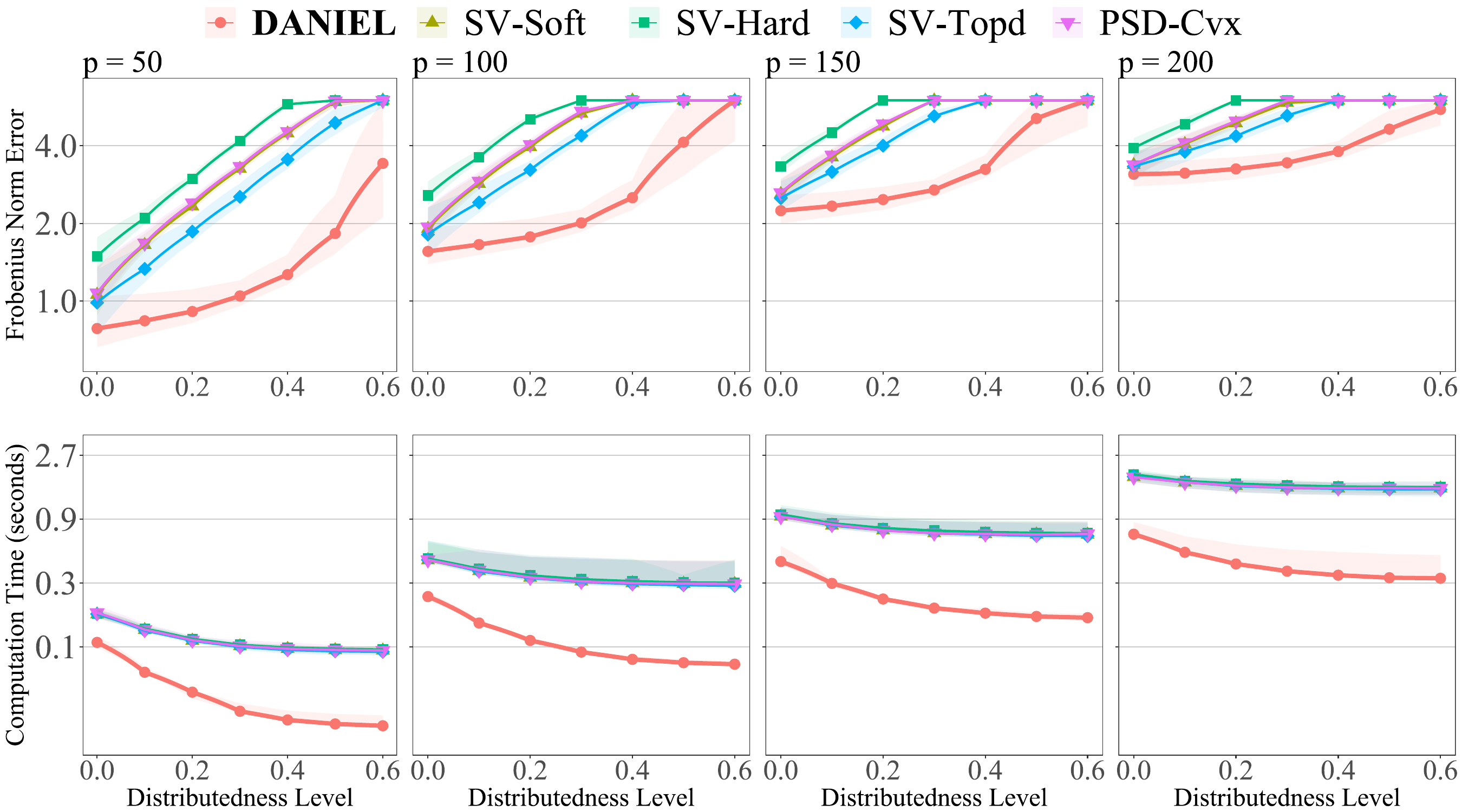}
     \end{center}
     \vspace{-20pt}
     \caption{The trajectories of the $\|\cdot\|_{\text{F}}$-error and computation time across different feature dimensions $p$, with the distributedness level $x$ (where $m = \lfloor n^{x} \rfloor$) varying from 0 to 0.6. The total sample size is set at $n = 1,000$. The trajectories of DANIEL are shown in red, demonstrating superior performance and efficiency over baseline methods for distributedness levels $x < 0.5$ (i.e., $m = o(\sqrt{n})$). Vertically, lower $\|\cdot\|_{\text{F}}$-error and shorter computation time indicate better performance for any fixed $x$; horizontally, a flatter error trajectory as $x$ increases is preferred, as this indicates that the distributed estimators remain valid as compared with their centralized counterparts. }
    \label{fig:simu_result_errtime}
\end{figure}

We visualize the estimation accuracy of DANIEL with fixed total sample size $n = 1,000$ across different setups of the feature dimension $p$ and distributedness level $x$ in Figure~\ref{fig:simu_result_errtime}, with additional results on the performance of DANIEL on larger $n$ detailed in Appendix~\ref{app:simu_result} and visualized in Figure~\ref{fig:app_simu}. Furthermore, we distribute $n=10,000$ samples across $m = 15$ institutions (approximately corresponding to $x = 0.3$) and evaluate the estimation accuracy of DANIEL and the baselines as $p$ increases from $20$ to $200$, with the trajectories presented in Figure~\ref{fig:tmp1}. We summarize the key observations below: (i) From the perspective of non-convex optimization, DANIEL consistently achieves lower Frobenius norm error across all tested settings compared to baselines focusing on convex optimization, both in centralized settings \textit{and} in distributed settings where the number of institutions satisfies the classic theoretical threshold $m = o(n^{0.5})$ as indicated by \cite{jordan2019dist}. The empirical evidence, especially for the centralized case where $x=0$, provides strong support for Theorem~\ref{thm:Divide_conquer}, showing that DANIEL's non-convex approach remains superior to baselines even when the conditions on the gradient of the loss $\mathcal{L}(\cdot)$ and the magnitude of $n$ specified in the theorem are not satisfied. This result confirms that these regularity conditions for the proof of Theorem~\ref{thm:Divide_conquer} are sufficient but not necessary in practice. (ii) From the perspective of effectiveness on initialization, DANIEL converges reliably across all tested settings, aligning with our theoretical result in Theorem \ref{thm:main}. Our simulation study enhances the robustness and practical utility of DANIEL by showing that solving the convex optimization problem \eqref{equ:question} on the hub institution $\mathcal{S}_1$ serves as a valid initialization strategy, while obtaining $\widehat{\bTheta}_0$ as the rank-$d$ SVD of the exact solution $\widehat{\bTheta}_{\mathrm{cvx}}$ to \eqref{equ:question} is not practically necessary, as noted in Remark~\ref{rem:main} and reflected in our study setup detailed in Appendix~\ref{app:hyperparam_sim}. (iii) From the perspective of distributed layout, DANIEL shows steady Frobenius norm error across varying levels of distributed level $x$, even when the conditions on initialization $\widehat{\bTheta}_0$ and the number of institutions $m$ in Theorem~\ref{thm:rate_ctr_DANIEL} are substantially relaxed. The results indicate that the inherent distributed framework of DANIEL maintains performance comparable to its centralized version up to the classical threshold $m = o( n^{0.5})$, consistent with Theorem~\ref{thm:rate_ctr_DANIEL} under the ideal asymptotic scenario where $p = O(1)$ and $n \xrightarrow[]{} \infty$. Furthermore, DANIEL presents a flatter slope in Figure~\ref{fig:tmp1} as compared to the baselines, suggesting its robustness and broad applicability for high-dimensional, large-scale data where the feature dimension $p$ is relatively large.


\subsection{Evaluation on Computation Efficiency of DANIEL}

To answer \textbf{RQ3}, we compare the average computation time of DANIEL against baseline methods as we raise the distributedness level $x$ from 0 to 0.6 for each fixed $(p,d,n)$ configuration. We focus on both the static efficiency at individual values of $x$ and the trajectory of computation time as $x$ increases.

We visualize the computation time of DANIEL and baselines with fixed total sample size $n = 1,000$ in Figure~\ref{fig:simu_result_errtime}, with additional results on larger $n$ presented in Figure~\ref{fig:app_simu} and Table~\ref{fig:app_simu_tab_time}. We summarize the key observations below: (i) From the perspective of static efficiency, DANIEL consistently outperforms convex baselines in average computation time. This advantage arises because a full convergence on \eqref{equ:question} is not required in the initialization of DANIEL, and DANIEL's subsequent iterations only involve the gradient updates of the $p \times d$ bi-factors. In contrast, the baseline methods perform costly SVD on a full-size $p \times p$ matrix at each iteration until the convergence criteria are met. (ii) From the perspective of efficiency trend with distributedness, DANIEL further reduces computation time relative to its centralized counterpart as the same amount of sample $n$ is distributed in more institutions under data privacy constraints. This result highlights the scalability of DANIEL's distributed framework, showing its strong suitability for real-world applications, where the data are naturally fragmented across multiple institutions yet a unified analysis is required.

\begin{figure}[h!]
    \begin{center}
        \includegraphics[width=4in]{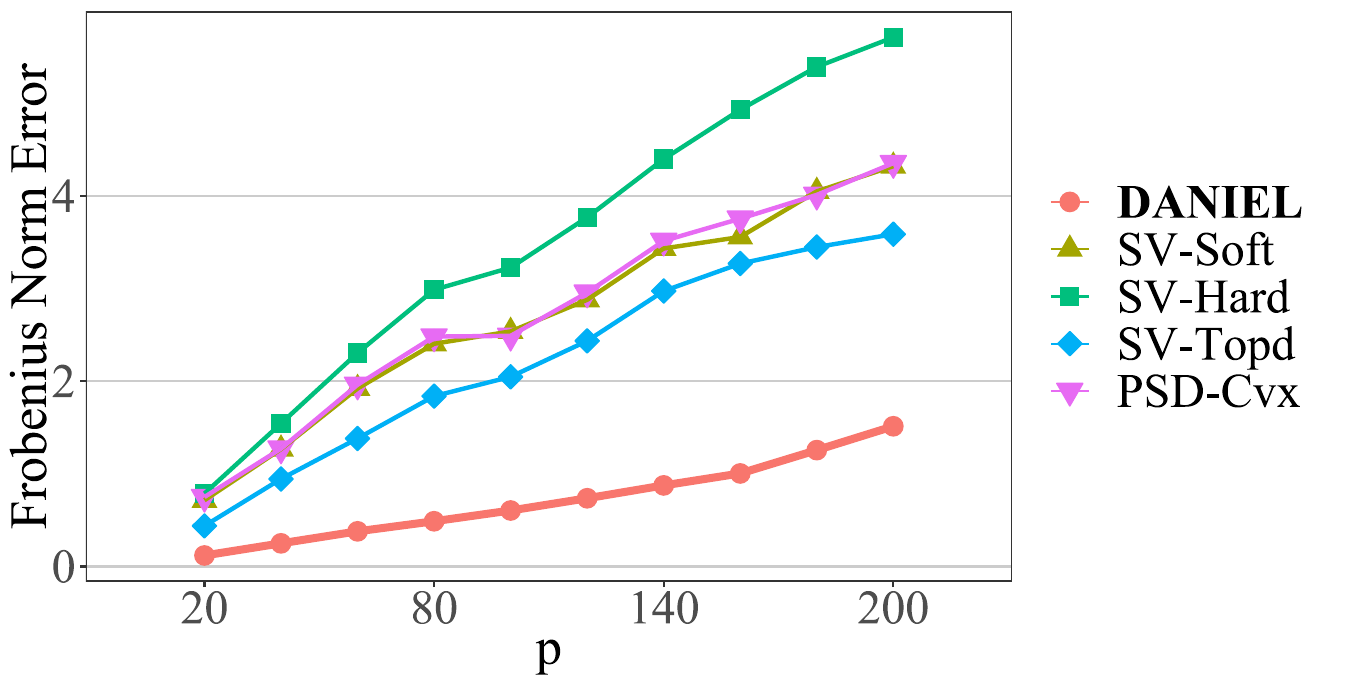}
    \end{center}
    \vspace{-30pt}
    \caption{The trajectories of the $\|\cdot\|_{\text{F}}$-errors across different methods in the simulation setup with a total of $n=10,000$ samples partitioned into $m = 15$ institutions and the feature dimension $p$ increasing from 20 to 200. The slope of each trajectory indicates the sensitivity of estimation accuracy to increasing $p$. Flatter trajectories are preferred, as they indicate greater robustness to high-dimensional features. }
    \label{fig:tmp1}
\end{figure}


\section{Application of DANIEL to EHR for Neurodegenerative Diseases}
\label{sec:real_data}


In this section, we evaluate the practical applicability of the proposed DANIEL algorithm using real-world EHR data, focusing on neurodegenerative diseases. Our primary objective is to address the research question: How does DANIEL perform in terms of estimation accuracy across various downstream tasks when compared with centralized non-convex methods, convex optimization-based methods, and machine learning (ML)-based methods?

\subsection{Model Description} \label{sec:model_EHR}

Constructing patient representations based on EHR data facilitates a range of clinical tasks, such as identifying similar patients \citep{sharafoddini2017patient}, predicting mortality \citep{allyn2017comparison}, and patient phenotyping \citep{zhu2016measuring, blinov2021patient}. For each patient, we assume that a latent variable $\yb \in \mathbb{R}^d$ encapsulates essential information influencing feature occurrence and joint distribution in clinical data. This unobservable latent Gaussian variable $\yb$, referred to as the \textit{patient embedding}, informs downstream tasks such as patient clustering. 
We model the joint distribution $p(\xb, \yb)$ between the binary feature occurrence vector $\xb$  and the patient embedding $\yb$ as a mixed-value graphical model previously studied by \cite{nussbaum2019ising}:
\begin{equation} \label{equ:EHR_model}
    p(\xb,\yb) \propto \exp\Big( 2\yb\trans \Ub\trans\xb - \yb\trans \yb + \bb\trans \xb \Big).
\end{equation}
 In model \eqref{equ:EHR_model}, the feature embedding matrix $\Ub\in\mathbb{R}^{p\times d}$ aligns the distribution of binary variable $\xb$ with the patient embedding $\yb$, describing how latent patient embeddings determine the probabilities of clinical feature occurrences. The marginal distribution of $\xb$ takes the form
$p(\xb) \propto \exp \big(\xb\trans \Ub\Ub\trans \xb + \bb\trans\xb \big) $, which corresponds precisely to the low-rank Ising model considered in our work. To ensure model consistency, we set the linear term coefficients $\bb = (b_1,\ldots,b_p)\trans$ as $b_j=(\Ub\Ub\trans)_{jj}$, which ensures that the structure satisfy the Ising model constraints with a low-dimensional positive semi-definite parameter matrix $\bTheta^*=\Ub^{*}(\Vb^{*})\trans=\Ub\Ub\trans$, where $\Ub^{*}=\Vb^{*}=\Ub$. We use $\bTheta ^ *$ and $\Ub ^ *$ to generate patient profiles and support downstream tasks such as patient clustering, as detailed in Section~\ref{sec:real_data_method}.


\subsection{Settings of Real-World EHR Data Analysis}
\label{sec:real_data_method}

\noindent\textbf{Explanation of EHR Datasets and Preprocessing.} To evaluate DANIEL's capability for embedding-based phenotyping, we utilize the EHR data from $28,955$ patients diagnosed with Alzheimer's Disease (AD) at the University of Pittsburgh Medical Center (UPMC) between 2011-2021, and $29,293$ patients diagnosed with AD from Mass General Brigham (MGB) between 2011-2021. We collect three domains of codified data: (i) {\tt PheCode} codes for diagnosis \citep{world1978international,bramer1988international}; (ii) clinical classification software ({\tt CCS}) codes for procedures \citep{elixhauser2009clinical}; and (iii) {\tt RxNorm} codes for medication usage \citep{bennett2012utilizing}. The comprehensive feature set involves the binary indicators of $p = 2,142$ features, including $1,339$ {\tt PheCode} codes, $175$ {\tt CCS} codes, and $628$ {\tt RxNorm} codes. To construct longitudinal patient profiles, we apply non-overlapping $6$-month observation windows to each patient's clinical history, yielding $n_1 = 400,904$ records from UPMC and $n_2 = 505,626$ records from MGB. Our strategy for EHR data preprocessing follows established standards in EHR studies, such as \citep{Gan2025ARCH}.

\noindent\textbf{Tasks and Evaluation Protocols.} Using DANIEL, we solve for the estimators $\widehat\bTheta\in\mathbb{R}^{p\times p}$ and $\widehat\Ub\in\mathbb{R}^{p\times d}$ for the true parameter matrix $\bTheta ^ *$ and feature embedding $\Ub ^ *$ defined in Section \ref{sec:model_EHR}, with the rank parameter set as $d = 200$. The empirical rank $d$ is selected by maximizing the \underline{A}rea \underline{U}nder the Receiver Operating Characteristic \underline{C}urve (AUC) in detecting known-relationship pairs with the estimated parameter $\widehat\bTheta$. We then evaluate the clinical utility of the estimated $\widehat{\Ub}$ across three downstream tasks: (i) detection of known relationships among clinical features; (ii) patient phenotyping for refined clinical subgroup identification; and (iii) unsupervised patient clustering to capture differences in long-term clinical outcomes. Finally, we construct and visualize a knowledge graph for clinical features relevant to AD and multiple sclerosis (MS), showing DANIEL's potential in studies of global EHR embedding.

\noindent\textbf{Baselines.} We compare the performance of DANIEL in the EHR study with six baseline methods. Three of these baselines are the same as those used in the simulation study, namely (a) \textit{SV-Hard}, (b) \textit{SV-Topd}, and (c) \textit{PSD-Cvx}, each representing variants of regularized convex optimization. The other three are variants of \underline{B}idirectional \underline{E}ncoder \underline{R}epresentations from \underline{T}ransformers (BERT), which serve as state-of-the-art ML techniques specially designed for representation learning under modern setup, including (d) \textit{BERT}, a pretrained multi-layer bidirectional Transformer encoder that overcomes the unidirectionality constraints by utilizing a masked language model pretraining objective \citep{devlin2018bert}; (e) \textit{BioBERT}, which is a variant of BERT trained on both \textit{general domain} and \textit{biomedical domain} corpora to provide representation tailored for biomedical context \citep{lee2020biobert}; and (f) \textit{SapBERT}, or \underline{S}elf-\underline{A}lignment \underline{P}retraining of BERT, which enhances biomedical entity representation by aligning synonymous terms in the Unified Medical Language System (UMLS) \citep{Bodenreider2004UMLS} during pretraining, thereby improving the model’s ability to capture fine-grained biomedical semantics \citep{liu2020self}. The clinical feature embeddings from these BERT-based models are derived solely from textual descriptions of EHR concepts. The baseline models output either estimators $\widehat{\bTheta}_{\mathrm{base}}$ for the true parameter matrix $\bTheta ^ *$ or the estimators for the clinical feature embeddings $\Ub^*$. To ensure a fair comparison, we apply the same cross-validation process to all methods in the performance evaluation of patient phenotyping. A detailed summary of the baseline configuration is provided in Appendix~\ref{sec:Exp_config}.




\subsection{Detecting Known-relationship Pairs} \label{sec:DKP}

To assess the quality of the embeddings, we evaluate DANIEL on the task of detecting known-relationship pairs among EHR features. This evaluation strategy is a widely adopted standard in prior work on feature embedding and knowledge graph construction, such as \citep{Hong2021KESER, Gan2025ARCH}. Heuristically, a high-quality embedding strategy should preserve the proximity of clinically similar or associated concepts in the embedding space. In our study, we focus on two widely studied categories of relationships in codified EHR data, namely \textit{similarity} and \textit{relatedness}, which are curated by domain experts and reflect consensus biomedical knowledge \citep{Bodenreider2004UMLS}. Similar feature pairs are defined on hierarchical structures among {\tt PheCode} and {\tt RxNorm} codes, which are often considered as close substitutes in specific clinical contexts. For example, {\tt PheCode:189.1} (Cancer of kidney and renal pelvis) is \textit{similar to} {\tt PheCode:189.11} (Malignant neoplasm of kidney, except renal pelvis), and such codes with inherent hierarchical relationships are typically treated as interchangeable in EHR data curation for studies of patients diagnosed with renal cell carcinoma (RCC), such as \citep{Hou2024RCC}. On the other hand, related feature pairs are extracted from UMLS relations, including \textit{related to}, \textit{mapped to}, and \textit{classifies}. The \textit{related to} relation reflects general clinical associations without implying hierarchy or equivalence. For example, {\tt PheCode:537} (Disorders of stomach function) is \textit{related to} {\tt PheCode:785} (Symptoms involving cardiovascular system) in the context of chest pain arising from gastrointestinal conditions such as gastroesophageal reflux disease (GERD). The \textit{mapped to} relation captures EHR code transformations. For example, {\tt PheCode:292.12} (Drug-induced mental disorders) is \textit{mapped to} {\tt PheCode:315.1} (Developmental delays and disorders), as the same clinical case involving behavioral health conditions due to excessive lead exposure in childhood may be coded differently across institutions, with one code reinterpreted as the other. The \textit{classifies} relation groups concepts into broader clinical categories without implying a strict hierarchical structure. For example, {\tt PheCode:561} (Functional digestive disorders) \textit{classifies} {\tt PheCode:579} (Other disorders of the digestive system), suggesting that both codes share a common diagnostic or phenotypic context that is useful for classification and aggregation. For each relationship type, we randomly sample the same number of pairs that \textit{have} the respective relationship and pairs that \textit{do not}. We then compute the AUC by comparing the model-induced proximity of known-relationship pairs against that of randomly paired features. To further evaluate the validity of DANIEL's distributed layout, we also report the AUC for the centralized version of DANIEL trained solely on UPMC or MGB data. A higher AUC achieved by a method indicates more accurate discrimination between pairs with and without the relationship, reflecting the method's superior ability in detecting clinically meaningful associations. Further details on model implementation and AUC computation are provided in Appendix~\ref{sec:Exp_config}.

\begin{table}[h!]
    \caption{AUC results for detecting different types of known-relationship pairs using DANIEL and baseline methods. Bio and Sap stand for BioBERT and SapBERT, respectively. UPMC and MGB stand for the cases where the centralized version of DANIEL is trained solely on UPMC or MGB data. Higher AUC values indicate better performance and are preferred. The best results are in \textbf{bold}.}
    \label{tab:AUC}
    \centering
    \resizebox{\textwidth}{!}{
        \addtolength{\tabcolsep}{-5pt}
        \begin{tabular}{r|c|cccccc|cc}
            \toprule
            \multirow{2}*{Relation types} & &  \multicolumn{6}{c}{Baselines} & \multicolumn{2}{|c}{Partial Data} \\
            & \textbf{DANIEL} & SV-Hard & SV-Topd & PSD-Cvx & BERT & Bio & Sap & UPMC & MGB \\
            \hline
            {\tt PheCode} Hierachy & \textbf{0.889} & 0.854 & 0.776 & 0.866 & 0.609 & 0.593 & 0.856 & 0.792 & 0.781 \\
            {\tt RxNorm} Hierachy & \textbf{0.735} & 0.592 & 0.520 & 0.590 & 0.569 & 0.570 & 0.573 & 0.559 & 0.467 \\
            \hline
            related to & 0.743 & 0.681 & 0.653 & 0.684 & 0.586 & 0.533 & \textbf{0.790} & 0.642 & 0.627 \\
            mapped to & \textbf{0.820} & 0.776 & 0.673 & 0.778 & 0.494 & 0.502 & 0.795 & 0.710 & 0.739 \\
            classifies & 0.759 & 0.726 & 0.674 & 0.728 & 0.621 & 0.573 & \textbf{0.816} & 0.672 & 0.664 \\
            \bottomrule
        \end{tabular}
    }
\end{table}

We present the AUC results for DANIEL and baseline methods in distinguishing known-relationship pairs in Table~\ref{tab:AUC}. We summarize the key observations below: (i) From the perspective of the superiority of non-convex optimization, DANIEL consistently achieves higher AUC than SV-Hard, SV-Topd, and PSD-Cvx where all models are trained on the combined UPMC-MGB dataset subject to privacy constraints. This result highlights the performance advantage of DANIEL's non-convex optimization framework over conventional approaches that rely on convex optimization. (ii) From the perspective of comparison with BERT-based models, DANIEL outperforms BERT and BioBERT in all the tasks and surpasses the state-of-the-art SapBERT in the majority of tasks. These results demonstrate the untapped potential of structured EHR data and emphasize that clinical feature extraction should not rely exclusively on textual descriptions. (iii) From the perspective of unifying distributed data under privacy constraints, DANIEL, which integrates the data from both UPMC and MGB, consistently outperforms its centralized version trained on a single institution. This result validates DANIEL's distributed framework in harmonizing multi-institutional information while preserving compliance with healthcare data privacy regulations.



\subsection{Patient Phenotyping}

Accurate patient phenotyping is crucial for understanding and managing chronic neurodegenerative diseases such as AD and related dementia \citep{tang2022deep,wei2025automated}, where a precise differentiation between \textit{pre-diagnosis} and \textit{post-diagnosis} stages is essential in both research and clinical care \citep{khan2024exploring,hoang2025economic}.  To evaluate DANIEL's utility in this task, we compare its performance with baseline methods in predicting whether a patient record from the UPMC or MGB AD cohort falls in the \textit{pre-dementia} or \textit{post-dementia} stage within its respective 6-month observation window, given the occurrence patterns of the other features \citep{wang2021derivation, gupta2022flexible}. In data preparation, we use the first recorded target code {\tt PheCode:290.1} (Dementias) for each patient to establish the diagnosis time. We extract the earliest records from each patient, which are labeled as \textit{negative} samples (no dementia), and an equal number of records from each patient randomly after the first recorded {\tt PheCode:290.1} as \textit{positive} samples (with dementia), with details presented in Appendix~\ref{sec:Exp_config}.  For DANIEL and the convex baselines (SV-Hard, SV-Topd and PSD-Cvx), we train the model and compare their AUC in an \textit{unsupervised} setting, where the true labels are compared against the conditional probabilities evaluated via \eqref{equ:conditional_dist}, using the occurrence patterns of all features except {\tt PheCode:290.1} and the plug-in estimators to the true parameter matrix $\bTheta ^ *$, namely $\widehat\bTheta$ from DANIEL and $\widehat{\bTheta}_{\mathrm{base}}$ from the baselines. For all the baselines, including the BERT-based models, we obtain the feature embeddings that do not directly predict the dementia label of a record. An additional \textit{supervised} procedure is therefore implemented to assess each model's ability to classify the records by the patients' dementia status in the respective 6-month window. We implement a 10-fold cross-validation to prevent data leakage, ensuring that in each fold the designated test set includes all patients reserved exclusively for testing, while the training is conducted on the remaining patients. Following Section \ref{sec:DKP}, we report the AUC for centralized DANIEL trained solely on UPMC or MGB data. A higher AUC indicates a method's stronger ability to distinguish patient records by dementia, the phenotype of interest in our study. Further details on the supervised evaluation procedure and the cross-validation setup are provided in Appendix~\ref{sec:Exp_config}. 

We present the AUC results for DANIEL and baseline methods in classifying records by dementia status within their respective observation window in Table~\ref{tab:Risk_pred}. We summarize the key observations below: (i) From the perspective of optimization strategy, DANIEL consistently achieves the highest AUC among all baseline methods in both unsupervised and supervised setups, indicating its superiority for patient phenotyping based on the occurence patterns of codified EHR; (ii) From the perspective of EHR data integration, DANIEL outperforms its centralized version trained solely on UPMC or MGB data, indicating the advantage of its distributed framework and highlighting its potential to leverage multi-institutional EHR data for improved real-world patient phenotyping.

\begin{table}[h]
    \caption{AUC results for AD patient phenotyping using DANIEL and baseline methods. Bio and Sap stand for BioBERT and SapBERT, respectively. UPMC and MGB stand for the cases where the centralized version of DANIEL is trained solely on UPMC or MGB data. Higher AUC values indicate better performance and are preferred. The best results are in \textbf{bold}. }
    \label{tab:Risk_pred}
    \centering
    \resizebox{\textwidth}{!}{
        \addtolength{\tabcolsep}{-4pt}
        \begin{tabular}{c|c|cccccc|cc}
            \toprule
            \multirow{2}*{Method} & & \multicolumn{6}{c}{Baselines} & \multicolumn{2}{|c}{Partial Data} \\
            & \textbf{DANIEL} & SV-Hard & SV-Topd & PSD-Cvx & BERT & Bio & Sap & UPMC & MGB \\
            \hline
            Unsupervised & \textbf{0.834} & 0.790 & 0.789 & 0.794 & - & - & - & 0.782 & 0.824 \\
            Supervised & \textbf{0.891} & 0.876 & 0.883 & 0.876 & 0.877 & 0.877 & 0.877 & 0.877 & 0.890 \\
            \bottomrule
        \end{tabular}
    }
\end{table}

\subsection{Patient Clustering}

In this study, we utilize DANIEL-generated patient embeddings to cluster patients in the UPMC-MGB AD cohort into two groups based on the information gathered regarding the onset of dementia \citep{huang2021patient,li2025discovering}. After estimating the feature embedding $\widehat{\Ub}$ with DANIEL, we compute patient embeddings as $\widehat\yb = \widehat\Ub\trans \xb_{\mathrm{init}}$, where $\xb_{\mathrm{init}}$ represents each patient's initial record of dementia diagnosis, identified by the occurrence of {\tt PheCode:290.1} (Dementias). The resulting patient embeddings $\widehat{\yb}$ are then clustered using $k$-means \citep{hartigan1979algorithm}. We present the result of DANIEL-induced patient clustering using Kaplan-Meier survival curves \citep{efron1988logistic} in Figure~\ref{fig:KMCurve}. In Figure~\ref{fig:KMCurve}, we observe a bifurcation in the probability trajectories of doctor-ordered nursing home admission (NHA) between the two patient clusters, reflecting a different pace of dementia progression. This finding demonstrates DANIEL's capability to generate clinically meaningful patient embeddings that capture underlying heterogeneity in disease progression. The separation in NHA outcomes with statistical significance ($p < 0.0001$) observed over a 5-year timespan highlights DANIEL's strength in uncovering latent structures that are relevant to long-term clinical outcomes.

\begin{figure}[h!]
    \begin{center}
        \includegraphics[width=0.8\linewidth]{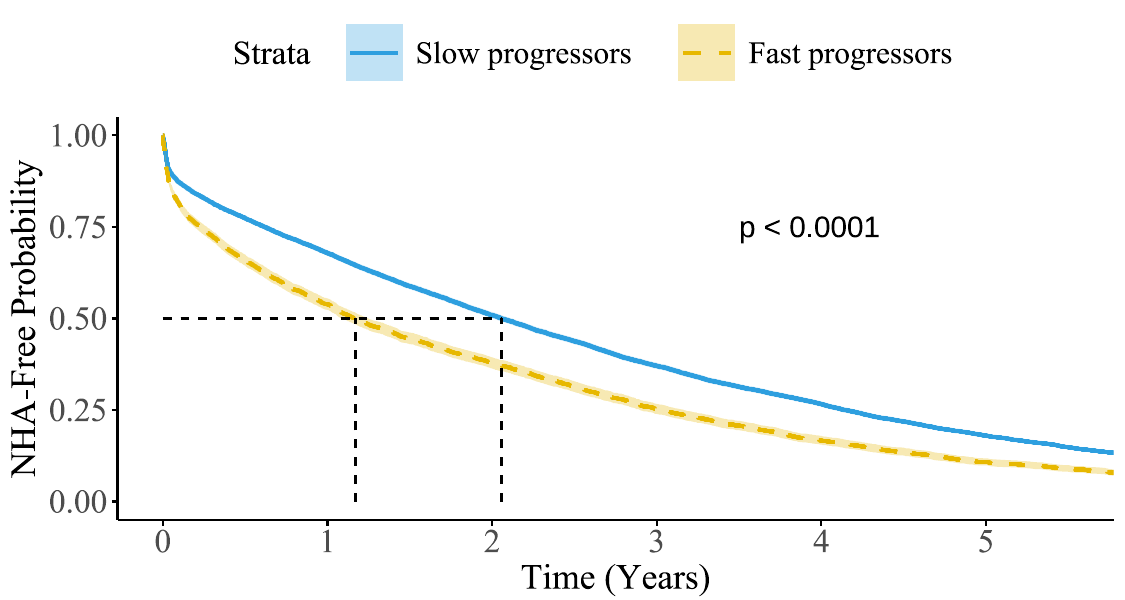}
    \end{center}
    \caption{Kaplan-Meier survival curves for two patient clusters obtained using $k$-means clustering on DANIEL-generated patient embeddings. The y-axis represents the estimated probability of nursing home admission, and the x-axis denotes time. A $p$-value is reported to assess statistical significance between the two clusters, with $p < 0.01$ indicating a difference with statistical significance. }
    \label{fig:KMCurve}
\end{figure}

\subsection{Knowledge Graph}

\begin{figure}[htb]
    \centering
    \vspace{-10pt}
    \subfigure[top features of patients with AD]{
        \includegraphics[width=2.7in]{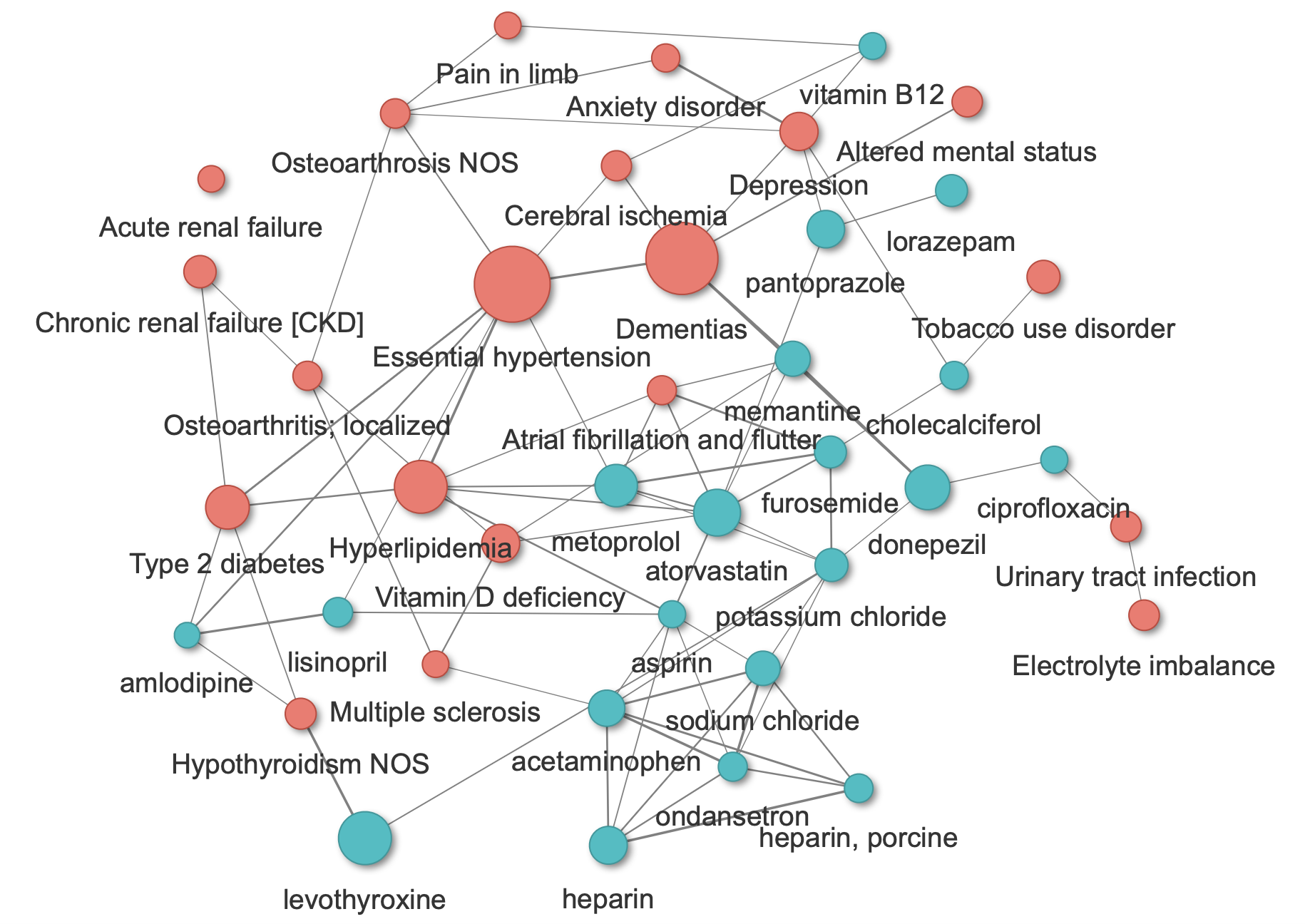}
        \label{fig:wordcloud_AD}
    }
    \subfigure[top features of patients with MS]{
	\includegraphics[width=2.7in]{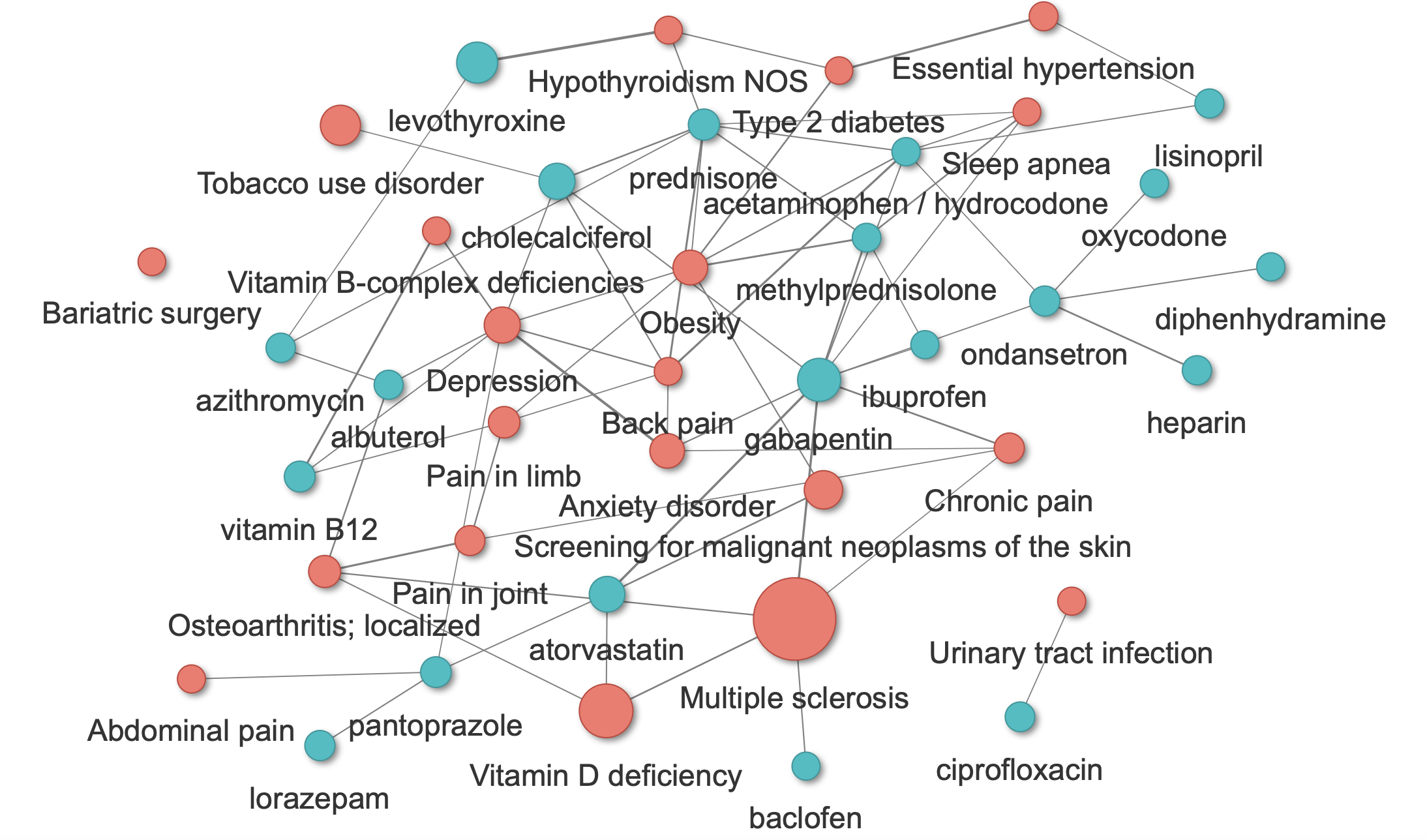}
        \label{fig:wordcloud_MS}
    }
    \caption{Knowledge graphs of (a) top features for AD patients and (b) top features for MS patients. Node size reflects the occurrence probability of each feature. Red nodes correspond to {\tt PheCode} for diagnosis, and blue nodes correspond to {\tt RxNorm} for medication usage. The presence and thickness of edges between nodes are determined by the values of DANIEL-estimated parameter matrix $\widehat{\bTheta}$.} 
    \label{fig:WordCloud}
\end{figure}

In this section, we visualize the knowledge graphs of main clinical features associated with AD and MS, as shown in Figure~\ref{fig:WordCloud}. In Figure~\ref{fig:WordCloud}, each node represents a top feature related to AD (left) or MS (right), with node size reflecting the frequency of feature occurrence within the patient group with AD or MS diagnosis. Red nodes denote {\tt PheCode} diagnosis codes, and blue nodes denote {\tt RxNorm} medication codes. The presence of the edges between the nodes is determined by DANIEL-derived parameter matrix $\widehat{\bTheta}$, where we set an appearance threshold of the $95\%$ quantile of all off-diagonal entries of $\widehat{\bTheta}$, and the edges $(j,k)$ corresponding to higher values $\widehat{\theta}_{jk}$ are marked thicker. 

In Figure~\ref{fig:WordCloud}, the patients with AD diagnosis exhibit an elevated likelihood of comorbid conditions, such as {\tt type 2 diabetes}, {\tt hyperlipidemia}, and {\tt atrial fibrillation}, in contrast to patients with MS. These comorbidities are supported by prior studies. For example, studies indicate that individuals with type 2 diabetes face an increased risk of developing dementia, including AD \citep{sridhar2015emerging}. Moreover, hyperlipidemia (especially hypercholesterolemia, i.e., elevated blood cholesterol level) and atrial fibrillation are also identified as significant risk factors for AD \citep{rosendorff2007cardiovascular, loera2019alterations, nakase2023impact}. While no medication can \textit{cure} AD, FDA-approved treatments, such as {\tt Donepezil} \citep{birks2018donepezil} and {\tt Memantine} \citep{reisberg2003memantine}, which reduce symptom progression and improve the patients' daily functioning, emerge as pivotal features in Figure~\ref{fig:WordCloud}.

In contrast, MS is recognized as a painful disease \citep{kenner2007multiple}, and the top associated features include {\tt chronic pain}, {\tt abdominal pain}, {\tt pain in joint}, {\tt back pain}, and {\tt pain in limb}. Prior studies support the observed comorbidities, such as {\tt Vitamin D deficiency} and {\tt Tobacco use disorder}. For example, lower vitamin D levels are linked to an increased risk of MS \citep{sintzel2018vitamin}, and smoking has been shown to accelerate MS severity and disability progression  \citep{manouchehrinia2013tobacco}. Similar to AD, MS is \textit{incurable} by medication. Instead, symptomatic treatments such as {\tt methylprednisolone} \citep{sloka2005mechanism}, {\tt gabapentin} \citep{houtchens1997open}, and {\tt baclofen} \citep{sammaraiee2019intrathecal}, which mitigate dysesthesias caused by MS, relief spasticity, reduce relapse severity, and slow long-term disability progression, are present in Figure~\ref{fig:WordCloud}. In summary, the alignment between the knowledge graphs presented in Figure~\ref{fig:WordCloud} and the literature in AD and MS studies highlights DANIEL's ability to learn global representations of healthcare concepts and convert them into insights relevant to clinical practice.


\section{Discussion}




In this paper, we introduce DANIEL, a scalable and privacy-aware framework for representation learning on high-dimensional, multi-source data under strict data-sharing restrictions. We employ Markov Random Fields, specifically the Ising model, for data represented as binary vectors, in a real-world, multi-institution setting where modern data privacy constraints invalidate the key assumptions of the classical theory for centralized low-rank Ising model optimization. Our objective is to learn high-quality global feature representations and patient embeddings for downstream analysis to the codified EHR data, including feature relationship detection, patient phenotyping, and patient clustering that can inform a long-term, precise treatment plan for the patient cohort of interest. To address the technical challenges in scalability and privacy, we develop DANIEL, which implements a distributed Ising framework that optimizes a non-convex bi-factored surrogate of the conditional likelihood. Our design not only exploits the low-rank structure of the data for efficiency as compared with computationally expensive SVD-based approaches, but also limits cross-institutional communication to a single round of non-sensitive gradient exchange, thereby eliminating the transfer of raw data. 

Our theory shows that, under standard regularity conditions, the DANIEL estimator attains finite-sample error rates comparable to a centralized oracle while requiring only limited cross-institution aggregation. We further establish the validity of our initialization procedure and show that the gap between the DANIEL estimator in its inherent distributed setting and its centralized counterpart is of lower order than the overall estimation error. Simulation studies across a wide range of configurations demonstrate that the DANIEL estimator achieves lower Frobenius norm estimation error and substantially shorter runtime than baselines focusing on convex optimization that conduct SVD at every iteration until convergence. Our empirical results indicate the robustness of DANIEL and suggest that the theoretical performance guarantees on the DANIEL estimator, including its initialization procedure and its distributed setting, continue to hold even when the conditions on feature dimensionality and the number of participating institutions proposed in our theory are substantially relaxed.

Extensive experiments on the codified EHR from the UPMC-MGB Alzheimer's Disease dataset demonstrate the superiority of DANIEL in detecting EHR feature pairs with known clinical relationships and distinguishing patients before versus after dementia diagnosis, compared to baseline methods based on convex optimization and modern variants of BERT focusing on biomedical data processing. Furthermore, DANIEL's patient embedding induces $k$-means clusters that stratify patients by their pace of Alzheimer's disease progression, suggesting that the learned representations capture disease heterogeneity relevant to long-term care pathways. The feature embeddings, as visualized via knowledge graphs, recover numerous literature-supported medical insights and demonstrate the potential for discovering additional feature co-occurrence patterns in Alzheimer's disease and multiple sclerosis (MS), a highly correlated disease, which motivates further research.

Looking forward, two future directions warrant future exploration and refinement. On the methodological and theoretical side, we aim to extend the DANIEL framework beyond the low-rank Ising model to broader probabilistic graphical model families. We would consider scenarios such as mixed discrete-continuous variables, time-varying interactions, and ontology-aware structures, which potentially broaden DANIEL's applicability while preserving a general federated design suitable for privacy-aware studies. On the practical and ML implementation side, we plan to integrate our federated statistical backbone with state-of-the-art LLM-driven sequence encoders and graph-augmented retrieval frameworks, such as GraphRAG \citep{Edge2025GraphRAG, Han2025GraphRAG}, to yield hybrid representations that maintain privacy, align with the stated research questions, and further improve downstream performance and utility for end users.

\section{Acknowledgments}

Zongqi Xia acknowledges funding in part by NINDS R01NS098023 and NINDS R01NS124882. Tianxi Cai acknowledges funding by NIH R01LM013614.

\appendix

\section{Proof for the Main Theorem}\label{sec:proof-main}

In this section, we provide the proof for Theorem \ref{thm:Divide_conquer}, the main theorem of our work, by assembling the auxiliary lemmas proved in Section \ref{sec:B}.
Recalling that $\bTheta ^ * = \Ub^* \Vb^{*\top}$, for simplicity and without loss of generality, we assume that  $\sigma_d(\bTheta ^ *) = \sigma_d ^2 (\Ub ^ *) =  \sigma_d ^2 (\Vb ^ *) $. We define $\kappa ^ * = \min \{\frac{1}{8}, \frac{5}{128} \frac{\kappa_{\min} \kappa_{\max}}{\kappa_{\min} + \kappa_{\max}} \}$ as the constant depending on the singular values.
Define
\begin{equation} \label{equ:K_U}
    K(\Ub^*; \kappa_{\min}, \kappa_{\max}) = \frac{4}{5} \kappa ^ * \sigma_d^2(\Ub^*) \min \Big\{ \frac{1}{(\kappa_{\min} + \kappa_{\max})}, 2 \Big\}.
\end{equation}

We first introduce several notations. 
Suppose that we have the stacked matrix $\Zb_0 = \begin{bmatrix}
\Ub_0 \\
\Vb_0
\end{bmatrix}$.  For the true value $\bTheta ^ * = \Ub^* \Vb^{*\top}$, we have $\Zb ^ * = \begin{bmatrix}
\Ub^* \\
\Vb^*
\end{bmatrix}$. We denote $\Zb_\gamma = \begin{bmatrix}
\Ub_\gamma \\
\Vb_\gamma
\end{bmatrix}$ for the $\gamma$th step of iteration in DANIEL, where $\gamma \in [\Gamma]$, in the hub institution $\mathcal{S}_1$ as specified in Assumption \ref{al:Divide_conquer}. Furthermore, we define the subspace distance $\rho ^ 2 (\cdot, \cdot)$ for two stacked matrices $\Zb_1 = \begin{bmatrix}
\Ub_1 \\
\Vb_1
\end{bmatrix}$ and $\Zb_2 = \begin{bmatrix}
\Ub_2 \\
\Vb_2
\end{bmatrix}$ such that
\begin{equation*} 
    \rho ^ 2 (\Zb_1, \Zb_2) = \inf_ {\Ob \in \mathcal{O}(d)}  \big\{\|\Ub_1 - \Ub_2 \Ob \| ^ 2 _ {\mathrm{F}} + \|\Vb_1 - \Vb_2 \Ob \| ^ 2 _{\mathrm{F}} \big \},
\end{equation*}
where $\mathcal{O}(d)$ implies the collection of $d \times d$ orthogonal matrices.

\begin{proof}[Proof of Theorem \ref{thm:Divide_conquer}]
    As the institution number $m \lesssim  \sqrt{ndp^7\log p }$ and thus by the rate bound of the initialization established in Theorem \ref{thm:main}, the initial value $\widehat{\bTheta}_0 = \widehat{\bTheta}_{\rm cvx}$ satisfies $p^3 \| \widehat \bTheta_0 - \bTheta ^ * \|_{\rm F} = o(1)$ and \eqref{equ:diff_initial}. Furthermore, the upper bound for the statistical error $e^2_{\mathrm{stat}}$ as a function of the initialization error $\| \widehat \bTheta_0 - \bTheta ^ * \|_{\rm F}$, as depicted by Lemma \ref{lem:epsilon_non-convex}, guarantees that $e_{\mathrm{stat}} = o(1)$ with probability $1 - O(p^{-10})$ under the initialization condition indicated for Theorem \ref{thm:Divide_conquer}.

    The proof of Theorem \ref{thm:Divide_conquer} is conducted by induction. As indicated by \eqref{equ:estat_iter}, the initialization condition \eqref{equ:diff_initial} implies that
    \begin{equation*}
        \rho^2(\Zb_{1}, \Zb^*) \leq \Big(1 - 2 \eta \kappa ^ * \sigma_d(\bTheta^*) \Big) \rho^2(\Zb_0, \Zb^*) + \frac{6 \eta (\kappa_{\min} + \kappa_{\max}) e^2_{\mathrm{stat}}}{\kappa_{\min} \kappa_{\max}}.
    \end{equation*}
    For any positive integer $k \geq 1$ such that $\rho ^ 2 (\Zb_k, \Zb^*) \leq K(\Ub^*; \kappa_{\min}, \kappa_{\max})$, we conclude from \eqref{equ:estat_iter} that the condition
    \begin{equation} \label{equ:estat_cond}
        e^2_{\mathrm{stat}} \leq \frac{\kappa ^ * \kappa_{\min} \kappa_{\max} \sigma_d(\bTheta ^ *) K(\Ub^*; \kappa_{\min}, \kappa_{\max})}{3 (\kappa_{\min} + \kappa_{\max}) }
    \end{equation}
    suffices to guarantee that $\rho ^ 2 (\Zb_{k+1}, \Zb^*) \leq K(\Ub^*; \kappa_{\min}, \kappa_{\max})$. Here our initialization conditions indicate that \eqref{equ:estat_cond} holds with probability $1 - O(p^{-10})$, given that the right-hand side of \eqref{equ:estat_cond} is a constant term.  Therefore, by applying Lemma \ref{lem:contraction_mapping} recursively for $k \geq 1$, we have
  \begin{equation}
    \label{eq:Z_gamma}
    \rho^2(\Zb_k, \Zb^*) \leq \Big(1 - 2 \eta \kappa ^ * \sigma_d(\bTheta^*) \Big)^k \rho^2(\Zb_0, \Zb^*) + \frac{12 \eta (\kappa_{\min} + \kappa_{\max}) e^2_{\mathrm{stat}}}{\kappa_{\min} \kappa_{\max}}.
  \end{equation}
  Consequently, if we choose the number of iterations $\Gamma = \Omega \Big( \log \big(\frac{n}{d p \log p} \big) \Big)$, we have 
    \begin{equation*}
        \inf_{\Ob \in \mathcal{O}(d)} \Big\{ \big\|\widehat{\Ub}  - \Ub^* \Ob \big\|_{\mathrm{F}} ^ 2 + \big\|\widehat{\Vb}  - \Vb^* \Ob \big\|_{\mathrm{F}} ^ 2 \Big\} \lesssim \frac{d p \log p}{n},
    \end{equation*}
and    \begin{equation*}
   \big\|\widehat{\bTheta} - \bTheta ^ * \big\|_{\mathrm{F}} \lesssim \sqrt{\frac{d p \log p}{n}} \text{ \ \ with probability \ \ } 1 - O(p ^ {-10}).
\end{equation*}
\end{proof}

\subsection{Contraction in One Step of Iteration}\label{app:eta}

In this section, we focus on establishing the contraction property of the distance $\rho ^ 2(\cdot, \cdot)$ after one single iteration of non-convex gradient descent as proposed in Algorithm \ref{al:Divide_conquer} in the hub institution $\mathcal{S}_1$ for the bi-factored surrogate loss
\begin{equation*}
    \widetilde{\mathcal{L}} (\Ub, \Vb; \widehat \bTheta_0) = \mathcal{L}_S (\Ub \Vb \trans; \widehat \bTheta_0) + \mathcal{Q}(\Ub, \Vb),
\end{equation*}
where $\cL_S(\bTheta; \widehat \bTheta_0) = \mathcal{L}_1(\bTheta)  + \big\langle \nabla \mathcal{L} (\widehat\bTheta_0) - \nabla \mathcal{L}_1 (\widehat\bTheta_0), \bTheta \big \rangle$ and $\mathcal{Q}(\Ub, \Vb) = \frac{1}{4}\big\|\Ub\trans \Ub - \Vb \trans \Vb\big\|_{\mathrm{F}} ^ 2$. For the simplicity of notation, we use $\widetilde{\mathcal{L}} (\Ub, \Vb) = \widetilde{\mathcal{L}} (\Ub, \Vb; \widehat \bTheta_0)$ and $\mathcal{L}_S (\Ub \Vb \trans) = \mathcal{L}_S (\Ub \Vb \trans; \widehat \bTheta_0)$ in further deduction. Our results are demonstrated in Lemma \ref{lem:contraction_mapping}. In the proof of Lemma \ref{lem:contraction_mapping}, we consider the EHR-based phenotyping task with the true parameter matrix $\bTheta ^ * = \Ub^* \Vb^{*\top}$ and initial value $\widehat\bTheta_0= \Ub_0 \Vb_0 \trans$. A statistical error $e_{\mathrm{stat}}$ is defined as
\begin{equation} \label{equ:e_stat}
    e_{\mathrm{stat}} = \sup_{{\rm rank}(\bDelta) \leq d, \|\bDelta \|_{\mathrm{F}} \leq 1} \big\langle \nabla _{\bTheta} \mathcal{L}_S (\Ub^* \Vb^{*\top}), \bDelta \big\rangle.
\end{equation}

\begin{lemma}[Contraction in One Step of Contraction] \label{lem:contraction_mapping}
Suppose that Assumptions \ref{ass:Positive_minimal_eigenvalue}-\ref{ass:reg_grad} hold and $\rho^2(\Zb_\gamma, \Zb ^ *) \leq K(\Ub^*; \kappa_{\min}, \kappa_{\max})$ in the $\gamma$th iteration step of the non-convex gradient descent approach on a bi-factored surrogate loss $\widetilde{\mathcal{L}}$ denoted in \eqref{equ:surrogate_obj} with an initial value $\widehat\bTheta_0 = \Ub_0 \Vb_0\trans$ satisfying \eqref{equ:diff_initial}. Then for a stepsize $\eta$ satisfying
\begin{equation} \label{equ:eta_ineq}
    \eta \leq \frac{1}{12 \big\|\Zb_0 \big\|_{\rm op}^2} \min\Big\{\frac{1}{(\kappa_{\min} + \kappa_{\max})}, 1 \Big\},
\end{equation}
the outcome for the $(\gamma + 1)$th iteration step, $\Zb_{\gamma + 1} = \begin{bmatrix}
\Ub_{\gamma + 1} \\
\Vb_{\gamma + 1}
\end{bmatrix}$, holds that
\begin{equation} \label{equ:estat_iter}
    \rho^2(\Zb_{\gamma + 1}, \Zb^*) \leq \Big(1 - 2 \eta \kappa ^ * \sigma_d(\bTheta^*) \Big) \rho^2(\Zb_\gamma, \Zb^*) + \frac{6 \eta (\kappa_{\min} + \kappa_{\max}) e^2_{\mathrm{stat}}}{\kappa_{\min} \kappa_{\max}}.
\end{equation}
\end{lemma}

\begin{proof}[Proof of Lemma \ref{lem:contraction_mapping}]
First, we observe that Corollary \ref{col:surrogate_RSCRSS} guarantees the RSC/RSS condition for $\mathcal{L}_S (\cdot)$ for some sufficiently large sample size $n$. We then re-write $\rho^2(\Zb_{\gamma+1}, \Zb^*)$ as
    \begin{equation} \label{equ:rho_zzstar}
        \rho^2(\Zb_{\gamma+1}, \Zb^*) = \inf_{O \in \mathcal{O}(d)} \Bigg\|\begin{bmatrix} 
\Ub_{\gamma+1}  \\
\Vb_{\gamma+1} 
\end{bmatrix} - \begin{bmatrix}
\Ub^* \Ob \\
\Vb^* \Ob 
\end{bmatrix} \Bigg\| ^ 2_{\textrm{F}} \leq \Bigg\|\begin{bmatrix} 
\Ub_{\gamma+1} \\
\Vb_{\gamma +1}
\end{bmatrix} - \begin{bmatrix}
\Ub^* \widehat \Ob \\
\Vb^* \widehat \Ob 
\end{bmatrix} \Bigg\| ^ 2_{\textrm{F}},
    \end{equation}
for $\widehat \Ob \in \mathcal{O}(d)$ being the orthogonal matrix that attains the minimum of $\rho ^ 2 (\Zb_{\gamma}, \Zb ^ *)$. Furthermore, we apply the updating law
\begin{align*}
    \Ub_{\gamma+1}  &= \Ub_\gamma - \eta \nabla_{\Ub} \mathcal{L}_S \big(\Ub \Vb _{\gamma} \trans \big) \Big|_{\Ub = \Ub_\gamma } - \eta \nabla_{\Ub} \mathcal{Q} \big(\Ub, \Vb _\gamma \big) \Big|_{\Ub = \Ub_\gamma};\\
    \Vb_{\gamma +1}  &= \Vb_\gamma  - \eta \nabla_{\Vb} \mathcal{L}_S \big(\Ub_\gamma \Vb \trans \big) \Big|_{\Vb = \Vb _\gamma} - \eta \nabla_{\Vb} \mathcal{Q} \big(\Ub_\gamma, \Vb \big) \Big|_{\Vb = \Vb_\gamma}.
\end{align*}
For the simplicity of notation, we denote the updating law alternatively as
\begin{align*}
    \Ub_{\gamma + 1} &= \Ub_{\gamma}  - \eta \nabla_{\Ub} \mathcal{L}_S \big( \Ub _{\gamma} \Vb_{\gamma}\trans \big)  - \eta \nabla_{\Ub} \mathcal{Q} \big( \Ub_\gamma, \Vb_\gamma \big);\\
    \Vb_{\gamma +1} &= \Vb _ {\gamma}  - \eta \nabla_{\Vb} \mathcal{L}_S \big(\Ub_\gamma  \Vb _{\gamma} \trans \big)  - \eta \nabla_{\Vb} \mathcal{Q} \big(\Ub _\gamma,  \Vb _\gamma \big).
\end{align*}
Taking the expansion for the right-hand side of \eqref{equ:rho_zzstar} yields that
\begin{align}
\nonumber &   \rho^2(\Zb_{\gamma+1} , \Zb^*)  \leq  \rho^2(\Zb_\gamma, \Zb^*) \\ \nonumber
    & \quad \quad \quad \quad \quad \ \   - 2 \eta \Bigg\langle \begin{bmatrix} 
\Ub_\gamma  - \Ub^*\widehat{\Ob} \\ \nonumber
\Vb_\gamma  - \Vb^*\widehat{\Ob}
\end{bmatrix},
\begin{bmatrix} 
\nabla_{\Ub} \mathcal{L}_S \big(\Ub_\gamma  \Vb _{\gamma} \trans\big)  + \nabla_{\Ub} \mathcal{Q} \big( \Ub _\gamma,  \Vb _\gamma \big) \\ \nonumber
\nabla_{\Vb} \mathcal{L}_S \big( \Ub_\gamma  \Vb _{\gamma} \trans \big)  + \nabla_{\Vb} \mathcal{Q} \big( \Ub _\gamma,  \Vb _\gamma \big)
\end{bmatrix}
\Bigg \rangle \\  \label{eq:gd}
& \quad \quad \quad \quad \quad \ \  + \eta ^ 2 \Bigg\|\begin{bmatrix} 
\nabla_{\Ub} \mathcal{L}_S \big(\Ub_\gamma  \Vb _{\gamma} \trans \big)  + \nabla_{\Ub} \mathcal{Q} \big( \Ub _\gamma,  \Vb _\gamma\big) \\ 
\nabla_{\Vb} \mathcal{L}_S \big( \Ub_\gamma  \Vb _{\gamma} \trans \big)  + \nabla_{\Vb} \mathcal{Q} \big( \Ub _\gamma,  \Vb _\gamma \big)
\end{bmatrix} \Bigg\|_{\mathrm{F}} ^ 2 \\  \nonumber
    &\quad \quad \quad \quad \quad \ \  \leq  \rho^2(\Zb_\gamma , \Zb^*) - 2 \eta \underbrace{ \Bigg\langle \begin{bmatrix} 
\Ub_\gamma  - \Ub^*\widehat{\Ob} \\
\Vb_\gamma  - \Vb^*\widehat{\Ob}
\end{bmatrix},
\begin{bmatrix} 
\nabla_{\Ub} \mathcal{L}_S \big( \Ub_\gamma  \Vb _{\gamma} \trans \big) \\
\nabla_{\Vb} \mathcal{L}_S \big( \Ub_\gamma  \Vb _{\gamma} \trans \big) 
\end{bmatrix}
\Bigg \rangle }_{T_1} \\ \nonumber
& - 2 \eta \underbrace{\Bigg\langle \begin{bmatrix} 
\Ub_\gamma - \Ub^*\widehat{\Ob} \\ \nonumber
\Vb_\gamma - \Vb^*\widehat{\Ob}
\end{bmatrix},
\begin{bmatrix} 
\nabla_{\Ub} \mathcal{Q} \big( \Ub _\gamma,  \Vb _\gamma \big) \\ \nonumber
\nabla_{\Vb} \mathcal{Q} \big( \Ub _\gamma,  \Vb _\gamma \big)
\end{bmatrix}
\Bigg \rangle }_{T_2} + 2 \eta ^ 2 \underbrace{\Bigg\|\begin{bmatrix} 
\nabla_{\Ub} \mathcal{L}_S \big(\Ub_\gamma  \Vb _{\gamma} \trans \big) \\ \nonumber
\nabla_{\Vb} \mathcal{L}_S \big( \Ub_\gamma  \Vb _{\gamma} \trans \big) 
\end{bmatrix} \Bigg\|_{\mathrm{F}} ^ 2}_{T_3} + 2 \eta ^ 2 \underbrace{ \Bigg\|\begin{bmatrix} 
\nabla_{\Ub} \mathcal{Q} \big( \Ub _\gamma,  \Vb _\gamma \big) \\ \nonumber
\nabla_{\Vb} \mathcal{Q} \big( \Ub _\gamma,  \Vb _\gamma \big)
\end{bmatrix} \Bigg\|_{\mathrm{F}} ^ 2}_{T_4}.
\end{align}
We aim to bound the terms $T_1$ through $T_4$ and then aggregate the results. Heuristically, we need to show that $T_1$ and $T_2$ are large, while $T_3$ and $T_4$ being small. Under this scenario, we obtain an upper bound of $\rho^2(\Zb_{\gamma + 1}, \Zb^*)$. Technically, the equality
\begin{align*}
    \Bigg \langle\begin{bmatrix} 
\Ab \\
\Bb
\end{bmatrix}, \begin{bmatrix} 
\Cb \\
\Db
\end{bmatrix} \Bigg \rangle = \tr \Bigg(\begin{bmatrix} 
\Ab\Cb\trans & \Ab\Db\trans \\
\Bb\Cb\trans & \Bb\Db\trans
\end{bmatrix} \Bigg) = \tr (\Ab\Cb\trans + \Bb\Db\trans) = \tr (\Ab\Cb\trans) + \tr (\Bb\Db\trans)
\end{align*}
is helpful and frequently applied in further procedures.

For $T_1$, we have that
\begin{align*}
    T_1 & = \tr \Big(\nabla_{\bTheta} \mathcal{L}_S \big(\Ub_\gamma  \Vb _{\gamma} \trans  \big) \Vb_\gamma (\Ub_\gamma  - \Ub^* \widehat \Ob) \trans \Big) \\
    & + \tr \Big(\nabla_{\bTheta} \mathcal{L}_S \big(\Ub_\gamma  \Vb _{\gamma} \trans  \big) \trans \Ub_\gamma(\Vb_\gamma - \Vb^* \widehat \Ob) \trans \Big)\\
    & = \tr \Big(\nabla_{\bTheta} \mathcal{L}_S \big(\Ub_\gamma  \Vb _{\gamma} \trans  \big) \big( \Vb_\gamma (\Ub_\gamma - \Ub^* \widehat \Ob) \trans + (\Vb_\gamma - \Vb^* \widehat \Ob) \Ub_{\gamma} \trans \big) \Big)\\
    & = \tr \Big(\nabla_{\bTheta} \mathcal{L}_S \big(\Ub_\gamma  \Vb _{\gamma} \trans  \big) \big( (\Vb_\gamma - \Vb^* \widehat \Ob)(\Ub_\gamma  - \Ub^* \widehat \Ob) \trans - \Vb^* \Ub^{*\transpose} + \Vb_\gamma  \Ub_{\gamma}\trans \big) \Big)\\
    & = \underbrace{\tr \Big(\big(\nabla_{\bTheta} \mathcal{L}_S \big(\Ub_\gamma  \Vb _{\gamma} \trans  \big) - \nabla_{\bTheta} \mathcal{L}_S \big(\Ub^* \Vb^{*\transpose} \big) \big) \big(\Vb_\gamma   \Ub_{\gamma} \trans  - \Vb^* \Ub^{*\transpose} \big) \Big)}_{T_{1.1}} \\
    & + \underbrace{\tr \Big( \nabla_{\bTheta} \mathcal{L}_S \big(\Ub^* \Vb^{* \transpose} \big) \big(\Vb_\gamma  \Ub_{\gamma}\trans - \Vb^* \Ub^{*\transpose} \big) \Big)}_{T_{1.2}} \\
    & + \underbrace{\tr \Big( \nabla_{\bTheta} \mathcal{L}_S \big(\Ub_\gamma  \Vb _{\gamma} \trans  \big)  \big( (\Vb_\gamma  - \Vb^* \widehat \Ob)(\Ub _\gamma - \Ub^* \widehat \Ob) \trans \big) \Big)}_{T_{1.3}}.
\end{align*}
Combining the definition of $e_{\rm stat}$, we have the following bounds for $T_{1.1}$, $T_{1.2}$, and $T_{1.3}$, respectively. For $T_{1.1}$, we apply Theorem 2.1.12 of \cite{nestrov2013ineq} with $\mu = \frac{\kappa_{\min}}{2}$ and $L = \kappa_{\max{}} + \frac{\kappa_{\min}}{2}$, given the RSC/RSS condition shown in Lemma \ref{lem:RSC_RSS}. Consequently, we have that
\begin{align*}
     T_{1.1} & \geq \frac{\frac{1}{2} \kappa_{\min}\Big( \frac{1}{2} \kappa_{\min} + \kappa_{\max} \Big)}{\kappa_{\min} + \kappa_{\max}} \big\|\Ub_\gamma  \Vb _{\gamma} \trans  - \Ub^* \Vb^{*\transpose} \big\|_{\textrm{F}} ^ 2  \\
     & + \frac{1}{\kappa_{\min} + \kappa_{\max}} \big\|\nabla_{\bTheta} \mathcal{L}_S \big(\Ub_\gamma  \Vb _{\gamma} \trans \big) - \nabla_{\bTheta} \mathcal{L}_S \big(\Ub^* \Vb^{*\transpose} \big) \big\|_{\textrm{F}}^2 \\
    & \geq \frac{1}{4} \kappa_{\min} \big\|\Ub_\gamma  \Vb _{\gamma} \trans  - \Ub^* \Vb^{*\transpose} \big\|_{\textrm{F}} ^ 2 + \frac{1}{\kappa_{\min} + \kappa_{\max}} \big\|\nabla_{\bTheta} \mathcal{L}_S \big(\Ub_\gamma  \Vb _{\gamma} \trans  \big) - \nabla_{\bTheta} \mathcal{L}_S \big(\Ub^* \Vb^{*\transpose} \big) \big\|_{\textrm{F}}^2.
\end{align*}
For $T_{1.2}$, we apply the inequality $2ab \leq a^2 + b^2$ and have that
\begin{align*}
     T_{1.2} & \geq -e_{\rm stat} \big\|\Ub_\gamma  \Vb _{\gamma} \trans  - \Ub^* \Vb^{*\transpose} \big\|_{\textrm{F}}  \geq - \frac{3}{32} \kappa_{\min} \big\|\Ub_\gamma  \Vb _{\gamma} \trans  - \Ub^* \Vb^{*\transpose} \big\|_{\textrm{F}} ^ 2  - \frac{8}{3 \kappa_{\min}} e_{\rm stat}^2;
\end{align*}
For $T_{1.3}$, we consider applying the Cauchy-Schwarz inequality on Frobenius inner product, which yields that
\begin{align*}
     T_{1.3} & \geq - \Big| \big\langle  \nabla_{\bTheta} \mathcal{L}_S \big( \Ub ^ * \Vb ^ {* \transpose} \big),   (\Ub_\gamma - \Ub^* \widehat \Ob) (\Vb_\gamma - \Vb^* \widehat \Ob) \trans \big\rangle \Big| \\
     & - \Big| \big\langle \nabla_{\bTheta} \mathcal{L}_S \big(\Ub_\gamma  \Vb _{\gamma} \trans  \big) - \nabla_{\bTheta} \mathcal{L}_S \big(\Ub^* \Vb^{*\transpose} \big), (\Ub_\gamma - \Ub^* \widehat \Ob) (\Vb_\gamma - \Vb^* \widehat \Ob) \trans  \big\rangle  \Big| \\
    & \geq - \Big (e_{\rm stat} + \big\|\nabla_{\bTheta} \mathcal{L}_S \big(\Ub_\gamma  \Vb _{\gamma} \trans  \big) - \nabla_{\bTheta} \mathcal{L}_S \big(\Ub^* \Vb^{*\transpose} \big) \big\|_{\textrm{F}} \Big) \big \| (\Ub_\gamma - \Ub^* \widehat \Ob) (\Vb_\gamma - \Vb^* \widehat \Ob) \trans \big \|_{\mathrm{F}} \\
    & \overset{(i)}{\geq} - \frac{1}{2} \Big (e_{\rm stat} + \big\|\nabla_{\bTheta} \mathcal{L}_S \big(\Ub_\gamma  \Vb _{\gamma} \trans  \big) - \nabla_{\bTheta} \mathcal{L}_S \big(\Ub^* \Vb^{*\transpose} \big) \big\|_{\textrm{F}} \Big) \rho^2(\Zb_\gamma , \Zb^*) \\
    & \overset{(ii)}{\geq} - \Big (e_{\rm stat} + \big\|\nabla_{\bTheta} \mathcal{L}_S \big(\Ub_\gamma \Vb_\gamma^{\transpose}  \big) - \nabla_{\bTheta} \mathcal{L}_S \big(\Ub^* \Vb^{*\transpose} \big) \big\|_{\textrm{F}} \Big) \rho(\Zb_\gamma, \Zb^*)\sqrt{\frac{\kappa ^ * \sigma_d (\bTheta ^ *)}{5(\kappa_{\min} + \kappa_{\max})}} \\
  &\overset{(iii)}{\geq} - \frac{1}{2(\kappa_{\min} + \kappa_{\max})}\Big(e_{\rm stat}^2 + \big\|\nabla_{\bTheta} \mathcal{L}_S \big(\Ub_\gamma  \Vb _{\gamma} \trans  \big) - \nabla_{\bTheta} \mathcal{L}_S \big(\Ub^* \Vb^{*\transpose} \big) \big\|_{\textrm{F}}^2 \Big) - \frac{1}{5} \kappa^* \sigma_{d}(\bTheta^*) \rho^2(\Zb_\gamma, \Zb^*).
\end{align*}
Here $(i)$ holds since
\begin{equation*}
     \big \| (\Ub_\gamma - \Ub^* \widehat \Ob) (\Vb_\gamma - \Vb^* \widehat \Ob) \trans \big \|_{\mathrm{F}} \leq  \big \| \Ub_\gamma - \Ub^* \widehat \Ob \big\|_{\mathrm{F}} \big \| \Vb_\gamma - \Vb^* \widehat \Ob \big \|_{\mathrm{F}} \leq \frac{1}{2} \rho ^ 2(\Zb_\gamma, \Zb^*),
\end{equation*}
$(ii)$ holds by the condition such that $\rho^2(\Zb_\gamma, \Zb ^ *) \leq K(\Ub^*; \kappa_{\min}, \kappa_{\max})$ in the $\gamma$th iteration step, and $(iii)$ holds as a consequence of the inequality $ab \leq \frac{a^2}{2 \epsilon} + \frac{\epsilon b^2}{2} \text{ for any } a, b, \epsilon \geq 0$. We summarize the terms $T_{1.1}$, $T_{1.2}$, and $T_{1.3}$ as
\begin{align*}
    & T_1 = T_{1.1} + T_{1.2} + T_{1.3} \\
    &\geq \frac{5}{32} \kappa_{\min} \big\|\Ub_\gamma  \Vb _{\gamma} \trans  - \Ub^* \Vb^{*\transpose} \big\|_{\textrm{F}} ^ 2 + \frac{1}{2(\kappa_{\min} + \kappa_{\max})} \big\|\nabla_{\bTheta} \mathcal{L}_S \big(\Ub_\gamma  \Vb _{\gamma} \trans  \big) - \nabla_{\bTheta} \mathcal{L}_S \big(\Ub^* \Vb^{*\transpose} \big) \big\|_{\textrm{F}}^2  \\
    & - \frac{1}{5} \kappa^* \sigma_{d}(\bTheta ^ *) \rho ^ 2(\Zb_\gamma, \Zb^*) - \Big(\frac{8}{3 \kappa_{\min}} + \frac{1}{2(\kappa_{\min} + \kappa_{\max})}\Big) e_{\rm stat}^2.
\end{align*}
For $T_2$, we apply the chain rule for derivatives such that for $\bPi_\gamma = \Ub_{\gamma} \trans \Ub_\gamma - \Vb_{\gamma} \trans \Vb _\gamma$ and have that
\begin{align*}
    \nabla_{\Ub} \mathcal{Q} \big(\Ub_\gamma, \Vb_\gamma \big) = \Ub_{\gamma}  \nabla_{\bPi} \mathcal{Q} \big(\bPi _\gamma \big); \ \  \nabla_{\Vb} \mathcal{Q} \big(\Ub_\gamma, \Vb_\gamma  \big) = - \Vb_\gamma \nabla_{\bPi} \mathcal{Q} \big(\bPi_\gamma \big).
\end{align*}
We then apply Lemma B.1 of \cite{Park2016lowrank} such that 
\begin{align*}
    T_2 &\geq \underbrace{\frac{1}{8} \Big[ \big\|\Ub_\gamma \Ub_{\gamma} \trans - \Ub^* \Vb^{*\transpose} \big\|_{\textrm{F}} ^ 2  + \big\|\Vb_\gamma \Vb_{\gamma} \trans - \Ub^* \Vb^{* \transpose} \big\|_{\textrm{F}} ^ 2  - 2\big\|\Ub_\gamma \Vb_{\gamma} \trans  - \Ub^* \Vb^{*\transpose} \big\|_{\textrm{F}} ^ 2  \Big]}_{T_{2.1}} \\
    &+ \underbrace{\frac{1}{2} \big\| \nabla_{\bPi} \mathcal{Q} \big(\bPi_\gamma \big) \big \|_{\mathrm{F}}^2}_{T_{2.2}} - \underbrace{\frac{1}{2} \big\| \nabla_{\bPi} \mathcal{Q} \big(\bPi_\gamma  \big) \big \|_{\rm op} \Bigg\| \begin{bmatrix} 
\Ub_\gamma  - \Ub^*\widehat{\Ob} \\
\Vb_\gamma  - \Vb^*\widehat{\Ob}
\end{bmatrix} \Bigg\|_{\mathrm{F}}^2}_{T_{2.3}}.
\end{align*}
Here we notice that $T_{2.2}$ is no smaller than 0 thanks to its construction. We then apply Lemma 5.4 of \cite{tu2016lowrank} on both $T_{2.1}$ and $T_{2.3}$ and have that
\begin{align*}
    T_{2.1} &+ \frac{5 \kappa_{\min}}{32} \big\|\Ub_\gamma \Vb_{\gamma} \trans  - \Ub^* \Vb^{*\transpose} \big\|_{\textrm{F}} ^ 2  \geq \min \Big\{\frac{1}{8}, \frac{5 \kappa_{\min}}{128} \Big\} \Big[\big\|\Ub_\gamma \Ub_{\gamma} \trans  - \Ub^* \Vb^{*\transpose} \big\|_{\textrm{F}} ^ 2  \\
    &+ \big\|\Vb_\gamma \Vb_{\gamma} \trans - \Ub^* \Vb^{*\transpose} \big\|_{\textrm{F}} ^ 2 + 2 \big\|\Ub_{\gamma} \Vb_{\gamma} \trans  - \Ub^* \Vb^{*\transpose} \big\|_{\textrm{F}} ^ 2  \Big] \\
    &= \kappa^* \Big[ \big\|\Ub_\gamma \Ub_{\gamma} \trans  - \Ub^* \Vb^{*\transpose}  \big\|_{\textrm{F}} ^ 2 + \big\|\Vb_\gamma \Vb_{\gamma} \trans - \Ub^* \Vb^{*\transpose}  \big\|_{\textrm{F}} ^ 2  + 2 \big\|\Ub_{\gamma} \Vb_{\gamma} \trans  - \Ub^* \Vb^{*\transpose} \big\|_{\textrm{F}} ^ 2  \Big] \\
    &= \kappa^* \Bigg\|\begin{bmatrix} 
\Ub_\gamma \Ub_{\gamma} \trans & \Ub_\gamma \Vb_{\gamma} \trans \\
\Vb_\gamma \Ub_{\gamma} \trans & \Vb_\gamma \Vb_{\gamma} \trans
\end{bmatrix} - \begin{bmatrix} 
\Ub^* \Vb^{*\transpose} & \Ub^* \Vb^{*\transpose} \\
\Ub^* \Vb^{*\transpose} & \Ub^* \Vb^{*\transpose}
\end{bmatrix} \Bigg\|_{\mathrm{F}} ^ 2 \geq \frac{8}{5} \kappa ^ * \sigma_d(\bTheta^*)  \rho^2(\Zb_\gamma , \Zb^*); \\
    T_{2.3} &\leq \frac{1}{2} \big\| \nabla_{\bPi} \mathcal{Q} \big(\bPi_\gamma  \big) \big \|_{\rm op} \rho^2(\Zb_\gamma, \Zb^*) \overset{(i)}{\leq} \frac{1}{2}  \big\| \nabla_{\bPi} \mathcal{Q} \big(\bPi_\gamma  \big) \big \|_{\rm op} \rho(\Zb_\gamma, \Zb^*) \sqrt{\frac{4 \kappa ^ * \sigma_d(\bTheta ^*)}{5}} \\
    &\overset{(ii)}{\leq} \frac{1}{8} \big\| \nabla_{\bPi} \mathcal{Q} \big(\bPi_\gamma \big) \big \|_{\rm op}^2 + \frac{2}{5} \kappa^* \sigma_d(\bTheta^*) \rho^2(\Zb_\gamma, \Zb^*) \\
    &\leq \frac{1}{8} \big\| \nabla_{\bPi} \mathcal{Q} \big(\bPi_\gamma \big) \big \|_{\rm F}^2+ \frac{2}{5} \kappa^* \sigma_d(\bTheta^*) \rho^2(\Zb_\gamma, \Zb^*).
\end{align*}
Here $(i)$ holds thanks to $\rho^2(\Zb_\gamma, \Zb ^ *) \leq K(\Ub^*; \kappa_{\min}, \kappa_{\max})$ in the $\gamma$th iteration step, and $(ii)$ holds given the inequality $ab \leq \frac{a^2}{2 \epsilon} + \frac{\epsilon b^2}{2} \text{ for any } a, b, \epsilon \geq 0$, which is analogous to the approach of bounding $T_{1.3}$. Combining the terms of $T_1$ and $T_2$, we have that
\begin{align*}
    T_1 + T_2 &\geq  \frac{1}{2(\kappa_{\min} + \kappa_{\max})} \big\|\nabla_{\bTheta} \mathcal{L}_S \big(\Ub_\gamma \Vb_{\gamma} \trans \big) - \nabla_{\bTheta} \mathcal{L}_S \big(\Ub^* \Vb^{*\transpose} \big) \big\|_{\textrm{F}}^2  \\
    & +\kappa^* \sigma_d(\bTheta^*) \rho^2(\Zb_\gamma, \Zb^*) - \Big(\frac{8}{3 \kappa_{\min}} + \frac{1}{2(\kappa_{\min} + \kappa_{\max})}\Big) e_{\rm stat}^2 \\  
    &+ \frac{3}{8}  \big\| \nabla_{\bPi} \mathcal{Q} \big(\bPi_\gamma \big) \big \|_{\mathrm{F}}^2.
\end{align*}
For $T_3$, we consider the chain rule for derivatives such that
\begin{equation*}
    T_3 = \underbrace{ \big\|\nabla_{\bTheta}\mathcal{L}_S(\Ub_\gamma \Vb_{\gamma} \trans) \Vb_\gamma  \big\|_{\mathrm{F}} ^ 2}_{T_{3.1}} + \underbrace{ \big\|\Ub_\gamma  ^ { \transpose} 
 \nabla_{\bTheta}\mathcal{L}_S(\Ub_\gamma \Vb_{\gamma} \trans) \big\|_{\mathrm{F}} ^ 2 }_{T_{3.2}}.
\end{equation*}
For $T_{3.1}$ and $T_{3.2}$, we have
\begin{align*}
    T_{3.1} & \leq 2 \big\|\nabla_{\bTheta}\mathcal{L}_S(\Ub^* \Vb^{*\transpose}) \Vb_\gamma \big\|_{\mathrm{F}} ^ 2 + 2 \Big\| \big( \nabla_{\bTheta}\mathcal{L}_S(\Ub_\gamma \Vb_{\gamma} \trans) - \nabla_{\bTheta}\mathcal{L}_S(\Ub^* \Vb^{*\transpose}) \big) \Vb_\gamma \Big\|_{\mathrm{F}} ^ 2; \\
    T_{3.2} & \leq 2 \big\|\Ub _\gamma ^ { \transpose} \nabla_{\bTheta}\mathcal{L}_S(\Ub^* \Vb^{*\transpose}) \big\|_{\mathrm{F}} ^ 2 + 2 \Big\| \Ub _\gamma ^ {\transpose} \big( \nabla_{\bTheta}\mathcal{L}_S(\Ub_\gamma \Vb_{\gamma} \trans) - \nabla_{\bTheta}\mathcal{L}_S(\Ub^* \Vb^{*\transpose}) \big) \Big\|_{\mathrm{F}} ^ 2.
\end{align*}
Combining the upper bound for $T_{3.1}$ and $T_{3.2}$ implies that
\begin{align*}
    T_3 = T_{3.1} + T_{3.2} & \overset{(i)}{\leq} 2\Big(e_{\rm stat}^2 + \big \|\nabla_{\bTheta}\mathcal{L}_S(\Ub_\gamma \Vb_{\gamma} \trans) - \nabla_{\bTheta}\mathcal{L}_S(\Ub^* \Vb^{*\transpose}) \big \|_{\mathrm{F}}^2 \Big) \big\|\Vb_\gamma \big\|_{\rm op}^2 \\
    &+ 2\Big(e_{\rm stat}^2 + \big \|\nabla_{\bTheta}\mathcal{L}_S(\Ub_\gamma \Vb_{\gamma} \trans) - \nabla_{\bTheta}\mathcal{L}_S(\Ub^* \Vb^{*\transpose}) \big \|_{\mathrm{F}}^2\Big) \big\|\Ub_\gamma \big\|_{\rm op}^2 \\
    & \overset{(ii)}{\leq}  4 \Big (e_{\rm stat}^2 + \big \|\nabla_{\bTheta}\mathcal{L}_S(\Ub_\gamma \Vb_{\gamma} \trans) - \nabla_{\bTheta}\mathcal{L}_S(\Ub^* \Vb^{*\transpose}) \big \|_{\mathrm{F}}^2 \Big) \big\|\Zb_\gamma \big\|_{\rm op}^2.
\end{align*}
Here $(i)$ holds as $\|\Ab\Bb\|_{\mathrm{F}} \leq \|\Ab\|_{\rm op}\|\Bb\|_{\mathrm{F}}$ and
\begin{align*}
    \big\|\nabla_{\bTheta}\mathcal{L}_S(\Ub^* \Vb^{*\transpose}) \Vb_\gamma \big\|_{\mathrm{F}} & = \sup_{\| \Mb \|_{\mathrm{F}} = 1} \Big\langle \nabla_{\bTheta}\mathcal{L}_S(\Ub^* \Vb^{*\transpose}) \Vb_\gamma , \Mb \Big \rangle \\
    & = \sup_{\| \Mb \|_{\mathrm{F}} = 1} \Tr \Big(\nabla_{\bTheta}\mathcal{L}_S(\Ub^* \Vb^{*\transpose}) \Vb_\gamma  \Mb \trans \Big)  \leq e_{\rm stat} \big\|\Vb_\gamma \big\|_{\rm op}.
\end{align*}
Similarly, $ \Big\|\Ub_{\gamma} \trans  \nabla_{\bTheta}\mathcal{L}_S(\Ub^* \Vb^{*\transpose}) \Big\|_{\mathrm{F}} \leq e_{\rm stat} \big\|\Ub_\gamma  \big\|_{\rm op}$. Furthermore, $(ii)$ holds since $\big\|\Ub_\gamma \big\|_{\rm op}^2 \leq \big\|\Zb_\gamma  \big\|_{\rm op}^2$ and $\big\|\Vb_\gamma  \big\|_{\rm op}^2 \leq \big\|\Zb_\gamma \big\|_{\rm op}^2$.

For $T_4$, we have that
\begin{equation*}
    T_4 \leq \big\| \Ub_\gamma  \nabla_{\bPi} \mathcal{Q} \big(\bPi_\gamma  \big) \big \|_{\rm F}^2 + \big\| \Vb_\gamma \nabla_{\bPi} \mathcal{Q} \big(\bPi_\gamma \big) \big \|_{\rm F}^2 \leq 2 \big\|\Zb_\gamma  \big\|_{\rm op}^2 \big\| \nabla_{\bPi} \mathcal{Q} \big(\bPi_\gamma  \big) \big \|_{\rm F}^2.
\end{equation*}
We combine the bounds for $T_3$ and $T_4$, which indicates that
\begin{align*}
    T_3 + T_4 \leq  4 \Big (e_{\rm stat}^2 + \big \|\nabla_{\bTheta}\mathcal{L}_S(\Ub_\gamma \Vb_{\gamma} \trans ) & - \nabla_{\bTheta}\mathcal{L}_S(\Ub^* \Vb^{*\transpose}) \big \|_{\mathrm{F}}^2 \Big) \big\|\Zb_\gamma \big\|_{\rm op}^2 \\
    &\ \ \ + 2 \big\|\Zb_\gamma  \big\|_{\rm op}^2 \big\| \nabla_{\bPi} \mathcal{Q} \big(\bPi_\gamma \big) \big \|_{\rm F}^2.
\end{align*}
We assemble the bound from $T_1$ to $T_4$, which implies that for $\eta \leq \frac{1}{8 \|\Zb_{\gamma} \|_{\rm op}^2} \min{\Big\{1, \frac{1}{\kappa_{\min} + \kappa_{\max}} \Big\}}$,
\begin{align}
 \nonumber   \rho^2 \big(\Zb_{\gamma+1} , \Zb ^ * \big) &\leq \rho^2 \big(\Zb_\gamma, \Zb^* \big) - 2 \eta \big[T_1 + T_2 \big] + 2 \eta ^ 2 \big[T_3 + T_4 \big] \\  \nonumber  
    & \overset{(i)}{\leq} \big(1 - 2 \eta \kappa^* \sigma_d(\bTheta ^*) \big) \rho^2 \big(\Zb_\gamma  , \Zb^* \big) +  \eta \Big[\frac{2}{\kappa_{\min} + \kappa_{\max}} + \frac{16}{3 \kappa_{\min}} \Big] e_{\rm stat}^2 \\ \label{eq:one-step}
    & \leq \big(1 - 2 \eta \kappa^* \sigma_d(\bTheta^*) \big) \rho^2 \big(\Zb_\gamma, \Zb^* \big) +  6 \eta \frac{\kappa_{\min} + \kappa_{\max}}{\kappa_{\min} \kappa_{\max}} e_{\rm stat}^2.
\end{align}
In $(i)$, the selection of $\eta$ guarantees that the $\big\|\nabla_{\bTheta} \mathcal{L}_S \big(\Ub_\gamma \Vb_{\gamma} \trans \big) - \nabla_{\bTheta} \mathcal{L}_S \big(\Ub^* \Vb^{*\transpose} \big) \big\|_{\textrm{F}}^2$ term and the $\big\| \nabla_{\bPi} \mathcal{Q} \big(\bPi_\gamma \big) \big \|_{\rm F}^2$ term in $T_3 + T_4$ with positive coefficients are offset by such terms in $T_1 + T_2$ with negative coefficients. Finally, we show $\big\|\Zb_\gamma \big\|_{\rm op}^2 \leq 1.5 \big\|\Zb_0 \big\|_{\rm op}^2$ for any $\Zb_\gamma$ such that $\rho ^ 2 (\Zb_\gamma, \Zb ^ *) \leq K(\Ub^*; \kappa_{\min}, \kappa_{\max})$. Specifically,
\begin{align*}
    \big\|\Zb_\gamma \big\|_{\rm op} & \leq \big\|\Zb^ * \widehat{\Ob}\big\|_{\rm op}  + \big\|\Zb_\gamma - \Zb^ * \widehat{\Ob} \big\|_{\rm op} = \big\|\Zb^ *\big\|_{\rm op} + \rho (\Zb_\gamma, \Zb ^ *) \leq \big\|\Zb^ * \big\|_{\rm op} + \sqrt{\frac{4 \kappa ^ * \sigma_d(\bTheta ^*)}{5(\kappa_{\min} + \kappa_{\max})}} \\
    &\overset{(i)}{\leq} \big\|\Zb^ * \big\|_{\rm op} + \sqrt{\frac{\kappa_{\min} \kappa_{\max} \sigma^2_d (\Zb^*)}{64 (\kappa_{\min} + \kappa_{\max}) ^ 2}} \overset{(ii)}{\leq} \Big(1 + \frac{1}{16} \Big) \big\|\Zb^ * \big\|_{\rm op} = \frac{17}{16} \big\|\Zb^ * \big\|_{\rm op}.
\end{align*}

Here $(i)$ follows from the property that $\kappa ^ * \leq \frac{5\kappa_{\min} \kappa_{\max}}{128(\kappa_{\min} + \kappa_{\max})}$ and $\sigma^2_d (\Zb^*) = 2 \sigma_d(\bTheta ^*)$. $(ii)$ follows from the inequalities $\frac{xy}{(x+y) ^ 2} \leq 1/4$ and $\sigma_d (\Zb ^ *) \leq \big\|\Zb^ * \big\|_{\rm op}$. Applying a similar approach on the initial value $\Zb_0$ satisfying \eqref{equ:diff_initial} indicates that
\begin{equation*}
    \big\|\Zb_0 \big\|_{\rm op} \geq \big\|\Zb^ * \widehat{\Ob}\big\|_{\rm op}  - \big\|\Zb_0 - \Zb^ * \widehat{\Ob} \big\|_{\rm op} \geq \Big(1 - \frac{1}{16} \Big) \big\|\Zb^ * \big\|_{\rm op} = \frac{15}{16} \big\|\Zb^ * \big\|_{\rm op} .
\end{equation*}
We then conclude that $\big\|\Zb_0 \big\|_{\rm op} \geq \frac{15}{17} \big\|\Zb_\gamma \big\|_{\rm op}$, which completes our proof.
\end{proof}

\subsection{Bounding the Statistical Error in Bi-Factored Gradient Descent}

In this section, we propose an upper bound on the statistical error $e_{\rm stat}$ for bi-factored gradient descent through Lemma \ref{lem:epsilon_non-convex}, as inspired by \cite{jordan2019dist}. We use ${\rm mat} (\cdot)$ to denote the inverse operation of vectorization $\vec (\cdot)$.

\begin{lemma}[Statistical Error]\label{lem:epsilon_non-convex}
Under Assumptions \ref{ass:Positive_minimal_eigenvalue}-\ref{ass:reg_grad}, the statistical error $e_{\rm stat}$ in \eqref{equ:e_stat} has
\begin{equation*}
    e_{\mathrm{stat}} \lesssim \sqrt{\frac{dp \log p}{n}} + p ^ 3 \sqrt{\frac{\log p}{n/m}} \big\|\widehat \bTheta_0 - \bTheta ^ * \big\|_{\mathrm{F}} + p^3 \|\widehat \bTheta_0 - \bTheta ^ * \big\|_{\mathrm{F}} ^ 2,
\end{equation*}
 \text{with probability at least $1 - O(p^{-10})$.}
\end{lemma}

\begin{proof}
    First, we implement a second-order Taylor expansion for $\nabla _{\bTheta} \mathcal{L}_S (\bTheta ^ *; \widehat \bTheta_0)$, which implies that
    \begin{align}
     \nonumber   & \nabla_{\bTheta} \mathcal{L}_S (\bTheta ^ *; \widehat \bTheta_0) = \nabla \mathcal{L}_1(\bTheta ^ *) + \frac{1}{m} \sum_{i = 1} ^ m \nabla \mathcal{L}_i (\widehat \bTheta_0) - \nabla \mathcal{L}_1(\widehat \bTheta_0) \\
     \nonumber     & = \underbrace{\frac{1}{m} \sum_{i = 1} ^ m \nabla \mathcal{L}_i \big( \bTheta^* \big)}_{\Eb_1} + \underbrace{ {\rm mat} \Bigg\{ \bigg\{\frac{1}{m} \sum_{i = 1} ^ m \nabla^2 \mathcal{L}_i \big(\bTheta ^ * \big) \bigg\} \vec \big( \widehat{\bTheta}_0 - \bTheta^* \big) - \Big\{ \nabla^2 \mathcal{L}_1 \big(\bTheta ^* \big) \Big\} \vec \big(\widehat \bTheta_0 - \bTheta^* \big) \Bigg\} }_{\Eb_2} \\
     \nonumber   & + \mat \Bigg\{ \bigg\{ \int_0 ^ 1 \bigg[ \frac{1}{m} \sum_{i = 1} ^ m \nabla ^ 2 \mathcal{L}_i \big(\bTheta ^ * + s \big(\widehat \bTheta_0 - \bTheta ^ * \big) \big) - \frac{1}{m} \sum_{i = 1} ^ m  \nabla ^ 2 \mathcal{L}_i \big(\bTheta ^ *\big)  \bigg] ds \bigg\} \vec \big(\widehat \bTheta_0 - \bTheta ^ * \big) \\
        & \ \ \ \ \ \ \ \ \ \ \ \underbrace{ - \bigg\{ \int_0 ^ 1 \bigg[ \nabla ^ 2 \mathcal{L}_1 \big(\bTheta ^ * + s \big(\widehat \bTheta_0 - \bTheta ^ * \big) \big) - \nabla ^ 2 \mathcal{L}_1 \big(\bTheta ^ *\big)  \bigg] ds \bigg\} \vec \big(\widehat \bTheta_0 - \bTheta ^ * \big) \Bigg\}}_{\Eb_3}. \label{eq:E3}
    \end{align}
    The linearity of the trace yields that
    \begin{equation*}
        \tr \big( \bDelta \trans \nabla _{\bTheta} \mathcal{L}_S (\bTheta ^ *; \widehat \bTheta_0) \big) = \tr \Big( \bDelta \trans \big(\Eb_1 + \Eb_2 + \Eb_3 \big) \Big) = \tr \big( \bDelta \trans \Eb_1 \big) + \tr \big( \bDelta \trans \Eb_2 \big) + \tr \big( \bDelta \trans \Eb_3 \big).
    \end{equation*}
    We then apply the matrix Hölder's inequality \citep{baumgartner2011matholder} with dual index $p = 1$, $q = \infty$ for $\tr \big( \bDelta \trans \Eb_1 \big)$ and the Cauchy-Schwarz inequality for $\tr \big( \bDelta \trans \Eb_2 \big)$ and $\tr \big( \bDelta \trans \Eb_3 \big)$. For $\tr \big( \bDelta \trans \Eb_1 \big) \leq \| \bDelta \|_* \|\Eb_1 \|_{\rm op}$, we have $\|\bDelta\|_* \leq \sqrt{d} \|\bDelta\|_{\mathrm{F}} \leq \sqrt{d}$ thanks to the $d$-low-rankness of $\bDelta$ and $\|\bDelta\|_{\mathrm{F}} = 1$ as indicated by the definition of $e_{\rm stat}$. For $\Eb_1 = \nabla \mathcal{L} (\bTheta ^ *)$, we apply the matrix Bernstein inequality as given in \eqref{equ:bernstein_res}, which guarantees the existence of some $t \asymp \sqrt{\frac{p \log p}{n}}$ such that $ \|\Eb_1\|_{\rm op} \lesssim \sqrt{\frac{p \log p}{n}}$ with probability at least $1 - O(p ^ {-10})$.

    We then consider $\tr \big( \bDelta \trans \Eb_2 \big)$, where the definition of trace indicates that
    \begin{equation} \label{equ:E2_trace}
    \begin{split}
        \tr \big( \bDelta \trans \Eb_2 \big) & = \tr \bigg( \bDelta \trans \mat \Big\{ \Big( \widehat{\Hb} (\bTheta ^ *) - \widehat{\Hb}_1 (\bTheta ^ *)  \Big)\vec \big(\widehat \bTheta_0 - \bTheta ^ * \big) \Big\} \bigg) \\
        & = \vec (\bDelta) \trans \Big( \widehat{\Hb} (\bTheta ^ *) - \widehat{\Hb}_1 (\bTheta ^ *)  \Big)\vec \big(\widehat \bTheta_0 - \bTheta ^ * \big).
    \end{split}
    \end{equation}
    Applying Cauchy-Schwarz inequality on the right-hand side of \eqref{equ:E2_trace} indicates that
    \begin{align*}
         \tr \big( \bDelta \trans \Eb_2 \big) &\leq \big \| \vec(\bDelta) \big \|_2 \Big\|\widehat{\Hb} (\bTheta ^ *) - \widehat{\Hb}_1 (\bTheta ^ *) \Big\|_{\rm op}  \big \| \vec \big(\widehat \bTheta_0 - \bTheta ^ * \big) \big \|_2 \\
         &= \|\bDelta\|_{\rm F} \Big\|\widehat{\Hb} (\bTheta ^ *) - \widehat{\Hb}_1 (\bTheta ^ *) \Big\|_{\rm op} \big\|\widehat \bTheta_0 - \bTheta ^ *  \big\|_{\rm F} = \Big\|\widehat{\Hb} (\bTheta ^ *) - \widehat{\Hb}_1 (\bTheta ^ *) \Big\|_{\rm op} \big\|\widehat \bTheta_0 - \bTheta ^ *  \big\|_{\rm F}.
    \end{align*}
    The remaining task is to bound the operator norm $\Big\|\widehat{\Hb} (\bTheta ^ *) - \widehat{\Hb}_1 (\bTheta ^ *) \Big\|_{\rm op}$. Thanks to the symmetricity of $\widehat{\Hb} (\bTheta ^ *) - \widehat{\Hb}_1 (\bTheta ^ *)$, we have
    \begin{equation*}
        \Big\|\widehat{\Hb} (\bTheta ^ *) - \widehat{\Hb}_1 (\bTheta ^ *) \Big\|_{\rm op} = \sup_{\bgamma \in \mathbb{S} ^ {p ^ 2 - 1}} \bigg| \bgamma \trans \Big(\widehat{\Hb} (\bTheta ^ *) - \widehat{\Hb}_1 (\bTheta ^ *) \Big) \bgamma \bigg|.
    \end{equation*}
    Applying Lemma \ref{lem:H_toSigma} and \eqref{equ:sigma_maxnorm} for a reformed matrix $\bGamma = \mat \big( \bgamma \big) \in \mathbb{R} ^ {p \times p}$ yields that
    \begin{equation} \label{equ:ineq_sigma1}
        \begin{split}
            & \vec{(\bGamma)} \trans \big( \widehat{\Hb} (\bTheta ^ *) - \widehat{\Hb}_1 (\bTheta ^ *)  \big) \vec{(\bGamma)} \\
         = \ & \vec{(\Gb)} \trans \big( \widehat{\bSigma} (\bTheta ^ *) - \widehat{\bSigma}_1 (\bTheta ^ *)  \big) \vec{(\Gb)} \leq \|\bGamma\|_*^2 \big\| \widehat{\bSigma} (\bTheta ^ *) - \widehat{\bSigma}_1 (\bTheta ^ *)  \big\|_{\max},
        \end{split}
    \end{equation}
    with $\Gb$ denoting the diagonal matrix in the SVD of $\bGamma$ and
    \begin{equation*}
    \widehat{\boldsymbol\Sigma}_i(\bTheta) = \frac{1}{n/m} \sum_{\ell \in \mathcal{S}_i} \bigg\{ \sum_{r = 1} ^ p \mathrm{vec} \big(\widetilde{\mathbf{Y}}_{(-r)}^{(\ell)} (\bTheta) \big) \mathrm{vec} \big(\widetilde{\mathbf{Y}}_{(-r)}^{(\ell)} (\bTheta) \big) \trans \bigg\}.
    \end{equation*}
    Similar to \eqref{equ:concentration_Sigma}, we have
    \begin{equation*}
         \mathbb{P} \bigg(\Big|\big(\widehat\bSigma_{1}\big)_{jk}(\bTheta ^ *) - \widehat{\bSigma}_{jk}(\bTheta ^ *) \Big| > t, \text{ for any } j,k \in p ^ 2 \bigg) \lesssim p^4 \exp \bigg\{-\frac{t^ 2 n}{8 m p ^ 4} \bigg\}.
    \end{equation*}
    Setting $t = p ^ 2 \sqrt{\frac{\alpha \log p}{n/m}}$ for some constant $\alpha$ large enough indicates that
    \begin{equation*}
        \big\| \widehat{\bSigma} (\bTheta ^ *) - \widehat{\bSigma}_1 (\bTheta ^ *)  \big\|_{\max} \lesssim p^2 \sqrt{\frac{\log p}{n/m}} \ \ \text{with probability at least} \ \ 1 - O(p^{-10}).
    \end{equation*}
    Given that $\|\bGamma\|_*^2 \leq p \|\bGamma\|_{\rm F}^2 = p \|\vec{(\bGamma)}\|_2^2 = p \|\bgamma\|_2 ^ 2 = p$, we conclude from \eqref{equ:ineq_sigma1} that
    \begin{equation} \label{equ:H_H1}
    \begin{split}
        \tr \big( \bDelta \trans \Eb_2 \big) & \leq \Big\| \widehat{\Hb} (\bTheta ^ *) - \widehat{\Hb}_1 (\bTheta ^ *) \Big\|_{\rm op} \big\|\widehat \bTheta_0 - \bTheta ^ *  \big\|_{\rm F} \\
        & \leq p \Big\| \widehat{\bSigma} (\bTheta ^ *) - \widehat{\bSigma}_1 (\bTheta ^ *)  \Big\|_{\max} \big\|\widehat \bTheta_0 - \bTheta ^ *  \big\|_{\rm F} \lesssim p^3 \sqrt{\frac{\log p}{n/m}} \big\|\widehat \bTheta_0 - \bTheta ^ *  \big\|_{\rm F},
    \end{split}
    \end{equation}
    with probability at least $1 - O(p ^ {-10})$.

    For $\Eb_3$, we consider that for some $\mathring{\bTheta} = \bTheta ^ * + s (\widehat \bTheta_0 - \bTheta ^ *)$ with $s \in [0,1]$,
    \begin{align*}
        \big\| \widehat{\Hb} (\mathring{\bTheta}) - \widehat{\Hb} (\bTheta ^ *) \big\|_{\rm op} &\leq p^2 \big\| \widehat{\Hb} (\mathring{\bTheta}) - \widehat{\Hb} (\bTheta ^ *) \big\|_{\rm max}; \\
        \big\| \widehat{\Hb}_1 (\mathring{\bTheta}) - \widehat{\Hb}_1 (\bTheta ^ *) \big\|_{\rm op} &\leq p^2 \big\| \widehat{\Hb}_1 (\mathring{\bTheta}) - \widehat{\Hb}_1 (\bTheta ^ *) \big\|_{\rm max},
    \end{align*}
    and alternatively, choose to bound the entry-wise maximum norm of $\widehat{\Hb}(\mathring{\bTheta}) - \widehat{\Hb} (\bTheta ^ *)$ and $\widehat{\Hb}_1 (\mathring{\bTheta}) - \widehat{\Hb}_1 (\bTheta ^ *)$. Here we consider $\widehat{\Hb}(\cdot)$ as presented in \eqref{equ:hess_mat}. For some fixed $\ell$, we denote $\Rb_\ell(\bTheta) \in \mathbb{R} ^ {p^2 \times p^2}$ as
    \begin{equation} \label{equ:def_Rl}
        \Rb_\ell(\bTheta) = \sum_{r = 1} ^ p  A_{r,\ell}(\bTheta)  \vec \big(\mathbf{Y}_{(-r)}^{(\ell)} \big)   \vec \big(\mathbf{Y}_{(-r)}^{(\ell)} \big) \trans; \text{ where } \widehat{\Hb} (\bTheta) = \frac{1}{n} \sum_{\ell = 1} ^ n \Rb_\ell(\bTheta).
    \end{equation}
    The design of $\mathbf{Y}_{(-r)}^{(\ell)}$ ensures that
    \begin{equation*}
        \big\| \Rb_\ell(\bTheta) \big\|_{\max} = \max_{k,m} \Big\{ \big(\Rb_\ell(\bTheta)\big)_{km} \Big\} \leq \max_{i \neq j} \Big\{ 4\big| A_{i,\ell}(\bTheta) + A_{j, \ell} (\bTheta) \big| \Big\}.
    \end{equation*}
    Note that the function $g(t) = 4 e^t/(e^t + 1) ^ 2$ has Lipschitz property, we conclude that
    \begin{equation} \label{equ:Hhat_maxnorm}
    \begin{split}
        \big\| \widehat{\Hb} (\mathring{\bTheta}) - \widehat{\Hb} (\bTheta ^ *) \big\|_{\rm max} &= \bigg \|\frac{1}{n} \sum_{\ell = 1} ^ n \big(\Rb_\ell(\mathring{\bTheta}) - \Rb_\ell(\bTheta ^ *) \big)\bigg\|_{\max} \leq \frac{1}{n} \sum_{\ell = 1} ^ n \big\| \Rb_\ell(\mathring{\bTheta}) - \Rb_\ell(\bTheta ^ *) \big\|_{\max} \\
        & \leq \frac{4}{n} \sum_{\ell = 1} ^ n \max_{ i \neq j}\bigg\{\Big|\Big(A_{i,\ell}(\widetilde\bTheta) - A_{i,\ell}(\bTheta ^ *) \Big) + \Big( A_{j,\ell}(\widetilde\bTheta) - A_{j,\ell}(\bTheta ^ *) \Big) \Big|\bigg\} \\
        & \lesssim \frac{4}{n} \sum_{\ell = 1} ^ n  \bigg|\sum_{s \neq r} \Big(\widetilde{\theta}_{rs} - \theta ^ *_{rs} \Big) x_{r}^{(\ell)} x_{s}^{(\ell)} \bigg| \\
        & \leq \frac{4}{n} \sum_{\ell = 1} ^ n \big\|\widetilde\bTheta - \bTheta ^ * \big\|_{\max} \sum_{s \neq r} |x_{r}^{(\ell)} x_{s}^{(\ell)} | \leq 4p \big\|\widetilde\bTheta - \bTheta ^ * \big\|_{\max}.
    \end{split}
    \end{equation}
    Furthermore, we have that $\big\| \widehat{\Hb}_1 (\mathring{\bTheta}) - \widehat{\Hb}_1 (\bTheta ^ *) \big\|_{\rm max} \lesssim p \big\|\widetilde\bTheta - \bTheta ^ * \big\|_{\max}$ similar to \eqref{equ:Hhat_maxnorm}. Finally, we apply the mean value theorem of $\Eb_3$, which yields that 
    \begin{align}
     \nonumber   \tr \big( \bDelta \trans \Eb_3 \big) &\leq \|\bDelta\|_{\rm F} \|\Eb_3\|_{\rm F} \leq \|\Eb_3\|_{\rm F} = \big\|\vec{(\Eb_3)}\big\|_{2} 
        \\\nonumber &\leq \bigg\| \bigg\{ \int_0 ^ 1 \bigg[ \widehat\Hb \big(\bTheta ^ * + s \big(\widehat \bTheta_0 - \bTheta ^ * \big) \big) - \widehat\Hb \big(\bTheta ^ *\big)  \bigg] ds \bigg\} \vec \big(\widehat \bTheta_0 - \bTheta ^ * \big) \bigg \|_2 \\
        \nonumber    &\qquad + \bigg\| \bigg\{ \int_0 ^ 1 \bigg[ \widehat\Hb_1 \big(\bTheta ^ * + s \big(\widehat \bTheta_0 - \bTheta ^ * \big) \big) - \widehat\Hb_1 \big(\bTheta ^ *\big)  \bigg] ds \bigg\} \vec \big(\widehat \bTheta_0 - \bTheta ^ * \big) \bigg\|_2 \\
        \nonumber & \leq \Big( \big\| \widehat{\Hb} (\mathring{\bTheta}) - \widehat{\Hb} (\bTheta ^ *) \big\|_{\rm op} +  \big\| \widehat{\Hb}_1 (\mathring{\bTheta}) - \widehat{\Hb}_1 (\bTheta ^ *) \big\|_{\rm op} \Big) \big\|\widehat \bTheta_0 - \bTheta ^ * \big\|_{\rm F} \\
        \nonumber  & \leq p^2 \Big( \big\| \widehat{\Hb} (\mathring{\bTheta}) - \widehat{\Hb} (\bTheta ^ *) \big\|_{\rm max} +  \big\| \widehat{\Hb}_1 (\mathring{\bTheta}) - \widehat{\Hb}_1 (\bTheta ^ *) \big\|_{\rm max} \Big) \big\|\widehat \bTheta_0 - \bTheta ^ * \big\|_{\rm F} \\
        \label{eq:traceE3} & \lesssim p ^ 3 \big\|\widehat\bTheta_0 - \bTheta ^ * \big\|_{\max} \big\|\widehat \bTheta_0 - \bTheta ^ * \big\|_{\rm F} \leq p ^ 3 \big\|\widehat \bTheta_0 - \bTheta ^ * \big\|_{\rm F}^2.
    \end{align}
    Summarizing the upper bounds for $\tr \big( \bDelta \trans \Eb_1 \big)$, $\tr \big( \bDelta \trans \Eb_2 \big)$, and $\tr \big( \bDelta \trans \Eb_3 \big)$ completes the proof for Lemma \ref{lem:epsilon_non-convex}.
\end{proof}

\section{Proof for the Rationality for Initialization}
\label{sec: convex thm}
In this section, we prove the theoretical performance guarantee for the estimator $\widehat{\bTheta}_{\rm cvx}$ for the convex problem proposed in \eqref{equ:question}, which shows our rationality of initialization.

\begin{proof}[Proof of Theorem \ref{thm:main}] For the convenience of the notation, we prove the rate by replacing the loss $\mathcal{L}_1$ by $\mathcal{L}$.
Our proof has two major steps: in  Step 1, we show the scale of the gradient $\big\|\nabla \mathcal{L}(\boldsymbol\Theta ^ *)\big\|_{\text{op}}$, and in {Step 2}, let $\Deltab = \widehat{\bTheta}_{\rm cvx} - {\bTheta}^*$ and we  aim to  show the  restricted strongly convexity of the loss by providing the lower bound of $\vec (\boldsymbol\Delta) \trans   \widehat{\mathbf{H}}(\boldsymbol\Theta)   \vec(\boldsymbol\Delta)$ for $\bTheta \in \mathcal{T}$ given Assumptions \ref{ass:Positive_minimal_eigenvalue}-\ref{ass:reg_grad}. 

For Step 1, by Assumption \ref{ass:reg_grad}, we can apply the matrix Bernstein inequality and get
\begin{equation} \label{equ:bernstein_res}
    \mathbb{P} \bigg\{ \bigg\|\frac{1}{n} \sum_{\ell = 1} ^ n \mathbf{W}_{\ell} (\bTheta ^ *)\bigg\|_{\rm op} \geq t \bigg\} \leq 2 p \exp \bigg\{-\frac{nt^2/2}{p + pt/3} \bigg\}.
\end{equation}
Here we set that $t \asymp \sqrt{\frac{\gamma p \log p}{n}}$ with some $\gamma > 1$, and we have that
\begin{align*}
    \mathbb{P} \bigg\{ \bigg\|& \frac{1}{n} \sum_{\ell = 1} ^ n \mathbf{W}_{\ell} (\bTheta ^ *) \bigg\|_{\rm op} \geq C_2 \sqrt{\frac{\gamma p \log p}{n}} \bigg\} \\
    &\leq 2 p \exp \bigg\{- \frac{C_2 ^ 2 \gamma p \log p}{2p + 2 p C_2 \sqrt{\frac{\gamma p \log p}{n}} / 3} \bigg\} = 2 p ^ {1 - \frac{C ^ 2 _ 2 \gamma}{2 + 2 C_2 \sqrt{\frac{\gamma p \log p}{n}} / 3}}  \leq 2 p ^ {-10},
\end{align*}
for some constant $C_2 = \gamma = 10$, where the last inequality is due to $n \geq p \log p$. Therefore, if we choose $\lambda = C_0 \sqrt{\frac{\gamma p \log p}{n}}$ for sufficiently large $C_0$, we have $\lambda \ge 2  \big\|\nabla \mathcal{L}(\boldsymbol\Theta ^ *)\big\|_{\text{op}}$ with probability at least $1-2 p ^ {-10}$. By \cite{negahban2012RSC}, we have
\begin{equation} \label{equ:cone_property}
    \Deltab \in \mathcal{C}(\mathcal{M}, \overline{\mathcal{M}} ^ \perp, \boldsymbol\Theta ^ *) := \bigg\{\boldsymbol \Delta \in \mathbb{R} ^ {p \times p}; \big\|\boldsymbol\Delta_{\overline{\mathcal{M}} ^ \perp} \big\|_* \leq 3 \big \|\boldsymbol\Delta_{\overline{\mathcal{M}}} \big\|_* + 4 \sum_{j \geq d + 1} \sigma_j(\boldsymbol\Theta ^ *) \bigg\},
\end{equation}
for $\mathcal{M}$ being the span of the main direction of the SVD of $\boldsymbol\Theta^*$ with the largest $d$ singular values, and $\overline{\mathcal{M}} ^ \perp$ denotes the orthogonal complement of $\mathcal{M}$. As  $\text{rank} (\bTheta ^ *) = d$, we have that
\begin{align*}
    \|\bDelta\|_* \leq \big\|\bDelta_{\mathcal{M}_d} \big\|_* + \big\|\bDelta_{\overline{\mathcal{M}}_d^\perp} \big\|_* \leq 4 \big \|\bDelta_{\mathcal{M}_d} \big \|_* + 4 \sum_{j \geq d + 1} \lambda_j (\bTheta ^ *).
\end{align*}
Given that $\lambda_j(\bTheta ^ *) = 0$ where $j \geq d + 1$, we have $\|\bDelta\|_* \leq 4 \|\bDelta_{\mathcal{M}_d}\|_* \leq 4 \sqrt{2d} \|\bDelta\|_{\mathrm{F}}$, and thus
\begin{align}\label{eq:nuclear-frob}
    C_1 p^2 \sqrt{\frac{\log p}{n}} \|\bDelta \|_*^2 \leq 32C_1 d p ^2 \sqrt{\frac{\log p}{n}} \|\bDelta \|_{\mathrm{F}} ^ 2 \leq \frac{1}{2} \kappa_{\min} \|\bDelta \|^2_{\mathrm{F}} \end{align}
for sufficiently large $n$.

In Step 2, we verify the restricted strong convexity. We  first evaluate $\vec(\boldsymbol\Delta) \trans   \widehat{\mathbf{H}}(\boldsymbol\Theta) \vec (\boldsymbol\Delta)$ for $\bDelta \in \mathbb{R} ^ {p \times p}$ on $\bTheta \in \mathcal{T}$. We denote the SVD of $\bDelta$ as $\Ub \Db \Vb\trans$. Applying Lemma \ref{lem:H_toSigma}, we conclude from \eqref{equ:sigma_sigmahat} that for any $\bTheta \in \mathcal{T}$,
\begin{equation*} 
\vec (\boldsymbol\Delta) \trans   \widehat{\mathbf{H}}(\boldsymbol\Theta)   \vec (\boldsymbol\Delta) = \vec (\mathbf{D}) \trans   \bigg\{ \frac{1}{n} \sum_{\ell = 1} ^ n \sum_{r = 1} ^ p \vec  \big(\widetilde{\mathbf{Y}}_{(-r)}^{(\ell)} (\bTheta) \big) \vec  \big(\widetilde{\mathbf{Y}}_{(-r)}^{(\ell)} (\bTheta) \big) \trans \bigg\}   \vec (\mathbf{D}).
\end{equation*} 
Here $\widetilde{\Yb}_{(-r)} ^ {(\ell)}(\bTheta) = \sqrt{A_{r, \ell}(\bTheta)} \Ub \trans \Yb_{(-r)} ^ {(\ell)} \Vb$.  For any arbitrary $\ell \in [n]$, we define $\widehat{\bSigma}(\bTheta)$ as
\begin{equation*}
    \widehat{\bSigma}(\bTheta) := \frac{1}{n} \sum_{\ell = 1} ^ n \Zb_\ell, \text{ where } \Zb_\ell := \sum_{r = 1} ^ p \vec{}\big(\widetilde{\Yb}_{(-r)} ^ {(\ell)}(\bTheta) \big) \vec{}\big(\widetilde{\Yb}_{(-r)} ^ {(\ell)}(\bTheta) \big) \trans.
\end{equation*}

We then replace $\vec (\bDelta) \trans   \widehat{\mathbf{H}}(\bTheta)   \vec (\bDelta)$ with $\vec (\mathbf{D}) \trans   \widehat{\boldsymbol\Sigma } (\bTheta) \vec (\mathbf{D})$ by applying Lemma \ref{lem:H_toSigma}, and split $\widehat{\bSigma}(\bTheta)$ by $\widehat{\bSigma}(\bTheta) - \bSigma(\bTheta)$ and $\bSigma(\bTheta)$, where $\bSigma(\bTheta) = \mathbb{E}[\widehat{\bSigma}(\bTheta)]$. Consequently, we have
\begin{equation} \label{equ:two_terms}
     \vec (\mathbf{D}) \trans   \widehat{\boldsymbol\Sigma} (\bTheta) \vec (\mathbf{D}) = \underbrace{\vec (\mathbf{D}) \trans   {\boldsymbol\Sigma}(\bTheta)   \vec (\mathbf{D})}_{\mathcal{P}_1} + \underbrace{\vec (\mathbf{D}) \trans   (\widehat{\boldsymbol\Sigma} (\bTheta) - {\boldsymbol\Sigma}(\bTheta))  \vec (\mathbf{D})}_{\mathcal{P}_2}.
\end{equation}
Our remaining objective is to show the RSC property at $\boldsymbol\Theta \in \mathcal{T}$ by bounding the terms $\mathcal{P}_1$ and $\mathcal{P}_2$ in \eqref{equ:two_terms}. $\mathcal{P}_1$ is bounded from below with $\vec (\mathbf{D}) \trans \ {\boldsymbol\Sigma}(\bTheta) \ \vec (\mathbf{D}) \geq \kappa_{\text{min}} \|\bDelta\|_{\text{F}}^2$ thanks to Assumption \ref{ass:Positive_minimal_eigenvalue} which guarantees a minimal eigenvalue. To bound $\big| \mathcal{P}_2 \big|$ from above, we consider that
\begin{equation}  \label{equ:sigma_maxnorm}
     \vec (\mathbf{D}) \trans  (\widehat{\boldsymbol\Sigma}(\bTheta) - {\boldsymbol\Sigma}(\bTheta))  \vec (\mathbf{D}) \leq \|\boldsymbol\Delta\|_*^2 \big\|\widehat{\bSigma}(\bTheta) - \bSigma (\bTheta) \big\|_{\rm max},
\end{equation}
where we need to bound the entrywise maximum norm $\big\|\widehat{\bSigma}(\bTheta) - \bSigma (\bTheta) \big\|_{\max}$ instead. Furthermore, $A_{r,\ell}(\bTheta) \in [0,1]$ for any arbitrary $\bTheta$ given its closed form shown in \eqref{equ:hess_mat}. Therefore, we only need to calculate each entry of $\Ub \trans \Yb_{(-r)} ^ {(\ell)} \Vb$ for some $\ell \in [n]$, take summation with respect to $r \in [p]$, then obtain a uniform upper bound on the entrywise max norm of $\Zb_{\ell}$. We first calculate $\bPhi ^ {(\ell)}_{(-r)} := \Ub \trans \Yb_{(-r)} ^ {(\ell)}$ for some arbitrary $\ell \in [n]$ and all $r \in [p]$. For $\bPhi ^ {(\ell)}_{(-r)}$, matrix multiplication implies that
\begin{equation}\label{equ:phi}
    \phi ^ {(\ell); -r}_{tk} := \big(\bPhi ^ {(\ell)}_{(-r)} \big)_{tk} = \begin{cases}
			2\sum_{j = 1; j \neq k} ^ p x_j ^ {(\ell)}  \Ub_{jt} + 2 \Ub_{kt} , & \text{if $k = r$;}\\
          2  x_k^{(\ell)}  \Ub_{rt}, & \text{otherwise.}
		 \end{cases}
\end{equation}
Then we have $\bPsi_{(-r)} ^ {(\ell)} := \Ub \trans \Yb_{(-r)} ^ {(\ell)} \Vb = \bPhi ^ {(\ell)}_{(-r)} \Vb$ as
\begin{align} \label{equ:psi}
  \nonumber  \big(\bPsi ^ {(\ell)}_{(-r)} \big)_{ts} &= \sum_{w = 1} ^ p \phi^{(\ell); -r}_{tw} \Vb_{ws} = \phi^{(\ell); -r}_{tr} \Vb_{rs} + \sum_{w = 1; w \neq r} ^ p \phi^{(\ell); -r}_{tw} \Vb_{ws} \\
  & = 2\Vb_{rs} \sum_{j = 1; j \neq r} ^ p x_j ^ {(\ell)}  \Ub_{jt}  + 2  \Vb_{rs}\Ub_{rt}+  2\Ub_{rt} \sum_{w = 1; w \neq r} ^ p x_w^{(\ell)}  \Vb_{ws}.
\end{align}
As $x_j^{(\ell)} \in \{-1,1\}$, we have for all $t,s, \ell, r$,
\begin{align*}
    \Big| \big(\bPsi ^ {(\ell)}_{(-r)} \big)_{ts} \Big| & \leq 2 \sqrt{p}(\|\Vb_{s}\|_\infty  \|\Ub_t\|_2 + \|\Ub_t\|_\infty \|\Vb_s\|_2 )\leq 4 \sqrt{p}.
\end{align*}
As $\sqrt{A_{r, \ell}(\bTheta)} \in [0,1]$, we have  $\Big\|\widetilde{\Yb}_{(-r)} ^ {(\ell)}(\bTheta) \Big\|_{\max}= \Big\|\sqrt{A_{r, \ell}(\bTheta)}\bPsi ^ {(\ell)}_{(-r)}\Big\|_{\max} \le 4 \sqrt{p}$ and $\Zb_\ell \le 16p ^ 2$. Therefore, we can apply Hoeffding's inequality and the union bound and  have that
\begin{equation*}
    \mathbb{P} \bigg(\Big|\bSigma_{ij}(\bTheta) - \widehat{\bSigma}_{ij}(\bTheta) \Big| > t, \text{ for any } i,j \in [p ^ 2] \bigg) \leq 2 p^4 \exp \bigg\{-\frac{t^ 2 n}{8 p ^ 4} \bigg\}.
\end{equation*}
Choosing $t = p ^ 2 \sqrt{\frac{\alpha \log p}{n}}$ for  $\alpha = 112$, we have
\begin{equation} \label{equ:concentration_Sigma}
    \mathbb{P}\Big(\big\|\bSigma(\bTheta) - \widehat{\bSigma}(\bTheta)\big\|_{\max{}} > t \Big) \leq 2p^{-10}.
\end{equation}
Consequently, there exists a constant $C_1 = \sqrt{112}$ such that $\big|\mathcal{P}_2\big| \leq C_1 p ^ 2 \sqrt{\frac{ \log p}{n}}  \| \bDelta \|_* ^ 2$ with probability at least $1 - 2p^{-10}$. Combining the bounds for $\mathcal{P}_1$ and $\mathcal{P}_2$ yields that

\begin{equation*}
    \vec{(\bDelta)} \trans \widehat{\Hb}(\bTheta) \vec{(\bDelta)} \geq \kappa_{\min{}} \|\bDelta \|_{\mathrm{F}}^2 - C_1 p ^ 2 \sqrt{\frac{ \log p}{n}}  \| \bDelta \|_* ^ 2
\end{equation*}
with probability at least $1 - 2p^{-10}$. 
Combining with \eqref{eq:nuclear-frob}, under the condition $n \geq C d^2 p^4 \log p$ for some constant $C$, $\vec (\boldsymbol\Delta) \trans  \widehat{\mathbf{H}}(\boldsymbol\Theta)  \vec (\boldsymbol\Delta) \geq \kappa_{\min}/2  \cdot \|\bDelta\|_{\mathrm{F}}^2$ for any $\bTheta \in \mathcal{T}$. And similarly, we have with probability at least $1 - 2p^{-10}$,
\begin{equation}
    \label{eq:upper_bound_H}
    \vec{(\bDelta)} \trans \widehat{\Hb}(\bTheta) \vec{(\bDelta)} \leq \kappa_{\max} \|\bDelta \|_{\mathrm{F}}^2 + C p ^ 2 \sqrt{\frac{ \log p}{n}}  \| \bDelta \|_* ^ 2     
\end{equation}

This demonstrates the RSC condition for $\widehat{\mathbf{H}}(\boldsymbol\Theta)$.
By Theorem 2 of \cite{jianqing2017nuclear}, we have 
\begin{align*}
    \big\|\widehat{\bTheta}_{\rm cvx} - \bTheta ^ * \big\|_\mathrm{F} \lesssim \sqrt{\frac{d p \log p}{n}} \text{ , and }
    \big \|\widehat{\bTheta}_{\rm cvx} - \bTheta ^ * \big\|_* \lesssim d \sqrt{\frac{p \log p}{n}},
\end{align*}
holds with probability at least $1 - O(p ^ {-10})$.
\end{proof}


\section{Rate between the DANIEL and Centralized Estimator}\label{app:rate_ctr_DANIEL}

In this section, we prove Theorem \ref{thm:rate_ctr_DANIEL}.
Similar to \eqref{equ:e_stat}, we define the statistical error for the centralized-distributed difference as
\begin{equation} \label{equ:d_stat}
    \delta_{\mathrm{stat}} = \sup_{{\rm rank}(\bDelta) \leq d, \|\bDelta \|_{\mathrm{F}} \leq 1} \big\langle \nabla _{\bTheta} \mathcal{L}_S (\widetilde\bTheta) - \nabla _{\bTheta} \mathcal{L} (\widetilde\bTheta), \bDelta \big\rangle.
\end{equation}

\begin{proof}[Proof of Theorem \ref{thm:rate_ctr_DANIEL}]
    The proof strategy is similar to the proof of Theorem~\ref{thm:Divide_conquer}. However, the key difference is that the gradient $\nabla_{\bTheta}\cL(\widetilde\bTheta)$ is not necessarily zero due to the nonconvex nature of the loss in \eqref{eq:centralized_estimator}. In order to show the higher order rate in the theorem, we need a more careful analysis of the gradient $\nabla_{\bTheta}\cL(\widetilde\bTheta)$. Define $\widetilde{\Zb} = (\widetilde{\Ub}\trans,\widetilde{\Vb}\trans)\trans$. 
 In the proof, we use similar induction procedure as  the one-step contraction in Lemma~\ref{lem:contraction_mapping} by assuming the induction assumption that $\rho^2(\Zb_\gamma, \widetilde\Zb) \leq K(\Ub^*; \kappa_{\min}, \kappa_{\max})$ for a given $\gamma$ and show the contraction behavior for $\Zb_{\gamma+1}$. For simplicity, we omit the detailed formality of induction procedure.

    We first have the following   one step analysis similar to  Section~\ref{app:eta}. We have
\begin{equation} \label{equ:rho_zzstar2}
    \rho^2(\Zb_{\gamma+1}, \widetilde{\Zb}) = \inf_{O \in \mathcal{O}(d)} \Bigg\|\begin{bmatrix} 
\Ub_{\gamma+1}  \\
\Vb_{\gamma+1} 
\end{bmatrix} - \begin{bmatrix}
\widetilde{\Ub} \Ob \\
\widetilde{\Vb} \Ob 
\end{bmatrix} \Bigg\| ^ 2_{\textrm{F}} \leq \Bigg\|\begin{bmatrix} 
\Ub_{\gamma+1} \\
\Vb_{\gamma +1}
\end{bmatrix} - \begin{bmatrix}
\widetilde{\Ub} \widetilde \Ob \\
\widetilde{\Vb} \widetilde \Ob 
\end{bmatrix} \Bigg\| ^ 2_{\textrm{F}},
\end{equation}
where $\widetilde \Ob$ is the optimal rotation matrix.
Plugging in the update rule in gradient descent, we have
\begin{align*}
\rho^2(\Zb_{\gamma+1} , \widetilde{\Zb}) & \leq  \rho^2(\Zb_\gamma, \widetilde{\Zb}) \\
& - 2 \eta \Bigg\langle \begin{bmatrix} 
\Ub_\gamma  - \widetilde{\Ub}\widetilde \Ob \\
\Vb_\gamma  - \widetilde{\Vb}\widetilde \Ob
\end{bmatrix},
\begin{bmatrix} 
\nabla_{\Ub} \mathcal{L}_S \big(\Ub_\gamma  \Vb _{\gamma} \trans\big)  + \nabla_{\Ub} \mathcal{Q} \big( \Ub _\gamma,  \Vb _\gamma \big) \\
\nabla_{\Vb} \mathcal{L}_S \big( \Ub_\gamma  \Vb _{\gamma} \trans \big)  + \nabla_{\Vb} \mathcal{Q} \big( \Ub _\gamma,  \Vb _\gamma \big)
\end{bmatrix}
\Bigg \rangle \\
&+ \eta ^ 2 \Bigg\|\begin{bmatrix} 
\nabla_{\Ub} \mathcal{L}_S \big(\Ub_\gamma  \Vb _{\gamma} \trans \big)  + \nabla_{\Ub} \mathcal{Q} \big( \Ub _\gamma,  \Vb _\gamma\big) \\
\nabla_{\Vb} \mathcal{L}_S \big( \Ub_\gamma  \Vb _{\gamma} \trans \big)  + \nabla_{\Vb} \mathcal{Q} \big( \Ub _\gamma,  \Vb _\gamma \big)
\end{bmatrix} \Bigg\|_{\mathrm{F}} ^ 2 \\
&\leq  \rho^2(\Zb_\gamma , \widetilde{\Zb}) - 2 \eta \underbrace{ \Bigg\langle \begin{bmatrix} 
\Ub_\gamma  - \widetilde{\Ub}\widetilde \Ob \\
\Vb_\gamma  - \widetilde{\Vb}\widetilde \Ob
\end{bmatrix},
\begin{bmatrix} 
\nabla_{\Ub} \mathcal{L}_S \big( \Ub_\gamma  \Vb _{\gamma} \trans \big) \\
\nabla_{\Vb} \mathcal{L}_S \big( \Ub_\gamma  \Vb _{\gamma} \trans \big) 
\end{bmatrix}
\Bigg \rangle }_{I_1} \\
& - 2 \eta \underbrace{\Bigg\langle \begin{bmatrix} 
\Ub_\gamma - \widetilde{\Ub}\widetilde \Ob \\
\Vb_\gamma - \widetilde{\Vb}\widetilde \Ob
\end{bmatrix},
\begin{bmatrix} 
\nabla_{\Ub} \mathcal{Q} \big( \Ub _\gamma,  \Vb _\gamma \big) \\
\nabla_{\Vb} \mathcal{Q} \big( \Ub _\gamma,  \Vb _\gamma \big)
\end{bmatrix}
\Bigg \rangle }_{I_2}\\
& + 2 \eta ^ 2 \underbrace{\Bigg\|\begin{bmatrix} 
\nabla_{\Ub} \mathcal{L}_S \big(\Ub_\gamma  \Vb _{\gamma} \trans \big) \\
\nabla_{\Vb} \mathcal{L}_S \big( \Ub_\gamma  \Vb _{\gamma} \trans \big) 
\end{bmatrix} \Bigg\|_{\mathrm{F}} ^ 2}_{I_3} + 2 \eta ^ 2 \underbrace{ \Bigg\|\begin{bmatrix} 
\nabla_{\Ub} \mathcal{Q} \big( \Ub _\gamma,  \Vb _\gamma \big) \\
\nabla_{\Vb} \mathcal{Q} \big( \Ub _\gamma,  \Vb _\gamma \big)
\end{bmatrix} \Bigg\|_{\mathrm{F}} ^ 2}_{I_4}.
\end{align*}

For $I_1$, we have that
\begin{align*}
I_1 & = \tr \Big(\nabla_{\bTheta} \mathcal{L}_S \big(\Ub_\gamma  \Vb _{\gamma} \trans  \big) \Vb_\gamma (\Ub_\gamma  - \widetilde{\Ub} \widetilde \Ob) \trans \Big)  + \tr \Big(\nabla_{\bTheta} \mathcal{L}_S \big(\Ub_\gamma  \Vb _{\gamma} \trans  \big) \trans \Ub_\gamma(\Vb_\gamma - \widetilde{\Vb} \widetilde \Ob) \trans \Big)\\
& = \underbrace{\tr \Big(\big(\nabla_{\bTheta} \mathcal{L}_S \big(\Ub_\gamma  \Vb _{\gamma} \trans  \big) - \nabla_{\bTheta} \mathcal{L}_S \big(\widetilde{\Ub} \widetilde{\Vb}^{\transpose} \big) \big) \big(\Vb_\gamma   \Ub_{\gamma} \trans  - \widetilde{\Vb} \widetilde{\Ub}^{\transpose} \big) \Big)}_{I_{1.1}} \\
&\quad + \underbrace{\tr \Big( \big(\nabla_{\bTheta} \mathcal{L}_S \big(\widetilde{\bTheta} \big) - \nabla_{\bTheta} \mathcal{L} \big(\widetilde{\bTheta} \big) \big)\big(\Vb_\gamma  \Ub_{\gamma}\trans - \widetilde{\Vb} \widetilde{\Ub}^{\transpose} \big) \Big)}_{I_{1.2}} \\
& \quad + \underbrace{\tr \Big( \big(\nabla_{\bTheta} \mathcal{L}_S \big(\Ub_\gamma  \Vb _{\gamma} \trans  \big) +  \nabla_{\bTheta} \mathcal{L} \big(\widetilde \bTheta  \big)  \big)\big( (\Vb_\gamma  - \widetilde{\Vb} \widetilde \Ob)(\Ub _\gamma - \widetilde{\Ub} \widetilde \Ob) \trans \big) \Big)}_{I_{1.3}}\\
&\quad + \underbrace{\tr \Big(-\nabla_{\bTheta} \mathcal{L} \big(\widetilde \bTheta  \big) \widetilde\Vb (\widetilde\Ub - \Ub_\gamma \widetilde \Ob)\trans - \nabla_{\bTheta} \mathcal{L} \big(\widetilde \bTheta  \big) \trans \widetilde\Ub (\widetilde\Vb - \Vb_\gamma \widetilde \Ob)\trans \Big)}_{I_{1.4}}.
\end{align*}
For $I_{1.1}$, similar to $T_{1.1}$ in the proof of Theorem~\ref{thm:Divide_conquer}, we have that
\begin{align*}
 I_{1.1} & \geq \frac{1}{4} \kappa_{\min} \big\|\Ub_\gamma  \Vb _{\gamma} \trans  - \widetilde{\Ub} \widetilde{\Vb}^{\transpose} \big\|_{\textrm{F}} ^ 2 + \frac{1}{\kappa_{\min} + \kappa_{\max}} \big\|\nabla_{\bTheta} \mathcal{L}_S \big(\Ub_\gamma  \Vb _{\gamma} \trans  \big) - \nabla_{\bTheta} \mathcal{L}_S \big(\widetilde{\Ub} \widetilde{\Vb}^{\transpose} \big) \big\|_{\textrm{F}}^2.
\end{align*}
For $I_{1.2}$, we have
\begin{align*}
 I_{1.2} & \geq -\delta_{\rm stat} \big\|\Ub_\gamma  \Vb _{\gamma} \trans  - \widetilde{\Ub} \widetilde{\Vb}^{\transpose} \big\|_{\textrm{F}}  \geq - \frac{3}{32} \kappa_{\min} \big\|\Ub_\gamma  \Vb _{\gamma} \trans  - \widetilde{\Ub} \widetilde{\Vb}^{\transpose} \big\|_{\textrm{F}} ^ 2  - \frac{8}{3 \kappa_{\min}} \delta_{\rm stat}^2.
\end{align*}
For $I_{1.3}$, we consider applying the Cauchy-Schwarz inequality on Frobenius inner product, which yields that
\begin{align*}
 I_{1.3} & \geq - \Big| \big\langle  \nabla_{\bTheta} \mathcal{L}_S \big( \widetilde\bTheta \big) + \nabla_{\bTheta} \mathcal{L} \big( \widetilde\bTheta \big),   (\Ub_\gamma - \widetilde{\Ub} \widetilde \Ob) (\Vb_\gamma - \widetilde{\Vb} \widetilde \Ob) \trans \big\rangle \Big| \\
 & - \Big| \big\langle \nabla_{\bTheta} \mathcal{L}_S \big(\Ub_\gamma  \Vb _{\gamma} \trans  \big) - \nabla_{\bTheta} \mathcal{L}_S \big(\widetilde{\Ub} \widetilde{\Vb}^{\transpose} \big), (\Ub_\gamma - \widetilde{\Ub} \widetilde \Ob) (\Vb_\gamma - \widetilde{\Vb} \widetilde \Ob) \trans  \big\rangle  \Big| \\
& \geq - \Big (e_{\rm stat} + \tilde{e}_{\rm stat} + \big\|\nabla_{\bTheta} \mathcal{L}_S \big(\Ub_\gamma  \Vb _{\gamma} \trans  \big) - \nabla_{\bTheta} \mathcal{L}_S \big(\widetilde{\Ub} \widetilde{\Vb}^{\transpose} \big) \big\|_{\textrm{F}} \Big) \big \| (\Ub_\gamma - \widetilde{\Ub} \widetilde \Ob) (\Vb_\gamma - \widetilde{\Vb} \widetilde \Ob) \trans \big \|_{\mathrm{F}} \\
& \geq - \frac{1}{2} \Big (e_{\rm stat} + \tilde{e}_{\rm stat}  + \big\|\nabla_{\bTheta} \mathcal{L}_S \big(\Ub_\gamma  \Vb _{\gamma} \trans  \big) - \nabla_{\bTheta} \mathcal{L}_S \big(\widetilde{\Ub} \widetilde{\Vb}^{\transpose} \big) \big\|_{\textrm{F}} \Big) \rho^2(\Zb_\gamma , \widetilde{\Zb}) \\
& \geq - \frac{1}{8\kappa^* \sigma_{d}(\bTheta ^ *)}\Big((e_{\rm stat}+ \tilde{e}_{\rm stat} )^2 + \big\|\nabla_{\bTheta} \mathcal{L}_S \big(\Ub_\gamma  \Vb _{\gamma} \trans  \big) - \nabla_{\bTheta} \mathcal{L}_S \big(\widetilde{\Ub} \widetilde{\Vb}^{\transpose} \big) \big\|_{\textrm{F}}^2 \Big) \rho^2(\Zb_\gamma, \widetilde{\Zb}) \\ 
& \quad - \frac{1}{2}\kappa^* \sigma_{d}(\bTheta ^ *) \rho^2(\Zb_\gamma, \widetilde{\Zb}),
\end{align*}
where $\tilde{e}_{\rm stat}$ is defined in \eqref{equ:te_stat}. 
 We summarize the terms $I_{1.1}$, $I_{1.2}$, and $I_{1.3}$ as
\begin{align*}
& I_{1.1} + I_{1.2} + I_{1.3} 
\geq \frac{5}{32} \kappa_{\min} \big\|\Ub_\gamma  \Vb _{\gamma} \trans  - \widetilde{\Ub} \widetilde{\Vb}^{\transpose} \big\|_{\textrm{F}} ^ 2 \\
& \quad + \left(\frac{1}{\kappa_{\min} + \kappa_{\max}}- \frac{\rho^2(\Zb_\gamma, \widetilde{\Zb})}{8\kappa^* \sigma_{d}(\bTheta ^ *)}\right)\big\|\nabla_{\bTheta} \mathcal{L}_S \big(\Ub_\gamma  \Vb _{\gamma} \trans  \big) - \nabla_{\bTheta} \mathcal{L}_S \big(\widetilde{\Ub} \widetilde{\Vb}^{\transpose} \big) \big\|_{\textrm{F}}^2  \\
&\quad - \frac{1}{2} \kappa^* \sigma_{d}(\bTheta ^ *) \rho ^ 2(\Zb_\gamma, \widetilde{\Zb}) - \Big(\frac{8}{3 \kappa_{\min}} + \frac{1}{8\kappa^* \sigma_{d}(\bTheta ^ *)}\Big) (\delta_{\rm stat}^2 + (e_{\rm stat} + \tilde e_{\rm stat})^2\rho^2(\Zb_\gamma, \widetilde{\Zb})).
\end{align*}

By the first order optimality condition of $\widetilde\Zb$ in \eqref{eq:centralized_estimator}, defining $\widetilde \bPi = \widetilde\Ub\widetilde\Ub\trans- \widetilde\Vb\widetilde\Vb\trans$ we have
\begin{align*}
    I_{1.4} &=  \tr \Big(- \nabla_{\bTheta} \mathcal{L} \big(\widetilde \bTheta  \big) \widetilde\Vb (\widetilde\Ub - \Ub_\gamma \widetilde \Ob)\trans - \nabla_{\bTheta} \mathcal{L} \big(\widetilde \bTheta  \big) \trans \widetilde\Ub (\widetilde\Vb - \Vb_\gamma \widetilde \Ob)\trans \Big) \\
    & = \tr \Big(\nabla_{\Vb} \mathcal{Q} \big(\widetilde \bPi  \big) (\widetilde\Ub - \Ub_\gamma \widetilde \Ob)\trans + \nabla_{\Ub} \mathcal{Q}\big(\widetilde \bPi\big) \trans  (\widetilde\Vb - \Vb_\gamma \widetilde \Ob)\trans \Big).
\end{align*}
By Lemma B.1 of \cite{Park2016lowrank},  we have
\begin{align*}
I_2 + I_{1.4}&\geq \underbrace{\frac{1}{4} \Big[ \big\|\Ub_\gamma \Ub_{\gamma} \trans - \widetilde{\Ub} \widetilde{\Vb}^{\transpose} \big\|_{\textrm{F}} ^ 2  + \big\|\Vb_\gamma \Vb_{\gamma} \trans - \widetilde{\Ub} \widetilde{\Vb}^{\transpose} \big\|_{\textrm{F}} ^ 2  - 2\big\|\Ub_\gamma \Vb_{\gamma} \trans  - \widetilde{\Ub} \widetilde{\Vb}^{\transpose} \big\|_{\textrm{F}} ^ 2  \Big]}_{I_{2.1}} \\
&+ \underbrace{\frac{1}{2} \big\| \nabla_{\bPi} \mathcal{Q} \big(\bPi_\gamma \big) \big \|_{\mathrm{F}}^2 + \frac{1}{2} \big\| \nabla_{\bPi} \mathcal{Q} \big(\widetilde \bPi \big) \big \|_{\mathrm{F}}^2}_{I_{2.2}} - \underbrace{\frac{1}{2}\big( \big\| \nabla_{\bPi} \mathcal{Q} \big(\bPi_\gamma  \big) \big \|_{\rm op} + \big\| \nabla_{\bPi} \mathcal{Q} \big(\widetilde \bPi \big) \big \|_{\mathrm{op}}\big)\Bigg\| \begin{bmatrix} 
\Ub_\gamma  - \widetilde{\Ub}\widetilde \Ob \\
\Vb_\gamma  - \widetilde{\Vb}\widetilde \Ob
\end{bmatrix} \Bigg\|_{\mathrm{F}}^2}_{I_{2.3}}.
\end{align*}
Similar to the analysis of $T_2$ in the proof of Theorem~\ref{thm:Divide_conquer}, we have
\begin{align*}
I_{2.1} &+ \frac{5 \kappa_{\min}}{32} \big\|\Ub_\gamma \Vb_{\gamma} \trans  - \widetilde{\Ub} \widetilde{\Vb}^{\transpose} \big\|_{\textrm{F}} ^ 2  \geq \min \Big\{\frac{1}{8}, \frac{5 \kappa_{\min}}{128} \Big\} \Big[\big\|\Ub_\gamma \Ub_{\gamma} \trans  - \widetilde{\Ub} \widetilde{\Vb}^{\transpose} \big\|_{\textrm{F}} ^ 2  \\
&+ \big\|\Vb_\gamma \Vb_{\gamma} \trans - \widetilde{\Ub} \widetilde{\Vb}^{\transpose} \big\|_{\textrm{F}} ^ 2 + 2 \big\|\Ub_{\gamma} \Vb_{\gamma} \trans  - \widetilde{\Ub} \widetilde{\Vb}^{\transpose} \big\|_{\textrm{F}} ^ 2  \Big] \\
&= \kappa^* \Big[ \big\|\Ub_\gamma \Ub_{\gamma} \trans  - \widetilde{\Ub} \widetilde{\Vb}^{\transpose}  \big\|_{\textrm{F}} ^ 2 + \big\|\Vb_\gamma \Vb_{\gamma} \trans - \widetilde{\Ub} \widetilde{\Vb}^{\transpose}  \big\|_{\textrm{F}} ^ 2  + 2 \big\|\Ub_{\gamma} \Vb_{\gamma} \trans  - \widetilde{\Ub} \widetilde{\Vb}^{\transpose} \big\|_{\textrm{F}} ^ 2  \Big] \\
&= \kappa^* \Bigg\|\begin{bmatrix} 
\Ub_\gamma \Ub_{\gamma} \trans & \Ub_\gamma \Vb_{\gamma} \trans \\
\Vb_\gamma \Ub_{\gamma} \trans & \Vb_\gamma \Vb_{\gamma} \trans
\end{bmatrix} - \begin{bmatrix} 
\widetilde{\Ub} \widetilde{\Vb}^{\transpose} & \widetilde{\Ub} \widetilde{\Vb}^{\transpose} \\
\widetilde{\Ub} \widetilde{\Vb}^{\transpose} & \widetilde{\Ub} \widetilde{\Vb}^{\transpose}
\end{bmatrix} \Bigg\|_{\mathrm{F}} ^ 2 \geq 1.55 \kappa ^ * \sigma_d(\bTheta^*)  \rho^2(\Zb_\gamma , \widetilde{\Zb}); \\
I_{2.3} &\leq \frac{1}{8} \big\| \nabla_{\bPi} \mathcal{Q} \big(\bPi_\gamma \big) \big \|_{\rm F}^2 + \frac{1}{8} \big\| \nabla_{\bPi} \mathcal{Q} \big(\widetilde\bPi\big) \big \|_{\rm F}^2+ \frac{2}{5} \kappa^* \sigma_d(\bTheta^*) \rho^2(\Zb_\gamma, \widetilde{\Zb}),
\end{align*}
where we use $\sigma_d(\widetilde \bTheta) \ge \sigma_d(\bTheta^*)(1-o(1))$ with probability $1-O(p^{-10})$ by Theorem~\ref{thm:rate_ctr}.

Combining the terms of $I_1$ and $I_2$, we have that
\begin{align*}
I_1 + I_2 &\geq  \left(\frac{1}{\kappa_{\min} + \kappa_{\max}}- \frac{\rho^2(\Zb_\gamma, \widetilde{\Zb})}{8\kappa^* \sigma_{d}(\bTheta ^ *)}\right) \big\|\nabla_{\bTheta} \mathcal{L}_S \big(\Ub_\gamma \Vb_{\gamma} \trans \big) - \nabla_{\bTheta} \mathcal{L}_S \big(\widetilde{\Ub} \widetilde{\Vb}^{\transpose} \big) \big\|_{\textrm{F}}^2  \\
&\quad +\kappa^* \sigma_d(\bTheta^*) \rho^2(\Zb_\gamma, \widetilde{\Zb}) -\Big(\frac{8}{3 \kappa_{\min}} + \frac{1}{8\kappa^* \sigma_{d}(\bTheta ^ *)}\Big) (\delta_{\rm stat}^2 + (e_{\rm stat} + \tilde e_{\rm stat})^2\rho^2(\Zb_\gamma, \widetilde{\Zb})) \\  
&\quad + \frac{3}{8}  \big\| \nabla_{\bPi} \mathcal{Q} \big(\bPi_\gamma \big) \big \|_{\mathrm{F}}^2 +  \frac{3}{8}  \big\| \nabla_{\bPi} \mathcal{Q} \big(\widetilde\bPi \big) \big \|_{\mathrm{F}}^2.
\end{align*}
For $I_3$, we consider the chain rule for derivatives such that
\begin{equation*}
I_3 = \underbrace{ \big\|\nabla_{\bTheta}\mathcal{L}_S(\Ub_\gamma \Vb_{\gamma} \trans) \Vb_\gamma  \big\|_{\mathrm{F}} ^ 2}_{I_{3.1}} + \underbrace{ \big\|\Ub_\gamma  ^ { \transpose} 
\nabla_{\bTheta}\mathcal{L}_S(\Ub_\gamma \Vb_{\gamma} \trans) \big\|_{\mathrm{F}} ^ 2 }_{I_{3.2}}.
\end{equation*}
For $I_{3.1}$ and $I_{3.2}$, by the first order optimality condition of $\widetilde\Zb$ in \eqref{eq:centralized_estimator}, we have
\begin{align*}
    I_{3.1} & \leq 2 \big\|\big(\nabla_{\bTheta}\mathcal{L}_S(\widetilde{\Ub} \widetilde{\Vb}^{\transpose})-\nabla_{\bTheta}\mathcal{L}(\widetilde{\Ub} \widetilde{\Vb}^{\transpose}) \big)\Vb_\gamma \big\|_{\mathrm{F}} ^ 2 + 2 \Big\| \big( \nabla_{\bTheta}\mathcal{L}_S(\Ub_\gamma \Vb_{\gamma} \trans) - \nabla_{\bTheta}\mathcal{L}_S(\widetilde{\Ub} \widetilde{\Vb}^{\transpose}) \big) \Vb_\gamma \Big\|_{\mathrm{F}} ^ 2\\
    & \quad + 2 \big\|\nabla_{\bTheta}\mathcal{L}(\widetilde{\Ub} \widetilde{\Vb}^{\transpose}) (\Vb_\gamma - \widetilde \Vb \widetilde \Ob) \big\|_{\mathrm{F}} ^ 2 + 2\big\| \nabla_{\Vb} \mathcal{Q} \big(\widetilde\bPi \big) \big \|_{\mathrm{F}}^2\\
    I_{3.2} & \leq 2 \big\|\Ub_\gamma\trans\big(\nabla_{\bTheta}\mathcal{L}_S(\widetilde{\Ub} \widetilde{\Vb}^{\transpose})-\nabla_{\bTheta}\mathcal{L}(\widetilde{\Ub} \widetilde{\Vb}^{\transpose}) \big)\big\|_{\mathrm{F}} ^ 2 + 2 \Big\|\Ub_\gamma\trans \big( \nabla_{\bTheta}\mathcal{L}_S(\Ub_\gamma \Vb_{\gamma} \trans) - \nabla_{\bTheta}\mathcal{L}_S(\widetilde{\Ub} \widetilde{\Vb}^{\transpose}) \big)  \Big\|_{\mathrm{F}} ^ 2\\
    & \quad + 2 \big\|(\Ub_\gamma - \widetilde \Ub\widetilde \Ob) \trans\nabla_{\bTheta}\mathcal{L}(\widetilde{\Ub} \widetilde{\Vb}^{\transpose}) \big\|_{\mathrm{F}} ^ 2 + 2\big\| \nabla_{\Ub} \mathcal{Q} \big(\widetilde\bPi \big) \big \|_{\mathrm{F}}^2.
\end{align*}

Combining the upper bound for $I_{3.1}$ and $I_{3.2}$ implies that 
\begin{align*}
I_3 & {\leq}  4 \Big (\delta_{\rm stat}^2 + \big \|\nabla_{\bTheta}\mathcal{L}_S(\Ub_\gamma \Vb_{\gamma} \trans) - \nabla_{\bTheta}\mathcal{L}_S(\widetilde{\Ub} \widetilde{\Vb}^{\transpose}) \big \|_{\mathrm{F}}^2 \Big) \big\|\Zb_\gamma \big\|_{\rm op}^2 + 2\tilde e_{\rm stat}^2 \rho^2(\Zb_\gamma, \widetilde{\Zb}) + 2 \big\| \nabla_{\bPi} \mathcal{Q} \big(\widetilde\bPi \big) \big \|_{\rm F}^2,
\end{align*}
where $\tilde e_{\rm stat}$ is defined in \eqref{equ:te_stat}.
For $I_4$, we have that
\begin{equation*}
I_4 \leq \big\| \Ub_\gamma  \nabla_{\bPi} \mathcal{Q} \big(\bPi_\gamma  \big) \big \|_{\rm F}^2 + \big\| \Vb_\gamma \nabla_{\bPi} \mathcal{Q} \big(\bPi_\gamma \big) \big \|_{\rm F}^2 \leq 2 \big\|\Zb_\gamma  \big\|_{\rm op}^2 \big\| \nabla_{\bPi} \mathcal{Q} \big(\bPi_\gamma  \big) \big \|_{\rm F}^2.
\end{equation*}
We combine the bounds for $I_3$ and $I_4$, which indicates that
\begin{align*}
I_3 + I_4 &\leq  4 \Big (\delta_{\rm stat}^2 + \big \|\nabla_{\bTheta}\mathcal{L}_S(\Ub_\gamma \Vb_{\gamma} \trans )  - \nabla_{\bTheta}\mathcal{L}_S(\widetilde{\Ub} \widetilde{\Vb}^{\transpose}) \big \|_{\mathrm{F}}^2 \Big) \big\|\Zb_\gamma \big\|_{\rm op}^2 \\
&\quad + 2\tilde e_{\rm stat}^2 \rho^2(\Zb_\gamma, \widetilde{\Zb}) + 2 \big\|\Zb_\gamma  \big\|_{\rm op}^2 \big\| \nabla_{\bPi} \mathcal{Q} \big(\bPi_\gamma \big) \big \|_{\rm F}^2 + 2 \big\| \nabla_{\bPi} \mathcal{Q} \big(\widetilde\bPi \big) \big \|_{\rm F}^2.
\end{align*}

We assemble the bound from $I_1$ to $I_4$, along with $\rho^2(\Zb_\gamma, \widetilde\Zb)\le 2 \rho^2(\Zb_\gamma, \Zb^*) + 2\rho^2(\Zb^*, \widetilde\Zb) = o_P(1)$ by Theorems~\ref{thm:Divide_conquer} and~\ref{thm:rate_ctr}, which implies that for $\eta \leq \frac{1}{8 (\|\Zb_{\gamma} \|_{\rm op}^2+1)} \min{\Big\{1, \frac{1}{\kappa_{\min} + \kappa_{\max}} \Big\}}$,
\begin{align*}
\rho^2 \big(\Zb_{\gamma+1} ,\widetilde{\Zb} \big) &\leq \rho^2 \big(\Zb_\gamma, \widetilde{\Zb} \big) - 2 \eta \big[I_1 + I_2 \big] + 2 \eta ^ 2 \big[I_3 + I_4 \big] \\
& \leq \big(1 - 2 \eta \kappa^* \sigma_d(\bTheta^*) \big) \rho^2 \big(\Zb_\gamma, \widetilde{\Zb} \big) \\
& \quad +  6 \eta \left(\frac{\kappa_{\min} + \kappa_{\max}}{\kappa_{\min} \kappa_{\max}}+\frac{1}{8\kappa^* \sigma_{d}(\bTheta ^ *)} \right)(\delta_{\rm stat}^2 + (e_{\rm stat} + \tilde e_{\rm stat})^2\rho^2(\Zb_\gamma, \widetilde{\Zb})).
\end{align*}

As $e_{\rm stat} + \tilde e_{\rm stat} = o_P(1)$ by Lemma~\ref{lem:epsilon_non-convex} and \eqref{eq:tilde-estat-rate}, we have for some constant $C$ with probability at least $1 - O(p^{-10})$,  
\begin{align*}
    \rho^2 \big(\Zb_{\gamma+1} ,\widetilde{\Zb} \big)  \leq \big(1 - 1.5 \eta \kappa^* \sigma_d(\bTheta^*) \big) \rho^2 \big(\Zb_\gamma, \widetilde{\Zb} \big) +  C \delta_{\rm stat}^2.
    \end{align*}
By induction similar to the one-step contraction in Lemma~\ref{lem:contraction_mapping}, we can prove the rate of convergence of $\rho^2(\Zb_\gamma, \widetilde\Zb)$  and get the desired rates by combining  the rates in Lemma~\ref{lem:d_rate}, \eqref{eq:tilde-estat-rate}, and Lemma~\ref{lem:epsilon_non-convex} with the initialization rate in Theorem~\ref{thm:main}. 
\end{proof}

\subsection{Bounding the Centralized-Distributed Statistical Error}

We prove an upper bound on $\delta_{\rm stat}$ in \eqref{equ:d_stat}. The form and proof is similar to the proof of Lemma \ref{lem:epsilon_non-convex}.

\begin{lemma}\label{lem:d_rate}
Under Assumptions \ref{ass:Positive_minimal_eigenvalue}-\ref{ass:reg_grad}, the statistical error $\delta_{\rm stat}$ in \eqref{equ:d_stat} has
\begin{equation*}
    \delta_{\mathrm{stat}} \lesssim  p ^ 3 \sqrt{\frac{\log p}{n/m}} \big\|\widehat \bTheta_0 -  {\bTheta}^* \big\|_{\mathrm{F}} + p^3 \|\widehat \bTheta_0 -  {\bTheta}^* \big\|_{\mathrm{F}} ^ 2, \text{with probability at least $1 - O(p^{-10})$.}
\end{equation*}
If we use the initialization $\widehat\bTheta_{\rm cvx}$ in \eqref{equ:question}, by Theorem~\ref{thm:main}, we have if $m= o(\sqrt{n/(dp^7\log p)})$,
\[
\delta_{\mathrm{stat}} \lesssim  p^3 {\frac{dp \log p}{n/m}} = o\Big(\sqrt{\frac{dp \log p}{n}}\Big), \text{with probability at least $1 - O(p^{-10})$.}
\]
\end{lemma}

\begin{proof}
    By the definitions of two losses, we have
    \begin{align*}
        &\nabla _{\bTheta} \mathcal{L}_S (\widetilde{\bTheta})  - \nabla _{\bTheta} \mathcal{L} (\widetilde{\bTheta})  =  \nabla _{\bTheta} \mathcal{L}_1 (\widetilde{\bTheta})  - \nabla _{\bTheta} \mathcal{L}_1 (\widehat\bTheta_0)+ \nabla _{\bTheta} \mathcal{L} (\widehat\bTheta_0 )  - \nabla _{\bTheta} \mathcal{L} (\widetilde{\bTheta})\\
        &= \mat \bigg\{ \int_0 ^ 1 \bigg[  \nabla ^ 2 \mathcal{L} \big(\widetilde{\bTheta} + s \big(\widehat \bTheta_0 - \widetilde{\bTheta} \big) \big) - \nabla ^ 2 \mathcal{L}_1 \big(\widetilde{\bTheta} + s \big(\widehat \bTheta_0 - \widetilde{\bTheta} \big) \big)\bigg] ds \bigg\} \vec \big(\widehat \bTheta_0 - \widetilde{\bTheta} \big)\\
        & = \widetilde\Eb_2 + \widetilde\Eb_3,
    \end{align*}
 where $\widetilde\Eb_2, \widetilde\Eb_3$ are similar to \eqref{eq:E3} and defined as follows:
 \begin{align}
    \nonumber  \widetilde\Eb_2 &= {{\rm mat} \Bigg\{ \bigg[\frac{1}{m} \sum_{j = 1} ^ m \nabla ^ 2 \mathcal{L}_j \big(\bTheta ^ * \big) - \nabla ^ 2 \mathcal{L}_1 \big(\bTheta ^* \big)  \bigg] \vec \big(\widehat \bTheta_0 - \widetilde\bTheta \big) \Bigg\} } \\
    \nonumber \widetilde\Eb_3 &= \mat \Bigg\{ \bigg\{ \int_0 ^ 1 \bigg[ \frac{1}{m} \sum_{j = 1} ^ m \nabla ^ 2 \mathcal{L}_j \big(\widetilde\bTheta + s \big(\widehat \bTheta_0 - \widetilde\bTheta \big) \big) - \frac{1}{m} \sum_{j = 1} ^ m  \nabla ^ 2 \mathcal{L}_j \big(\bTheta ^ *\big)  \bigg] ds \bigg\} \vec \big(\widehat \bTheta_0 - \widetilde\bTheta \big) \\
    \nonumber  & \qquad  \qquad { - \bigg\{ \int_0 ^ 1 \bigg[ \nabla ^ 2 \mathcal{L}_1 \big(\widetilde\bTheta + s \big(\widehat \bTheta_0 - \widetilde\bTheta \big) \big) - \nabla ^ 2 \mathcal{L}_1 \big(\bTheta ^ *\big)  \bigg] ds \bigg\} \vec \big(\widehat \bTheta_0 - \widetilde\bTheta \big) \Bigg\}}. 
   \end{align}

   Similar to \eqref{equ:H_H1}, we have for any $\bDelta$ with $\|\bDelta\|_{\mathrm{F}} \leq 1$ and ${\rm rank}(\bDelta) \leq d$,
   \begin{equation} \label{equ:d_stat-H_H1}
   \begin{split}
       \tr \big( \bDelta \trans \widetilde\Eb_2 \big) & \leq \Big\| \widehat{\Hb} (\bTheta ^ *) - \widehat{\Hb}_1 (\bTheta ^ *) \Big\|_{\rm op} \big\|\widehat \bTheta_0 - \widetilde\bTheta \big\|_{\rm F} \\
       & \leq p \Big\| \widehat{\bSigma} (\bTheta ^ *) - \widehat{\bSigma}_1 (\bTheta ^ *)  \Big\|_{\max} \big\|\widehat \bTheta_0 - \widetilde\bTheta  \big\|_{\rm F} \lesssim p^3 \sqrt{\frac{\log p}{n/m}} \big\|\widehat \bTheta_0 - \bTheta^*\big\|_{\rm F},
   \end{split}
   \end{equation}
   with probability at least $1 - O(p ^ {-10})$.
    For $\widetilde \Eb_3$, similar to \eqref{eq:traceE3}, we have
    \begin{align*}
        \tr \big( \bDelta \trans \Eb_3 \big) &\lesssim p ^ 3 \big[\big\|\widehat\bTheta_0 - \widetilde\bTheta \big\|_{\max} + \big\|\widehat\bTheta_0 -  \bTheta ^ *  \big\|_{\max}\big] \big\|\widehat \bTheta_0 - \bTheta ^ * \big\|_{\rm F} \lesssim p ^ 3 \big\|\widehat \bTheta_0 - \bTheta^* \big\|_{\rm F}^2,
    \end{align*}
    where the last inequality above and last inequality of \eqref{equ:d_stat-H_H1} are due to Theorem~\ref{thm:rate_ctr} and Theorem~\ref{thm:main} that $\|\widetilde \bTheta - \bTheta^* \|_{\rm F} \ll \|\widehat \bTheta_0 - \bTheta^* \big\|_{\rm F}$.
\end{proof}

\subsection{Statistical Rate of the Centralized Estimator} \label{sec:center-rate}

In this section, we show the statistical rate of the centralized estimator $\widetilde \bTheta$ in \eqref{eq:centralized_estimator}.

\begin{theorem}[Rate of Centralized Estimator]\label{thm:rate_ctr}
    Under the same conditions of Theorem~\ref{thm:Divide_conquer}, we have 

    \begin{equation*}
        \inf_{\Ob \in \mathcal{O}(d)} \Big\{ \big\|\widetilde{\Ub}  - \Ub^* \Ob \big\|_{\mathrm{F}} ^ 2 + \big\|\widetilde{\Vb}  - \Vb^* \Ob \big\|_{\mathrm{F}} ^ 2 \Big\} \lesssim \frac{d p \log p}{n} \text{ \ \ with probability \ \ } 1 - O(p ^ {-10}),
    \end{equation*}
    where $\mathcal{O}(d)$ denotes the collection of $d \times d$ orthogonal matrices. Furthermore, we have 
    \begin{equation*}
        \big\|\widetilde{\bTheta} - \bTheta ^ * \big\|_{\mathrm{F}} \lesssim \sqrt{\frac{d p \log p}{n}} \text{ \ \ with probability \ \ } 1 - O(p ^ {-10}).
    \end{equation*}
\end{theorem}  
\begin{proof}
    By the first order optimality condition, we have
    \begin{align}
        \nonumber    \rho^2(\widetilde\Zb , \Zb^*) & =  \rho^2(\widetilde\Zb, \Zb^*) \\ \nonumber
            &  \quad - 2 \eta \Bigg\langle \begin{bmatrix} 
        \widetilde{\Ub}  - \Ub^*\widehat{\Ob} \\ \nonumber
        \widetilde{\Vb}  - \Vb^*\widehat{\Ob}
        \end{bmatrix},
        \begin{bmatrix} 
        \nabla_{\Ub} \mathcal{L} \big(\widetilde{\Ub}  \widetilde{\Vb} \trans\big)  + \nabla_{\Ub} \mathcal{Q} \big( \widetilde{\Ub},  \widetilde{\Vb} \big) \\ \nonumber
        \nabla_{\Vb} \mathcal{L} \big( \widetilde{\Ub}  \widetilde{\Vb} \trans \big)  + \nabla_{\Vb} \mathcal{Q} \big( \widetilde{\Ub},  \widetilde{\Vb} \big)
        \end{bmatrix}
        \Bigg \rangle.
        \end{align}
        This equation is similar to \eqref{eq:gd} in the proof of Theorem~\ref{thm:Divide_conquer}. The only difference is  that the loss is changed from the surrogate loss $\cL_S$ to $\cL$. Therefore, we can follow the same steps as the proof of Theorem~\ref{thm:Divide_conquer} with the only difference from bounding the rate of  $e_{\rm stat}$ in \eqref{equ:e_stat} to 
    \begin{equation}
        \label{equ:te_stat_tilde}
        \tilde e^*_{\rm stat} =  \sup_{{\rm rank}(\bDelta) \leq d, \|\bDelta \|_{\mathrm{F}} \leq 1} \big\langle \nabla _{\bTheta} \mathcal{L} ( \Ub^* \Vb^{*\top}), \bDelta \big\rangle.
    \end{equation}
     By \eqref{equ:bernstein_res}, we have $\tilde e^*_{\rm stat} = O_P(\sqrt{{dp \log p}/{n}})$. 
     The rest of the proof follows the same steps as the proof of Theorem~\ref{thm:Divide_conquer} that we omit the details here and we get the following result similar to \eqref{eq:one-step}
     \begin{align}
        \nonumber   \rho^2 \big(\widetilde\Zb , \Zb ^ * \big) 
           & \leq \big(1 - 2 \eta \kappa^* \sigma_d(\bTheta^*) \big) \rho^2 \big(\widetilde\Zb, \Zb^* \big) +  6 \eta \frac{\kappa_{\min} + \kappa_{\max}}{\kappa_{\min} \kappa_{\max}} (\tilde e^*_{\rm stat})^2,
       \end{align}
       which gives us the desired rates in the theorem.

    We also bound the following quantity which will be useful in the proof of Theorem~\ref{thm:rate_ctr_DANIEL}:
     \begin{equation}
        \label{equ:te_stat}
        \tilde e_{\rm stat} =  \sup_{{\rm rank}(\bDelta) \leq d, \|\bDelta \|_{\mathrm{F}} \leq 1} \big\langle \nabla _{\bTheta} \mathcal{L} ( \widetilde\Ub \widetilde\Vb^{\top}), \bDelta \big\rangle.
    \end{equation}
    By Assumption~\ref{ass:Positive_minimal_eigenvalue} and Theorem~\ref{thm:rate_ctr}, we have for any $\bDelta$ with $\|\bDelta\|_{\mathrm{F}} \leq 1$ and ${\rm rank}(\bDelta) \leq d$
    \begin{align*}
        \bigl\langle \nabla _{\bTheta} \mathcal{L} (\widetilde{\bTheta})  - \nabla _{\bTheta} \mathcal{L} ({\bTheta}^*) , \bDelta \bigr\rangle 
        &=\vec(\bDelta)\trans \mat \bigg\{ \int_0 ^ 1   \nabla ^ 2 \mathcal{L} \big({\bTheta}^* + s \big(\widetilde\bTheta - {\bTheta}^* \big) \big)  ds \bigg\} \vec \big(\widetilde \bTheta - {\bTheta}^* \big)\\
        &\lesssim \kappa_{\max} \sqrt{\frac{d p \log p}{n}},
    \end{align*}
    where the last inequality is due to  \eqref{eq:upper_bound_H} and the scaling assumption $p^2\sqrt{\log p/n} = o(1)$. Combining with \eqref{equ:te_stat} and the rate of $ \tilde e_{\rm stat}^* $ above, we have 
    \begin{equation}
        \label{eq:tilde-estat-rate}
        \tilde e_{\rm stat} \lesssim \sqrt{{dp \log p}/{n}}\text{ \ \ with probability \ \ } 1 - O(p ^ {-10}).
    \end{equation}
\end{proof}

\section{Auxiliary Lemmas for Non-Convex Optimization} \label{sec:B}

In this section, we provide the proofs of several auxiliary lemmas that support our proof to the main theorem. The theoretical foundation for the lemmas presented in Appendix \ref{sec:B} draws inspiration from \cite{Park2016lowrank}.

\subsection{RSC/RSS Conditions for the Non-Convex Theory}

In this section, we provide some insight into the restricted strong convexity/smoothness  conditions for the objective function $\mathcal{L}$ and the bi-factored surrogate loss $\widetilde{\mathcal{L}}$.

\begin{lemma}  \label{lem:RSC_RSS}
Given Assumptions \ref{ass:Positive_minimal_eigenvalue}-\ref{ass:reg_grad}, for any $\bTheta_1, \bTheta_2 \in \mathcal{T}$ where $\diag (\bTheta_1) = \diag (\bTheta_2)$ and $\max \big\{{\rm rank}(\bTheta_1), {\rm rank}(\bTheta_2) \big\} \leq d$, there exists some constant $C$ such that for $n \geq  C d ^ 2 p ^ 4 \log p$,
\begin{equation*}
     \frac{1}{4} \kappa_{\min} \big\|\bTheta_2 - \bTheta_1 \big\| ^ 2 _ {\mathrm{F}} \leq \mathcal{L}(\bTheta_2) -  \mathcal{L}(\bTheta_1) - \big\langle \nabla  \mathcal{L}(\bTheta_1), \bTheta_2 - \bTheta_1 \big\rangle \leq \Big(\frac{1}{2} \kappa_{\max} + \frac{1}{4} \kappa_{\min} \Big) \big\|\bTheta_2 - \bTheta_1 \big\| _ {\mathrm{F}} ^ 2 
\end{equation*}
with a probability of at least $1 - 2 p ^ {-10}$.
\end{lemma}

\begin{proof}[Proof of Lemma \ref{lem:RSC_RSS}]
First, we take second-order Taylor's expansion for $\mathcal{L}(\bTheta_2)$ at $\bTheta_1$, where
\begin{equation} \label{equ:LTheta2_LTheta1}
    \mathcal{L}(\bTheta_2) =  \mathcal{L}(\bTheta_1) + \big\langle \nabla  \mathcal{L}(\bTheta_1), \bTheta_2 - \bTheta_1 \big\rangle + \frac{1}{2} \vec{(\bDelta)} \trans \widehat{\Hb}(\mathring{\bTheta}) \vec{(\bDelta)}.
\end{equation}
Here $\bDelta = \bTheta_1 - \bTheta_2$ with $\diag(\bDelta) = 0$, $\mathring{\bTheta} =\lambda\bTheta_1 + (1-\lambda)\bTheta_2$ for some $\lambda \in [0,1]$. We then apply the triangular inequality on $\big\|\mathring{\bTheta} \big\|_\mathrm{F}$, which indicates that
\begin{equation*}
        \big\|\mathring{\bTheta} \big\|_\mathrm{F} = \|\lambda \bTheta_1 + (1 - \lambda) \bTheta_2 \|_\mathrm{F} \leq \lambda \|\bTheta_1\|_\mathrm{F} +  (1 - \lambda) \|\bTheta_2\|_\mathrm{F} \leq B,
\end{equation*}
given that $\bTheta_1, \bTheta_2 \in \mathcal{T}$. Besides, Assumption \ref{ass:Positive_minimal_eigenvalue} implies that for $\mathring{\bTheta}$,
\begin{equation} \label{equ:vec_difference}
        \vec{(\bDelta)} \trans \widehat{{\Hb}}(\mathring{\bTheta}) \vec{(\bDelta)} = \vec{(\bDelta)} \trans {\Hb}(\mathring{\bTheta}) \vec{(\bDelta)} + \vec{(\bDelta)} \trans (\widehat{{\Hb}}(\mathring{\bTheta}) - {\Hb}(\mathring{\bTheta})) \vec{(\bDelta)}.
\end{equation}
Combining \eqref{equ:LTheta2_LTheta1} and \eqref{equ:vec_difference}, we alternatively need to show that
\begin{equation} \label{equ:bound_kappamin}
        \vec{(\bDelta)} \trans (\widehat{{\Hb}}(\mathring{\bTheta}) - {\Hb}(\mathring{\bTheta})) \vec{(\bDelta)} \in \Big[-\frac{1}{2} \kappa_{\min{}} \|\bDelta\|_{\mathrm{F}}^2, \frac{1}{2} \kappa_{\min{}} \|\bDelta\|_{\mathrm{F}}^2 \Big].
\end{equation}
    Applying \eqref{equ:concentration_Sigma} under the same condition as proposed in Theorem \ref{thm:main}, we have that for some constant $C$,
    \begin{equation*}
        \vec{(\bDelta)} \trans (\widehat{{\Hb}}(\mathring{\bTheta}) - {\Hb}(\mathring{\bTheta})) \vec{(\bDelta)} \in \Big[- Cp^2 \sqrt{\frac{\log p}{n}} \|\bDelta\|_*^2, Cp^2 \sqrt{\frac{\log p}{n}} \|\bDelta\|_*^2 \Big],
    \end{equation*}
    with probability $1 - 2p ^ {-10}$. We then utilize the $2d$-low-rankness condition of $\bDelta$ along with the inequality $\|\Ab\|^2_* \leq {\rm rank} (\Ab) \|\Ab\|_{\mathrm{F}}^2$, which yields that
\begin{align*}
    C p^2 \sqrt{\frac{ \log p}{n}} \|\bDelta\|_*^2 \leq 2 C d p^2 \sqrt{\frac{ \log p}{n}} \|\bDelta \|_{\mathrm{F}} ^ 2.
\end{align*}
Setting $n \geq C d ^ 2 p ^ 4 \log p$ for some constant $C$ guarantees \eqref{equ:bound_kappamin} with probability $1 - 2p ^{-10}$. Hence, we have completed the proof for Lemma \ref{lem:RSC_RSS}.
\end{proof}

In Corollary \ref{col:surrogate_RSCRSS}, we consider the restricted strong convexity/smoothness (RSC/RSS) conditions, as inspired by \cite{jordan2019dist}, for the surrogate loss
\begin{equation} \label{def:surrogate_Ls}
    \mathcal{L}_S(\bTheta; \widehat\bTheta_0) = \mathcal{L}_1(\bTheta) + \big \langle \nabla \mathcal{L}(\widehat \bTheta_0) - \nabla \mathcal{L}_1(\widehat \bTheta_0), \bTheta \big \rangle,
\end{equation}
where $\mathcal{L}_S(\Ub \Vb \trans; \widehat\bTheta_0) + \frac{1}{4}\big\|\Ub\trans \Ub - \Vb \trans \Vb\big\|_{\mathrm{F}} ^ 2$ is equivalent with the bi-factored surrogate loss $\widetilde{\mathcal{L}}(\Ub, \Vb; \widehat\bTheta_0)$ as defined in \eqref{equ:surrogate_obj}. 

\begin{corollary} [RSC/RSS Conditions for the Surrogate Loss $\mathcal{L}_S(\bTheta; \widehat\bTheta_0)$] \label{col:surrogate_RSCRSS}
Given Assumptions \ref{ass:Positive_minimal_eigenvalue}-\ref{ass:reg_grad}, for any $\bTheta_1, \bTheta_2 \in \mathcal{T}$ with $\diag (\bTheta_1) = \diag (\bTheta_2)$ and rank bound 
\begin{equation*}
    \max \big\{{\rm rank}(\bTheta_1), {\rm rank}(\bTheta_2) \big\} \leq d,
\end{equation*}
there exists some constant $C$ such that for $n/m \geq C d ^ 2 p ^ 4 \log p$,
\begin{equation} \label{equ:RSC_RSS}
\begin{split}
     \frac{1}{4} \kappa_{\min} \big\|\bTheta_2 - \bTheta_1 \big\| ^ 2 _ {\mathrm{F}} \leq \mathcal{L}_S(\bTheta_2; \widehat\bTheta_0) - & \mathcal{L}_S(\bTheta_1; \widehat\bTheta_0) -  \big\langle \nabla  \mathcal{L}_S(\bTheta_1; \widehat \bTheta_0), \bTheta_2 - \bTheta_1 \big\rangle \\
     &\leq \Big(\frac{1}{2} \kappa_{\max} + \frac{1}{4} \kappa_{\min} \Big) \big\|\bTheta_2 - \bTheta_1 \big\| _ {\mathrm{F}} ^ 2,
\end{split}
\end{equation}
with a probability of at least $1 - 2 p ^ {-10}$.
\end{corollary}

\begin{proof}
    We establish the proof to Corollary \ref{col:surrogate_RSCRSS} by comparing $\mathcal{L}_1 (\bTheta_2) - \mathcal{L}_1 (\bTheta_1) - \big\langle \nabla \mathcal{L}_1 (\bTheta_1), \bTheta_2 - \bTheta_1 \big\rangle$ and its surrogate alternate $ \mathcal{L}_S(\bTheta_2; \widehat\bTheta_0) - \mathcal{L}_S(\bTheta_1; \widehat\bTheta_0) -  \big\langle \nabla  \mathcal{L}_S(\bTheta_1; \widehat \bTheta_0), \bTheta_2 - \bTheta_1 \big\rangle$. First, we apply $\frac{\partial}{\partial \Xb} \Tr (\Xb \Ab) = \Ab \trans$, which indicates that
    \begin{align*}
        \nabla  \mathcal{L}_S(\bTheta; \widehat \bTheta_0) &= \nabla \mathcal{L}_1(\bTheta) + \frac{1}{m} \sum_{j = 1} ^ m \nabla \mathcal{L}_j(\widehat\bTheta_0) - \nabla \mathcal{L}_1( \widehat \bTheta_0); \\
        \big \langle  \nabla  \mathcal{L}_S(\bTheta_1; \widehat \bTheta_0), \bTheta_2 - \bTheta_1 \big \rangle &- \big \langle \nabla \mathcal{L}_1 (\bTheta_1), \bTheta_2 - \bTheta_1 \big \rangle = \bigg \langle \frac{1}{m} \sum_{j = 1} ^ m \nabla \mathcal{L}_j( \widehat \bTheta_0) - \nabla \mathcal{L}_1( \widehat \bTheta_0), \bTheta_2 - \bTheta_1 \bigg \rangle.
    \end{align*}
  We thus have
    \begin{equation*} 
    \begin{split}
        \mathcal{L}_1 (\bTheta_2) - \mathcal{L}_1 (\bTheta_1) &- \langle \nabla \mathcal{L}_1 (\bTheta_1), \bTheta_2 - \bTheta_1 \rangle \\
        &= \mathcal{L}_S(\bTheta_2; \widehat\bTheta_0) - \mathcal{L}_S(\bTheta_1; \widehat\bTheta_0) -  \big\langle \nabla  \mathcal{L}_S(\bTheta_1; \widehat \bTheta_0), \bTheta_2 - \bTheta_1 \big\rangle,
    \end{split}
    \end{equation*}
    where the RSC/RSS conditions on $\mathcal{L}_S$ is guaranteed.
\end{proof}


\subsection{Rationality of the Initial Value}\label{sec:app:ini}

The following lemma shows the rationality of the initial value chosen in Algorithm \ref{al:Divide_conquer}. Recall that For technical reasons, we introduce the notation $\kappa ^ * = \min \{\frac{1}{8}, \frac{5}{128} \frac{\kappa_{\min} \kappa_{\max}}{\kappa_{\min} + \kappa_{\max}} \}$ and recall
\[
    K(\Ub^*; \kappa_{\min}, \kappa_{\max}) = \frac{4}{5} \kappa ^ * \sigma_d^2(\Ub^*) \min \Big\{ \frac{1}{(\kappa_{\min} + \kappa_{\max})}, 2 \Big\}.
\]

\begin{lemma}[Rationality on the Initial Value] \label{lem:initial_value} Consider $K(\Ub^*; \kappa_{\min}, \kappa_{\max})$ defined in \eqref{equ:K_U}. For any initial value $\widehat{\bTheta}$ with rank-$d$-SVD $\widehat\bTheta_0 = \Ub_0 \Vb_0 \trans$ such that
\begin{equation} \label{equ:diff_initial}
    \big \|\widehat{\bTheta} - \bTheta ^ * \big \|_{\mathrm{F}} \leq \frac{1}{5} \min \bigg\{\sigma_{d}(\bTheta^*), \sqrt{K(\Ub^*; \kappa_{\min}, \kappa_{\max}) \sigma_{d}(\bTheta^*)} \bigg\},
\end{equation}
the subspace distance $\rho^2 \big(\Zb_0, \Zb ^ * \big) \leq K(\Ub^*; \kappa_{\min}, \kappa_{\max})$.
\end{lemma}

\begin{proof}[Proof of Lemma \ref{lem:initial_value}]
    By applying the triangle inequality and Eckart-Young-Mirsky theorem to the rank-$d$ singular value decomposition of $\widehat{\bTheta}$, denoted as $\widehat\bTheta_0$, we have
    \begin{equation} \label{equ:theta_triangle}
        \big\|\widehat\bTheta_0 - \widehat \bTheta \big\|_{\mathrm{F}} \leq \big\|\widehat \bTheta - \bTheta ^ * \big\|_{\mathrm{F}}; \ \  \big\|\widehat\bTheta_0 - \bTheta ^ * \big\|_{\mathrm{F}} \leq 2 \big\|\widehat\bTheta - \bTheta ^ * \big\|_{\mathrm{F}} \leq \frac{1}{2} \sigma_d(\bTheta ^ *).
    \end{equation}
    Given \eqref{equ:theta_triangle}, Lemma 5.14 of \cite{tu2016lowrank} indicates that
    \begin{align*}
        \rho^2 \big(\Zb_0, \Zb ^ * \big) = & \rho ^ 2\bigg(\begin{bmatrix}
\Ub_0 \\
\Vb_0
\end{bmatrix}, \begin{bmatrix}
\Ub^* \\
\Vb^*
\end{bmatrix} \bigg)\leq \frac{2 \big\|\widehat\bTheta_0  - \bTheta ^ * \big\|^2_{\mathrm{F}}}{(\sqrt{2} - 1) \sigma_{d}(\bTheta ^ *)} \leq \frac{8 \big\|\widehat \bTheta - \bTheta ^ * \big\|^2_{\mathrm{F}}}{(\sqrt{2} - 1) \sigma_{d}(\bTheta ^ *)} \\
& = \frac{8 \big\|\widehat \bTheta - \bTheta ^ * \big\|^2_{\mathrm{F}}}{(\sqrt{2} - 1) \sigma^2_{d}(\Ub^*)} \leq \frac{8}{25 (\sqrt{2} - 1)} K(\Ub^*; \kappa_{\min} , \kappa_{\max}) \leq K(\Ub^*; \kappa_{\min} , \kappa_{\max}),
    \end{align*}
    which directly completes the proof.
\end{proof}

\section{Closed-Form Solutions to the Gradient and the Hessian Matrix} \label{sec:Appendix_D}

In this section, we provide additional supplementary results concerning the gradient and the Hessian matrix for the objective function $\mathcal{L}(\cdot)$ defined by \eqref{equ:conditional_dist}. 
Denote $Q_{q,\ell}(\bTheta) = 2\theta_{qq} x_q^{(\ell)}+ 2\sum_{j \neq q} \theta_{jq} x_q^{(\ell)}x_j^{(\ell)}$ and $B_{q, \ell} ^ * := \frac{-1}{1 + \exp \big(Q_{q,\ell}(\bTheta^*) \big)} $. 
Therefore, we have the gradient
\begin{equation} \label{eqn:B_star}
     \mathbf{W}_\ell(\bTheta ^ *) _{ij} = \nabla \mathcal{L}^{(\ell)}(\boldsymbol\Theta ^ *)_{ij} = \begin{cases}
  2 x_i^{(\ell)} x_j^{(\ell)} (B_{i, \ell} ^ *+B_{j, \ell} ^ *)    & \text{if $i \neq j$;}\\
     2x_i^{(\ell)} B_{i, \ell} ^ * & \text{if $i = j$.}
     \end{cases}
\end{equation}
Then, we can represent the $p^2 \times p^2$ Hessian matrix $\widehat{\Hb}(\bTheta)$ by
\begin{equation} \label{equ:hess_mat}
     \widehat{\mathbf{H}}(\bTheta)  \\
    = \frac{1}{n} \sum_{\ell = 1} ^ n \sum_{r = 1} ^ p  A_{r,\ell}(\bTheta)  \vec \big(\mathbf{Y}_{(-r)}^{(\ell)} \big)   \vec \big(\mathbf{Y}_{(-r)}^{(\ell)} \big) \trans \text{ as } A_{r, \ell} (\boldsymbol\Theta) = \frac{\exp \big(Q_{r,\ell}(\bTheta)  \big)}{\big(\exp(Q_{r,\ell}(\bTheta) ) + 1 \big) ^ 2},
\end{equation}
with matrix $\mathbf{Y}_{(-i)}^{(\ell)} \in \mathbb{R} ^ {p \times p}$ composes of $\xb ^ {(\ell)}_{(-i)} = 2\big( x_1^{(\ell)}, ... , x_{i-1}^{(\ell)}, 1, x_{i+1}^{(\ell)}, ... , x_{p}^{(\ell)} \big) \trans$ on the $i ^ {\text{th}}$ column and $\big( \xb ^ {(\ell)}_{(-i)} \big) \trans$ on the $i ^ {\text{th}}$ row, while the other entries equal to 0. 

Finally, we employ Lemma \ref{lem:H_toSigma}, as inspired by \cite{jianqing2017nuclear}, to demonstrate the rationale behind substituting the Hessian matrix $\widehat{\Hb}(\bTheta)$ with $\widehat{\bSigma}(\bTheta)$, as we proposed in the proof to Theorem \ref{thm:main}.

\begin{lemma}[Substituting Hessian matrix $\widehat{\Hb}(\bTheta)$ with $\widehat{\bSigma}(\bTheta)$] \label{lem:H_toSigma}
For any matrix $\bDelta$ and SVD $\bDelta = \Ub \Db \Vb \trans$, we have $\mathrm{vec}(\bDelta) \trans \widehat{\Hb}(\bTheta) \mathrm{vec}(\bDelta) = \mathrm{vec}(\Db) \trans \widehat{\bSigma}(\bTheta) \mathrm{vec}(\Db)$ for $\bTheta \in \mathcal{T}$ with the notation that
\begin{equation} \label{equ:sigma_sigmahat}
\begin{split}
    \widehat{\boldsymbol\Sigma}(\bTheta) &= \frac{1}{n} \sum_{\ell = 1} ^ n \bigg\{ \sum_{r = 1} ^ p \mathrm{vec} \big(\widetilde{\mathbf{Y}}_{(-r)}^{(\ell)} (\bTheta) \big) \mathrm{vec} \big(\widetilde{\mathbf{Y}}_{(-r)}^{(\ell)} (\bTheta) \big) \trans \bigg\}; \\
    {\boldsymbol\Sigma(\bTheta)} := \mathbb{E} [\widehat{\boldsymbol\Sigma}(\bTheta)] &= \mathbb{E} \bigg[ \sum_{r = 1} ^ p \mathrm{vec} \big(\widetilde{\mathbf{Y}}_{(-r)}^{(1)} (\bTheta) \big) \mathrm{vec} \big(\widetilde{\mathbf{Y}}_{(-r)}^{(1)} (\bTheta) \big) \trans \bigg],
\end{split}
\end{equation}
where $\widetilde{\mathbf{Y}}_{(-r)}^{(\ell)}(\boldsymbol\Theta) = \widetilde{\mathbf{Y}}_{(-r)}^{(\ell)}(\boldsymbol\Theta, \bDelta) = \sqrt{A_{r,\ell}(\boldsymbol\Theta)} \mathbf{U} \trans \mathbf{Y}_{(-r)}^{(\ell)} \mathbf{V}$. 
\end{lemma}

\begin{proof}[Proof of Lemma \ref{lem:H_toSigma}]
We start with the closed form of the Hessian matrix $\widehat{\mathbf{H}}(\boldsymbol\Theta)$ as given in \eqref{equ:hess_mat}, and we have that
\begin{align*}
    &\vec (\boldsymbol\Delta) \trans   \widehat{\mathbf{H}}(\boldsymbol\Theta)   \vec (\boldsymbol\Delta) 
    = \frac{1}{n} \sum_{\ell = 1} ^ n \sum_{r = 1} ^ p \vec (\boldsymbol\Delta) \trans   \vec \big(\mathbf{Y}_{(-r)}^{(\ell)} \big)   \    A_{r,\ell} (\boldsymbol\Theta)  \  \vec \big(\mathbf{Y}_{(-r)}^{(\ell)} \big) \trans  \vec (\boldsymbol\Delta).
\end{align*}
Applying the SVD of $\bDelta = \Ub \Db \Vb \trans$ yields that
\begin{align*}
    & \frac{1}{n} \sum_{\ell = 1} ^ n \sum_{r = 1} ^ p \vec(\boldsymbol\Delta) \trans   \vec\big(\mathbf{Y}_{(-r)}^{(\ell)} \big)   \    A_{r,\ell} (\boldsymbol\Theta)  \  \vec\big(\mathbf{Y}_{(-r)}^{(\ell)} \big) \trans  \vec(\boldsymbol\Delta) \\
    &= \frac{1}{n} \sum_{\ell = 1} ^ n \sum_{r = 1} ^ p \tr \Big(\sqrt{A_{r,\ell} (\boldsymbol\Theta)} \big( \mathbf{Y}_{(-r)}^{(\ell)} \big) \trans  \bDelta \Big) ^ 2 = \frac{1}{n} \sum_{\ell = 1} ^ n \sum_{r = 1} ^ p \tr \Big(\sqrt{A_{r,\ell}(\boldsymbol\Theta)}  \big( \mathbf{Y}_{(-r)}^{(\ell)} \big) \trans \Ub \Db \Vb \trans \Big) ^ 2 \\
    &= \frac{1}{n} \sum_{\ell = 1} ^ n \sum_{r = 1} ^ p \tr \Big(\sqrt{A_{r,\ell}(\boldsymbol\Theta)} \Vb \trans  \big( \mathbf{Y}_{(-r)}^{(\ell)} \big) \trans \Ub \Db \Big) ^ 2,
\end{align*}
since $A_{r,\ell}(\bTheta)$ is a scalar, and $\text{tr}(\Ab\Bb)$ = $\text{tr}(\Bb \Ab)$. Finally, we conclude that
\begin{align*}
   \vec(\boldsymbol\Delta) \trans   \widehat{\mathbf{H}}(\boldsymbol\Theta)  \vec(\boldsymbol\Delta) & = \frac{1}{n} \sum_{\ell = 1} ^ n \sum_{r = 1} ^ p \tr \Big( \big( \widetilde{\mathbf{Y}}_{(-r)}^{(\ell)} (\bTheta) \big) \trans \Db \Big) ^ 2 \\
    &= \vec (\mathbf{D}) \trans   \bigg\{ \frac{1}{n} \sum_{\ell = 1} ^ n \sum_{r = 1} ^ p \vec\big(\widetilde{\mathbf{Y}}_{(-r)}^{(\ell)} (\bTheta) \big) \vec \big(\widetilde{\mathbf{Y}}_{(-r)}^{(\ell)} (\bTheta) \big) \trans \bigg\}   \vec(\mathbf{D}),
\end{align*}
as $\tr (\Ab \trans \Bb) = \vec(\Ab) \trans \vec(\Bb)$.
\end{proof}

\section{Proof for the Rationality of Assumption} \label{app:E}

In this section, we prove Proposition \ref{prop:rationality_assumption} as one example where Assumption \ref{ass:reg_grad} has been rational.

\begin{proof}[Proof of Proposition \ref{prop:rationality_assumption}]
First, we compute each entry of the matrix $\bXi_\ell(\bTheta ^ *) = \Wb_{\ell}(\bTheta ^ *) \Wb_{\ell}(\bTheta ^ *) \trans$, which is given by the following expression 
\begin{equation} \label{equ:correlation_ising}
    \big [\bXi_\ell(\bTheta ^ *)\big]_{ij} =  \sum_{k \neq i,j} 4x_{i} ^ {(\ell)} x_j ^ {(\ell)} \big(B_{i,\ell} ^ * + B_{k, \ell} ^ *\big) \big(B_{j,\ell} ^ * + B_{k, \ell} ^ *  \big) + 4x_i (B_{i,\ell} ^ {*2} + B_{i,\ell} ^ * B_{j, \ell} ^ *)+4x_j (B_{j,\ell} ^ {*2} + B_{i,\ell} ^ * B_{j, \ell} ^ *)  
\end{equation}
if $i\neq j$ and $ \big [\bXi_\ell(\bTheta ^ *)\big]_{ii} = 4\sum_{k \neq i} \big(B_{i,\ell} ^ * + B_{k, \ell} ^ *\big) ^2 + 4 B_{i,\ell} ^{*2} $.
In this context, the correlation between the values on vertices in a graphical model is depicted by $\mathbb{E}[X_i X_j]$, where $X_i \in \{-1, +1\}$ for any arbitrary $i \in [p]$ representing the occurrence pattern for the $i$th EHR code in the graph $\mathcal{G}(\mathcal{V}, \mathcal{E})$ with $p$ codes of interest in total.

Our proof  involves finding a {lower} bound for the {diagonal} entries of $\mathbb{E} \big[\bXi_\ell(\bTheta ^*)\big]$ and an \textit{upper} bound for the {off-diagonal} entries of $\mathbb{E}\big[\bXi_\ell(\bTheta ^*)\big]$ using the closed-form expression presented in \eqref{equ:correlation_ising}, where $\big(X_i ^{(\ell)}\big)^2 = 1$ for any $i \in [p]$. We then show that the matrix $\mathbb{E}\big[\bXi_\ell(\bTheta ^*) \big]$ is diagonally dominant under mild regularization conditions imposed on the Ising model structure. Finally, we show Proposition \ref{prop:rationality_assumption} by applying the Gershgorin circle theorem to obtain an upper bound on the operator norm of $\mathbb{E}\big[\bXi_\ell(\bTheta ^ *)\big]$. By the binary Ising observation, we have
\begin{equation*}
  \big|Q_{i,\ell}(\bTheta^*)\big| = \Big|2\theta^*_{ii} x_i^{(\ell)}+ 2\sum_{j \neq i} \theta^*_{ji} x_i^{(\ell)}x_j^{(\ell)} \Big|\leq 2 \sum_{j = 1; j \neq i}^p \big|\theta^*_{ij}\big| \text{\ \ for any \ } i \in [p],\  \ell \in [n].
\end{equation*}
Since the maximum degree of $\mathcal{G}$ is $\tau$, there exists a lower bound for $\big\{\bXi_\ell(\bTheta ^ *)\big\}_{ii}$ such that
\begin{align}
  \nonumber  \big[\bXi_\ell(\bTheta ^ *)\big]_{ii} &= 4\sum_{k \neq i} \big(B_{i,\ell} ^ * + B_{k, \ell} ^ *\big) ^2 + 4 B_{i,\ell} ^{*2} \\
    &\geq\frac{16p}{\big(1 + e ^ {2 \sum_{j = 1; j \neq i} ^ p |\theta_{ij} ^ *|} \big)^2}  = \frac{16 p}{\big(1 + e ^ {2 \|\bTheta ^ *\|_{1, \infty}} \big)^2}. \label{eq:psi-low}
\end{align}

We need to study the correlation for the off-diagonal terms. Combining Theorem 1.1  in \cite{ding2023new} with Lemma C.2 and C.3 in \cite{naykov2018isingtest}, we have
\begin{equation} \label{equ:XsXt_noconnect}
    \mathbb{E}{[X_s X_t]} \leq \frac{\tau \tanh{(\theta_{\max})} ^ 2}{1 - (\tau - 1) \tanh{(\theta_{\max})}} \ \ \text{ when } \ \ (s,t) \notin \mathcal{E},
\end{equation}
where $\tau$ denote the maximum degree of the graph $\mathcal{G}$. And by Lemma 16 of \cite{nikola2021treeising}, we have 
\begin{equation} \label{equ:XsXt_connect}
     \mathbb{E}{[X_s X_t]} \leq \frac{\tau \tanh{(\theta_{\max})} ^ 2}{1 - (\tau - 1) \tanh{(\theta_{\max})}} + \tanh{(\theta_{\max})} \text{ for } (s,t) \in \mathcal{E}.
\end{equation}
By \eqref{equ:correlation_ising}, we have for any $i$,
\begin{align*}
  \sum_{j = 1; j \neq i} ^ p \mathbb{E}\Big[\big\{\bXi_\ell(\bTheta ^ *)\big\}_{ij} \Big] &\le \sum_{j = 1; j \neq i} ^ p 16p\mathbb{P}(X_iX_j=1) = \sum_{j = 1; j \neq i} ^ p 8p(\mathbb{E}[X_iX_j] +1) \\
  &\le 8p^2 \frac{\tau \tanh{(\theta_{\max})} ^ 2}{1 - (\tau - 1) \tanh{(\theta_{\max})}} + 8p^2\tanh{(\theta_{\max})}\\
  & \le 32C^2p\tau \leq \frac{1}{2} \mathbb{E} \Big[ \big\{ \bXi_{\ell} (\bTheta ^ *) \big\}_{ii} \Big],
\end{align*}
where the third inequality is due to  $\|\bTheta ^*\|_{1, \infty} \leq C/p$, $\theta_{\max} \leq \|\bTheta ^ *\|_{1, \infty}$ and $\tanh(x)\le x$ for $x\ge 0$, and the last inequality is due to \eqref{eq:psi-low}. By applying Gerschgorin's circle theorem \citep{gerschgorin1931disc}, we show that the matrix $\mathbb{E} \big[ \bXi_{\ell} (\bTheta ^ *) \big]$ has an operator norm of $O(p)$.
\end{proof}

\section{Supplementary Discussion on Experiment Configuration}
\label{app:hyperparam}


\subsection{Hyperparameters for DANIEL in Simulation Study} \label{app:hyperparam_sim}


\noindent\textbf{DANIEL:} The stopping criterion for DANIEL is based on the Frobenius norm of the difference between estimators across successive iterations. Specifically, we compute 
\begin{equation*}
    \delta_{\gamma} = \big\| \widehat\bTheta_{\gamma} - \widehat\bTheta_{\gamma-1} \big \|_{\rm F}, \text{ where } \widehat\bTheta_{\gamma} = \Ub_{\gamma}\Vb_{\gamma}^\top \text{ for the } \gamma\text{th iteration},
\end{equation*}
and terminate the iteration when either $\delta_{\gamma} < 10^{-5}$ or the maximum number of iterations, $\Gamma = 50$, is reached. The step size $\eta$ is selected via grid search from the range $[0.1, 0.3]$, depending on the $(n, p)$ setup. The algorithm is initialized using a rank-$d$ SVD of a preliminary solution to convex optimization problem~(\ref{equ:question}) with the local data at the hub institution $\mathcal{S}_1$, obtained by applying 5 gradient updates starting from a $p \times p$ zero matrix. The time required in initialization is included in the reported computation time for DANIEL.

\subsection{Detailed Explanation for Baseline Methods in Simulation Study} \label{app:hyperparam_base}


\noindent\textbf{Baseline Methods:} For all four baseline methods, the stopping criterion and step size $\eta$ are aligned with those of DANIEL. Additional method-specific hyper-parameters are set as follows: for SV-Soft and SV-Hard, the singular value threshold $\tau$ is fixed at $10^{-3}$; for SV-Topd, the number of retained singular values is set to match the true rank $d$ of the underlying matrix $\bTheta^*$.

\noindent\textbf{SV-Soft}:  After each gradient updates $\bTheta_{\gamma}^{(0)} = \bTheta_{\gamma -1} - \eta \nabla_{\bTheta} \mathcal{L}(\bTheta) \mid_{\bTheta = \bTheta_{\gamma-1}}$, denote the SVD of $\bTheta_{\gamma}^{(0)}=\Ub_{\gamma}^{(0)} \boldsymbol{\Sigma}_{\gamma}^{(0)} \big( \Vb_{\gamma}^{(0)} \big)\trans$, where $\boldsymbol{\Sigma}_{\gamma}^{(0)} = \diag(\sigma_{\gamma,1},\ldots,\sigma_{\gamma,p})$. A \textit{soft} thresholding is then applied in each iteration as
\begin{equation*}
    \boldsymbol{\Sigma}_{\gamma}^{\text{SV-Soft}} = \diag \big( (\sigma_{\gamma,1}-\tau)_+, \ldots, (\sigma_{\gamma,p}-\tau)_+ \big), \text{ and } \bTheta_{\gamma} = \Ub_{\gamma}^{(0)} \boldsymbol{\Sigma}_{\gamma}^{\text{SV-Soft}} \big( \Vb_{\gamma}^{(0)} \big)\trans.
\end{equation*}

\noindent\textbf{SV-Hard}: For \textit{hard} thresholding, we set $\tau$ as a hard threshold to the eigenvalues. In each iteration, we have
\begin{equation*}
    \boldsymbol{\Sigma}_{\gamma}^{\text{SV-Hard}} = \diag \big( \sigma_{\gamma,1} \mathbf{1}_{\sigma_{\gamma,1} > \tau}, \ldots, \sigma_{\gamma,p} \mathbf{1}_{\sigma_{\gamma,p} > \tau} \big), \text{ and } \bTheta_{\gamma} = \Ub_{\gamma}^{(0)} \boldsymbol{\Sigma}_{\gamma}^{\text{SV-Hard}} \big( \Vb_{\gamma}^{(0)} \big)\trans.
\end{equation*}

\noindent\textbf{SV-Topd}: In this scenario, we set 
\begin{equation*}
    \boldsymbol{\Sigma}_{\gamma}^{\text{SV-Topd}} = \diag \big( \sigma_{\gamma,1} , \ldots, \sigma_{\gamma,r} , 0, \ldots, 0 \big), \text{ and } \bTheta_{\gamma} = \Ub_{\gamma}^{(0)} \boldsymbol{\Sigma}_{\gamma}^{\text{SV-Topd}} \big( \Vb_{\gamma}^{(0)} \big)\trans.
\end{equation*}

\noindent\textbf{PSD-Cvx}: In this scenario, we consider the positive semi-definite (PSD) constraints \citep{boyd2004convex}, which yields the iterative update
\begin{equation*}
    \boldsymbol{\Sigma}_{\gamma}^{\text{PSD-Cvx}} = \diag \big( (\sigma_{\gamma,1})_+ , \ldots, (\sigma_{\gamma,p})_+ \big), \text{ and } \bTheta_{\gamma} = \Ub_{\gamma}^{(0)} \boldsymbol{\Sigma}_{\gamma}^{\text{PSD-Cvx}} \big( \Vb_{\gamma}^{(0)} \big)\trans.
\end{equation*}

\noindent As for the distributed settings of these baseline methods, we propose the loss 
\begin{equation*}
    \widetilde{\mathcal{L}}^{(b)} (\bTheta; \widehat\bTheta_0) = \mathcal{L}_1(\bTheta)  + \big\langle \nabla \mathcal{L} (\widehat\bTheta_0) - \nabla \mathcal{L}_1 (\widehat\bTheta_0), \bTheta  \big \rangle.
\end{equation*}
For each gradient, we update $\bTheta_\gamma^{(0)} = \bTheta_{\gamma -1} - \eta \nabla_{\bTheta} \widetilde{\mathcal{L}}^{(b)}(\bTheta) \mid_{\bTheta = \bTheta_{\gamma-1}}$, followed by a matrix update after SVD of $\bTheta_\gamma^{(0)}$.

\subsection{Hyperparameters for DANIEL's Implementation to EHR Study} \label{sec:Exp_config}

\noindent\textbf{DANIEL:} For the real-world EHR application, we set the maximum number of iterations $\Gamma = 200$, the step size $\eta = 0.01$, and the latent rank $d = 200$. These hyperparameters are selected to maximize the performance in identifying known relationship pairs. When applying DANIEL to data from both UPMC and MGB systems, we set the number of institutions to $m = 2$ to reflect the actual distributed data layout.

\noindent \textbf{Baseline Methods:} For PSD-Cvx, SV-Hard, and SV-Topd, the number of iterations and step size are kept the same as those used in DANIEL. The singular value threshold $\tau$ for SV-Hard is set to $10^{-4}$, while SV-Topd retains the top $d = 200$ singular values to maintain consistency with DANIEL. For BERT-based algorithms, the embedding dimension is set to $768$, following the original BERT specification \citep{devlin2018bert}.

\noindent \textbf{Negative Samples for Known-relationship Pairs:} For each relationship type, we first sample the same number of \textit{positive pairs} ($\text{feature}_j$, $\text{feature}_k$) with known relationship and random \textit{negative pairs} ($\text{feature}_r$, $\text{feature}_s$). We then compute the AUC by comparing $\widehat\theta_{jk}$ and $\widehat\theta_{rs}$ for DANIEL-derived parameter matrix $\widehat{\bTheta}$ and $\widehat{\bTheta}_{\mathrm{base}}$ from SVD-based baselines. For the BERT-based baselines, the relationship between clinical features is quantified using cosine similarity between their corresponding embedding vectors.

\begin{table}[]
    \centering
    \caption{Number of pairs with each category of known relationship.}
    \label{tab:app_n_pairs}
    \begin{tabular}{c|cc|ccc}
        \toprule
        \multirow{2}*{Relation types} & \multicolumn{2}{c|}{\textbf{Similar} pairs} & \multicolumn{3}{c}{\textbf{Related} pairs} \\
         & {\tt PheCode} Hierachy & {\tt RxNorm} Hierachy & related to & mapped to & classifies \\
        \hline
        Number of Pairs & 4352 & 7988 & 1755 & 1282 & 4667 \\
        \bottomrule
    \end{tabular}
\end{table}

\noindent \textbf{Textual Descriptions of EHR Concepts:} For each clinical code, we obtain the corresponding textual description from publicly available sources. For example, {\tt PheCode:199} is described as ``Neoplasm of uncertain behavior" and {\tt CCS:227} is described as ``Consultation, evaluation, and preventative care".

\noindent \textbf{Cross-validation:} 
To prevent data leakage, patients are uniformly and randomly partitioned into 10 folds. For each fold, DANIEL is trained on the remaining 9 folds to obtain the estimator $\widehat\bTheta$, and the resulting model is then used to predict the labels in the held-out fold. We repeat the process across all 10 folds to complete cross-validation for DANIEL-derived $\widehat\bTheta$, as well as $\widehat{\bTheta}_{\mathrm{base}}$ from SVD-based baselines, in an \textit{unsupervised} configuration. For the BERT-based baselines, they only produce feature-level embeddings and do not directly yield label predictions. To enable a fair comparison, we first construct patient-level embeddings for each patient in each time period by computing $\widehat\yb = \widehat\Ub_{-i}\trans \xb_{-i}$, where $\widehat\Ub_{-i}$ is the generated embeddings for all features excluding the target code {\tt PheCode:290.1}, and $\xb_{-i}$ is the corresponding feature-occurrence vector for that patient with {\tt PheCode:290.1} excluded. We then train a decision tree classifier \citep{Hastie2009Tree} on the embeddings based on the $\widehat\yb$ from the 9 training folds and evaluate the performance on the held-out fold. We repeat this \textit{supervised} procedure across all 10 folds, in parallel with the unsupervised setting, for DANIEL and all tested baselines.

\noindent \textbf{Patient Phenotyping:} The UPMC cohort consists of $28,955$ patients and $505,626$ total records, from which $9,294$ pre-diagnosis records and $9,294$ post-diagnosis records are selected. The MGB cohort consists of $29,293$ patients and $101,093$ records, with $21,238$ pre-diagnosis records and $21,238$ post-diagnosis records selected. We train the decision tree classifier used for supervised label prediction in patient phenotyping using the \texttt{rpart} package in {\tt R} \citep{Therneau2025rpart}.

\noindent \textbf{Summary of Input and Output for Each Method}:
\begin{table}[H]
    \centering
    \caption{Summary of the input and output for each method.}
    \label{tab:app.sum.in_out}
    \begin{tabular}{p{4.78cm} p{8cm} c}
        \toprule
        Methods & Input & Output \\
        \hline
        DANIEL, PSD-Cvx, SV-Hard, SV-Topd & $p$-dimensional vectors $\xb^{(\ell)}$, where $\xb^{(\ell)}_j\in\{\pm 1\}$ indicates the presence or absence of feature & $\widehat\bTheta$, $\widehat{\Ub}$  \\
        \hline
        BERT, BioBERT, SapBERT & textual description of EHR concepts & $\widehat{\Ub}$ \\
        \bottomrule
    \end{tabular}
\end{table}

\newpage

\section{Supplemental Results}
\label{app:simu_result}

\begin{table}[H]
    \small
    \resizebox{\textwidth}{!}{
        \setlength{\tabcolsep}{2pt}
        \begin{tabular}{r|r|c|c|c|c|c|c|c}
            \toprule
            n & method & 0 & 0.1 & 0.2 & 0.3 & 0.4 & 0.5 & 0.6\\
            \hline
            \multirow{5}*{1000} & \textbf{DANIEL} & \textbf{0.83±0.18} & \textbf{0.89±0.17} & \textbf{0.96±0.15} & \textbf{1.08±0.13} & \textbf{1.32±0.19} & \textbf{2.02±0.61} & \textbf{4.09±2.09}\\
             & SV-Soft & 1.10±0.23 & 1.68±0.14 & 2.36±0.12 & 3.28±0.14 & 4.50±0.20 & 5.97±0.31 & 7.42±0.44\\
             & SV-Hard & 1.54±0.23 & 2.13±0.16 & 2.99±0.13 & 4.18±0.13 & 5.82±0.21 & 7.78±0.21 & 8.00±0.00\\
             & SV-Topd & 1.04±0.32 & 1.40±0.25 & 1.89±0.21 & 2.59±0.25 & 3.59±0.34 & 4.99±0.57 & 6.52±0.82\\
             & PSD-Cvx & 1.12±0.22 & 1.72±0.14 & 2.42±0.11 & 3.34±0.13 & 4.55±0.19 & 6.04±0.38 & 7.49±0.43\\
            \hline
            \multirow{5}*{5000} & \textbf{DANIEL} & \textbf{0.53±0.23} & \textbf{0.56±0.22} & \textbf{0.61±0.20} & \textbf{0.70±0.18} & \textbf{0.87±0.15} & \textbf{1.38±0.30} & \textbf{3.79±1.96}\\
             & SV-Soft & 0.79±0.31 & 1.13±0.22 & 1.66±0.16 & 2.60±0.12 & 3.90±0.15 & 5.85±0.27 & 7.95±0.12\\
             & SV-Hard & 1.18±0.34 & 1.45±0.27 & 2.08±0.19 & 3.24±0.15 & 4.86±0.17 & 7.46±0.29 & 8.00±0.00\\
             & SV-Topd & 0.85±0.39 & 1.02±0.33 & 1.36±0.28 & 1.99±0.22 & 2.99±0.26 & 4.71±0.50 & 7.10±0.74\\
             & PSD-Cvx & 0.80±0.30 & 1.17±0.21 & 1.72±0.15 & 2.67±0.11 & 3.95±0.15 & 5.97±0.26 & 7.98±0.07\\
            \hline
            \multirow{5}*{10000} & \textbf{DANIEL} & \textbf{0.40±0.21} & \textbf{0.45±0.19} & \textbf{0.49±0.18} & \textbf{0.58±0.16} & \textbf{0.76±0.13} & \textbf{1.25±0.27} & \textbf{3.96±1.88}\\
             & SV-Soft & 0.64±0.28 & 1.16±0.17 & 1.58±0.13 & 2.51±0.10 & 3.98±0.16 & 6.36±0.31 & 8.00±0.00\\
             & SV-Hard & 1.05±0.32 & 1.45±0.24 & 1.91±0.18 & 3.02±0.13 & 4.79±0.19 & 7.73±0.25 & 8.00±0.00\\
             & SV-Topd & 0.71±0.38 & 0.99±0.29 & 1.24±0.26 & 1.87±0.21 & 2.91±0.26 & 4.88±0.47 & 7.65±0.45\\
             & PSD-Cvx & 0.66±0.27 & 1.21±0.15 & 1.64±0.12 & 2.57±0.10 & 4.05±0.16 & 6.41±0.28 & 8.00±0.00\\
            \hline
            \multirow{5}*{20000} & \textbf{DANIEL} & \textbf{0.28±0.18} & \textbf{0.34±0.16} & \textbf{0.39±0.15} & \textbf{0.48±0.13} & \textbf{0.65±0.11} & \textbf{1.11±0.16} & \textbf{4.03±1.59}\\
             & SV-Soft & 0.50±0.25 & 1.10±0.13 & 1.58±0.10 & 2.6±0.10 & 4.33±0.17 & 7.32±0.35 & 8.00±0.00\\
             & SV-Hard & 0.92±0.31 & 1.28±0.22 & 1.79±0.16 & 2.94±0.13 & 4.92±0.19 & 8.00±0.02 & 8.00±0.00\\
             & SV-Topd & 0.58±0.37 & 0.89±0.25 & 1.15±0.21 & 1.75±0.16 & 2.92±0.26 & 5.31±0.51 & 7.99±0.07\\
             & PSD-Cvx & 0.52±0.24 & 1.16±0.12 & 1.64±0.09 & 2.67±0.11 & 4.38±0.17 & 7.40±0.33 & 8.00±0.00\\
            \bottomrule
        \end{tabular}
    }
    
    \caption{Supplementary simulation results: Estimation error (in the form of mean$\pm$sd) of $p=50$. Each column represents the number of split batches as $m=n^x$ and $x$ is shown in the first row.}
    \label{fig:app_simu_tab_err}
\end{table}

\begin{table}[H]
    \centering
    \resizebox{\textwidth}{!}{
        \setlength{\tabcolsep}{1pt}
        \begin{tabular}{l|c|c|c|c|c|c|c|c|c|c}
            \toprule
            method & 20 & 40 & 60 & 80 & 100 & 120 & 140 & 160 & 180 & 200\\
            \hline
            \textbf{DANIEL} & \textbf{0.12±0.12} & \textbf{0.25±0.13} & \textbf{0.39±0.20} & \textbf{0.50±0.27} & \textbf{0.58±0.25} & \textbf{0.81±0.56} & \textbf{1.02±0.75} & \textbf{1.28±0.98} & \textbf{1.62±1.03} & \textbf{1.87±1.09}\\
            SV-Soft & 0.52±0.36 & 0.97±0.62 & 1.41±0.90 & 1.81±1.18 & 2.19±1.48 & 2.61±1.76 & 3.00±2.04 & 3.38±2.30 & 3.77±2.53 & 4.15±2.81\\
            SV-Hard & 0.63±0.32 & 1.24±0.56 & 1.82±0.83 & 2.35±1.14 & 2.86±1.45 & 3.40±1.77 & 3.95±2.10 & 4.48±2.45 & 5.05±2.77 & 5.23±2.78\\
            SV-Topd & 0.38±0.26 & 0.79±0.42 & 1.17±0.55 & 1.47±0.71 & 1.72±0.86 & 2.07±0.99 & 2.33±1.21 & 2.63±1.39 & 2.95±1.57 & 3.25±1.77\\
            PSD-Cvx & 0.54±0.38 & 1.00±0.64 & 1.44±0.92 & 1.86±1.22 & 2.24±1.52 & 2.66±1.78 & 3.06±2.07 & 3.45±2.35 & 3.84±2.57 & 4.22±2.83\\
            \bottomrule
        \end{tabular}
    }
    \caption{Supplementary simulation results: Estimation error (in the form of mean$\pm$sd) of sample size $n=10,000$ and different dimensions $p$ (as shown in the first row).}
    \label{tab:app_simu_EvP}
\end{table}

\begin{figure}[H]
    \centering
    \includegraphics[width=0.7\linewidth]{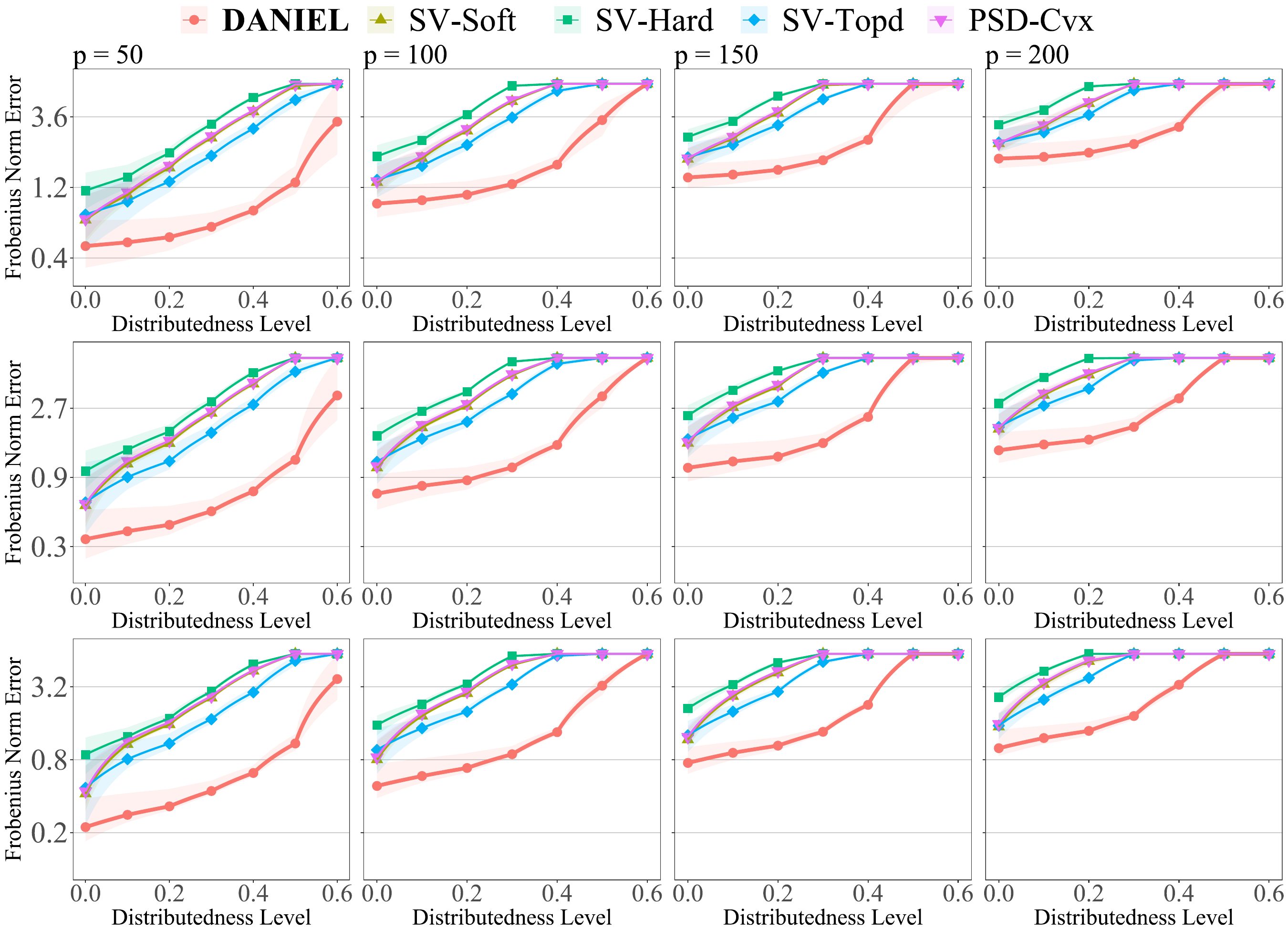}
    \includegraphics[width=0.7\linewidth]{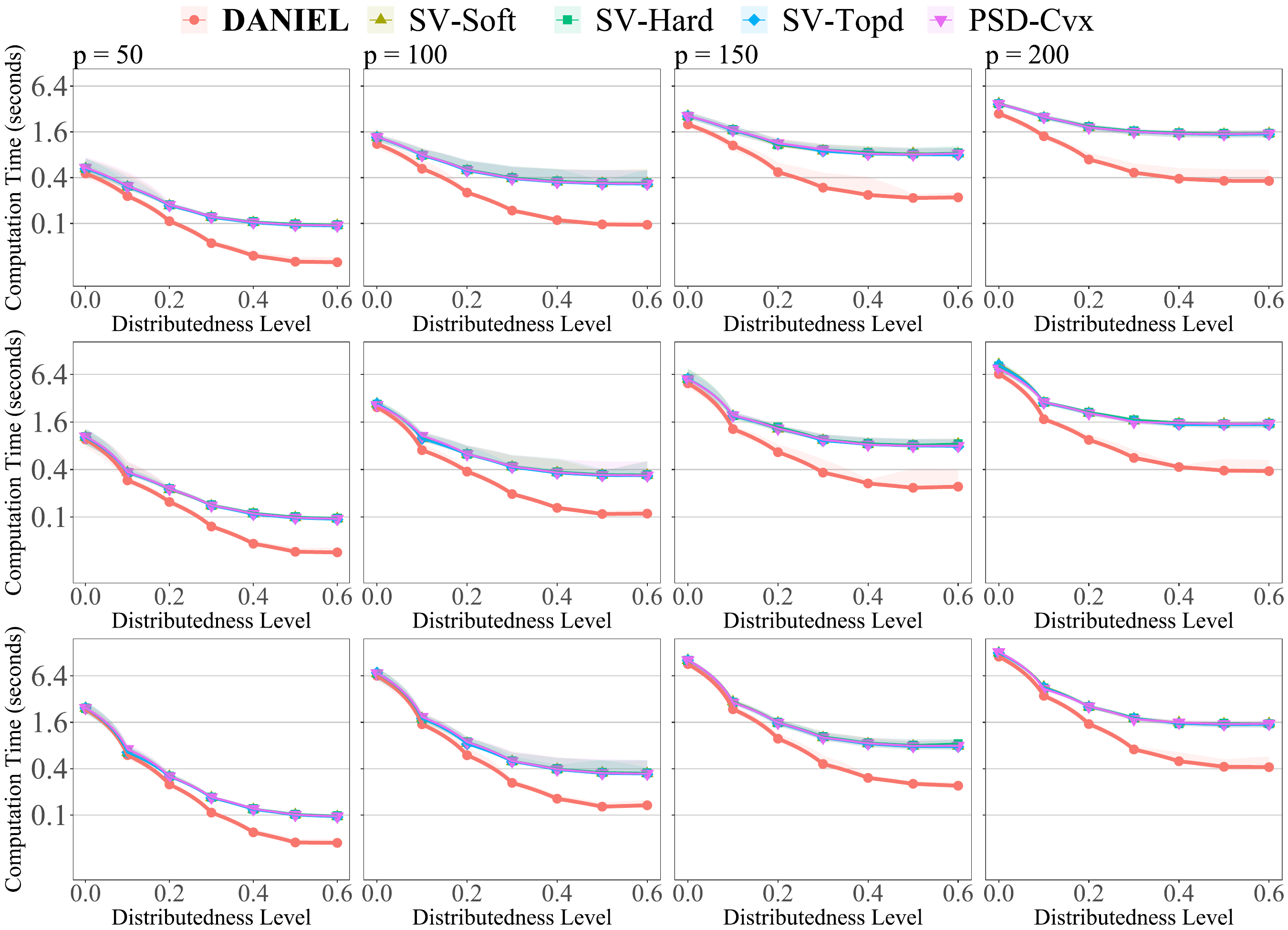}
    \caption{The trajectories of the $\|\cdot\|_{\text{F}}$-error and computation time across different feature dimensions $p$, with the distributedness level $x$ (where $m = \lfloor n^{x} \rfloor$) varying from 0 to 0.6. The total sample size is set at $n = 1,000$. The trajectories of DANIEL are shown in red, demonstrating superior performance and efficiency over baseline methods for distributedness levels $x < 0.5$ (i.e., $m = o(\sqrt{n})$). Vertically, lower $\|\cdot\|_{\text{F}}$-error and shorter computation time indicate better performance for any fixed $x$; horizontally, a flatter error trajectory as $x$ increases is preferred, as this indicates that the distributed estimators remain valid as compared with their centralized counterparts.}
    \label{fig:app_simu}
\end{figure}

\begin{table}[H]
    \small
    \resizebox{\textwidth}{!}{
        \setlength{\tabcolsep}{2pt}
        \begin{tabular}{r|r|c|c|c|c|c|c|c}
            \toprule
            n & method & 0 & 0.1 & 0.2 & 0.3 & 0.4 & 0.5 & 0.6\\
            \hline
            \multirow{5}*{1000} & \textbf{DANIEL} & \textbf{0.12±0.04} & \textbf{0.07±0.03} & \textbf{0.05±0.03} & \textbf{0.04±0.03} & \textbf{0.03±0.02} & \textbf{0.03±0.01} & \textbf{0.04±0.04}\\
            & SV-Soft & 0.20±0.07 & 0.15±0.05 & 0.13±0.05 & 0.12±0.05 & 0.12±0.06 & 0.11±0.05 & 0.11±0.06\\
            & SV-Hard & 0.19±0.06 & 0.15±0.05 & 0.13±0.06 & 0.12±0.05 & 0.11±0.05 & 0.11±0.04 & 0.11±0.05\\
            & SV-Topd & 0.19±0.06 & 0.15±0.06 & 0.13±0.06 & 0.12±0.05 & 0.11±0.04 & 0.11±0.06 & 0.10±0.05\\
            & PSD-Cvx & 0.19±0.05 & 0.16±0.06 & 0.13±0.07 & 0.13±0.07 & 0.12±0.06 & 0.12±0.06 & 0.11±0.05\\
            \hline
            \multirow{5}*{5000} & \textbf{DANIEL} & \textbf{0.49±0.10} & \textbf{0.25±0.06} & \textbf{0.11±0.04} & \textbf{0.06±0.03} & \textbf{0.04±0.02} & \textbf{0.04±0.03} & \textbf{0.04±0.02}\\
            & SV-Soft & 0.57±0.11 & 0.33±0.07 & 0.18±0.05 & 0.13±0.04 & 0.11±0.04 & 0.10±0.04 & 0.10±0.04\\
            & SV-Hard & 0.56±0.10 & 0.32±0.07 & 0.19±0.05 & 0.14±0.05 & 0.12±0.04 & 0.11±0.04 & 0.11±0.06\\
            & SV-Topd & 0.56±0.12 & 0.33±0.08 & 0.18±0.04 & 0.13±0.05 & 0.12±0.05 & 0.10±0.03 & 0.10±0.04\\
            & PSD-Cvx & 0.57±0.11 & 0.34±0.09 & 0.19±0.05 & 0.14±0.06 & 0.11±0.04 & 0.10±0.04 & 0.11±0.05\\
            \hline
            \multirow{5}*{10000} & \textbf{DANIEL} & \textbf{1.00±0.21} & \textbf{0.31±0.06} & \textbf{0.16±0.04} & \textbf{0.08±0.03} & \textbf{0.05±0.02} & \textbf{0.04±0.02} & \textbf{0.04±0.03}\\
            & SV-Soft & 1.07±0.21 & 0.39±0.07 & 0.24±0.06 & 0.15±0.05 & 0.12±0.04 & 0.11±0.04 & 0.11±0.04\\
            & SV-Hard & 1.09±0.21 & 0.39±0.07 & 0.24±0.05 & 0.15±0.04 & 0.12±0.04 & 0.11±0.05 & 0.11±0.04\\
            & SV-Topd & 1.08±0.23 & 0.38±0.07 & 0.24±0.06 & 0.15±0.05 & 0.12±0.04 & 0.10±0.02 & 0.10±0.04\\
            & PSD-Cvx & 1.09±0.22 & 0.39±0.07 & 0.24±0.05 & 0.16±0.06 & 0.12±0.04 & 0.11±0.04 & 0.10±0.04\\
            \hline
            \multirow{5}*{20000} & \textbf{DANIEL} & \textbf{2.42±0.35} & \textbf{0.63±0.09} & \textbf{0.26±0.04} & \textbf{0.11±0.02} & \textbf{0.06±0.01} & \textbf{0.05±0.00} & \textbf{0.05±0.02}\\
            & SV-Soft & 2.48±0.35 & 0.68±0.09 & 0.34±0.06 & 0.18±0.04 & 0.13±0.03 & 0.11±0.04 & 0.10±0.04\\
            & SV-Hard & 2.50±0.42 & 0.69±0.08 & 0.34±0.06 & 0.18±0.04 & 0.13±0.03 & 0.11±0.04 & 0.10±0.04\\
            & SV-Topd & 2.51±0.39 & 0.68±0.08 & 0.33±0.05 & 0.18±0.04 & 0.13±0.04 & 0.10±0.03 & 0.10±0.04\\
            & PSD-Cvx & 2.51±0.32 & 0.74±0.11 & 0.34±0.06 & 0.18±0.04 & 0.13±0.04 & 0.11±0.04 & 0.10±0.03\\
            \bottomrule
        \end{tabular}
    }

    \caption{Supplementary simulation results: Computation time (in the form of mean$\pm$sd) of $p$ = 50. Each column represents the number of split batches as $m=n^x$ and $x$ is shown in the first row.}
    \label{fig:app_simu_tab_time}
\end{table}



\bibliography{references}

\end{document}